\documentclass[11pt]{amsart}
\usepackage{amsthm}

 \usepackage{amsmath,amssymb}
 \usepackage{algorithm}
 \usepackage[noend]{algpseudocode}
\usepackage{hyperref}
\usepackage[all]{xy}
\usepackage{amssymb}
\usepackage{enumerate}
\usepackage{mathrsfs}
\usepackage{graphicx}
\usepackage{caption}
\usepackage{subcaption}
\usepackage{float}
\usepackage{epigraph}

\usepackage{booktabs}

\numberwithin{equation}{section}

\newtheorem{thm}{Theorem}[section]

\newtheorem{prop}[thm]{Proposition}
\newtheorem{rem}{Remark}[section]

\newcommand{\eq}[1]{(\ref{#1})}

\newcommand{\mbr}{\medbreak}
\newcommand{\sbr}{\smallbreak}

\renewcommand{\Re}{\operatorname{\rm Re}}
\renewcommand{\Im}{\operatorname{\rm Im}}

\newcommand{\beqast}{\begin{eqnarray*}}
\newcommand{\eqast}{\end{eqnarray*}}
\newcommand{\beqa}{\begin{eqnarray}}
\newcommand{\eqa}{\end{eqnarray}}

\newcommand{\bbe}{\begin{equation}}
\newcommand{\ee}{\end{equation}}

\renewcommand{\Re}{\operatorname{\rm Re}}
\renewcommand{\Im}{\operatorname{\rm Im}}

\newcommand{\bC}{{\mathbb C}}
\newcommand{\bE}{{\mathbb E}}

\newcommand{\bQ}{{\mathbb Q}}

\newcommand{\bR}{{\mathbb R}}

\newcommand{\bZ}{{\mathbb Z}}

\newcommand{\cK}{{\mathcal K}}

\newcommand{\cD}{{\mathcal D}}

\newcommand{\cL}{{\mathcal L}}

\newcommand{\cC}{{\mathcal C}}

\newcommand{\cU}{{\mathcal U}}

\newcommand{\hG}{{\hat G}}

\newcommand{\hV}{{\hat V}}
\newcommand{\hf}{{\hat f}}

\newcommand{\Om}{{\Omega}}

\newcommand{\al}{\alpha}

\newcommand{\be}{\beta}

\newcommand{\De}{\Delta}
\newcommand{\de}{\delta}
\newcommand{\eps}{\epsilon}
\newcommand{\ka}{\kappa}
\newcommand{\la}{\lambda}
\newcommand{\lp}{\lambda_+}
\newcommand{\lm}{\lambda_-}
\newcommand{\La}{\Lambda}
\newcommand{\mum}{\mu_-}
\newcommand{\mup}{\mu_+}

\newcommand{\sg}{\sigma}

\newcommand{\om}{\omega}

\newcommand{\ze}{\zeta}

\newcommand{\ga}{\gamma}
\newcommand{\gap}{\gamma_+}
\newcommand{\gam}{\gamma_-}

\newcommand{\Ga}{\Gamma}

\newcommand{\hh}{\hat h}

\begin{document}

\title[Fast reliable pricing and calibration of the rough Heston model II]
{Fast reliable pricing and calibration of the rough Heston model II}

\author[
S. Boyarchenko, M. de Innocentis and
S. Levendorski\u{i}]
{
Svetlana Boyarchenko, Marco de Innocentis and
Sergei Levendorski\u{i}}

\begin{abstract}
This is an extended and modified version of the preprint ``Fast reliable pricing and calibration of the rough Heston model ".
We suggest new fast and accurate methods for pricing and calibration
of the rough Heston model and analyze relative advantages of several popular methods of numerical Fourier
inversion. The pricing method and analysis of other methods are quite general and applicable to wide classes of models where the characteristic functions can be calculated explicitly or numerically. In application to the rough Heston model,  it is necessary  to numerically solve the fractional Volterra equation. We analyze sources of errors of several variations of the Adams method, and construct a modification that increases accuracy of the solution.   For moderate or long maturities and strikes near spot, thousands of prices are computed in several milliseconds (ms) in \textsc{Matlab} on a Mac with moderate specs, with relative errors $\lesssim 10^{-4}$. Even for options close to expiry and far-OTM, the pricing takes a few tens or hundreds of ms.  
\smallbreak 

\noindent We show that, for the calibrated parameters in El Euch and Rosenbaum (Math.\ Finance 2019, v.~29), the model implied vol surface is much flatter and fits the market data poorly; thus the calibration in \emph{op.\ cit.} is a case of ``ghost calibration'' (M.~Boyarchenko and S.~Levendorski\u{i}, Quant.\ Finance 2015, v.~15): the numerical error and model specification error offset each other, creating an apparently good fit that vanishes when a more accurate pricer is used. We explain how such errors arise in popular iFT implementations that use fixed numerical parameters, yielding spurious smiles/skews, and  provide numerical evidence that SINH acceleration (S.~Boyarchenko and S.~Levendorski\u{i} IJTAF 2019, v.~22), of Fourier inversion 
is faster and more accurate than competing methods. Robust error control is ensured by a general Conformal Bootstrap principle that we formulate; the principle is applicable to many Fourier-pricing methods. We outline how this principle and our method enable accurate calibration procedures that are hundreds of times faster than approaches commonly used in the industry. Using the methods developed in the paper and the CB principle, we address the outstanding problem
\emph{Markovian or non Markovian} - which class of models fits the empirical data best.

\medbreak

 \noindent \emph{Disclaimer:} The views expressed herein are those of the authors only.
 No other representation should be attributed.
 
\end{abstract}

\thanks{
\emph{S.B.:} Department of Economics, The
University of Texas at Austin, 2225 Speedway Stop C3100, Austin,
TX 78712--0301, {\tt sboyarch@utexas.edu} \\
\emph{M.I.:}  Deutsche Bank, 21 Moorfields, London, EC2Y 9DB.
Email address: {\tt
marcdein@gmail.com} \\
\emph{S.L.:}
Calico Science Consulting. Austin, TX.
 Email address: {\tt
levendorskii@gmail.com}}
    
\maketitle



\tableofcontents

\section{Introduction}\label{s:intro} 
Pricing methods based on the Fourier transform are amenable to fast calculations, hence, widely used
for calibration of various models.
However, pricing errors can have serious impact on calibration results and an assessment of 
the relative performance of models. As demonstrated in  \cite{one-sidedCDS}, naive applications of popular Fourier methods with fixed, non-optimal parameters can lead to \emph{ghost calibration}. In this phenomenon, numerical inaccuracies of the pricing algorithm systematically offset the model's specification errors, producing a fallacious, yet seemingly excellent, fit to market data with incorrect model parameters and a distorted implied volatility surface. Conversely, the correct model may be ruled out simply because, at the true parameters, numerical error dominates: 
  \emph{sundial calibration} \cite{paraHeston}: a sundial never shows midnight. 
  The first aim of the paper is to analyze sources of errors of several popular methods of Fourier inversion,
  and compare their performance in application to the Heston model and rough Heston model.
  The analysis admits a straightforward modification to the Bates model and its rough analog,
  and, if errors of numerical evaluation of the characteristic function are properly taken into
 account, to more complicated stochastic volatility (SV) models and models with stochastic interest rates.  In the case of the Heston model, Bates model and other simple models, the characteristic 
  function is calculated explicitly. If certain complex-analytical subtleties are properly taken into account,
  a careful choice of an algorithm for the Fourier inversion guarantees that the option price
  is calculated sufficiently accurately. 
  If the characteristic function can be calculated only numerically, then 
  an additional source of errors appears, and the errors increase as the maximal absolute value
  of the spectral parameter $\xi$ that is used in the quadrature for the Fourier inversion increases. 
  Let $N$ be the number of terms in the numerical procedure for the Fourier inversion, and $M$ the number of time steps
  in a chosen quadrature for the evaluation of the characteristic function. Then, typically, $M$ necessary
  to satisfy a given error tolerance increases with $N$, and, for small times to maturity, $N$ must be large;
  hence, for small maturities, both $M$ and $N$ must be large, and accurate option pricing becomes
  very time consuming. If time to maturity $T$ is large and a relatively simple quadrature for the evaluation
  of the characteristic function is used, then $M$ must be large as well. Only for moderate maturities, and for strikes $K$ not far from the spot, moderately large $N$ and $M$ can be used in essentially any reasonable procedure.
  In other regions of the maturity-strike space, an accurate pricing procedure is difficult to design. 
  Following Leo Tolstoy (``All happy families are alike; each unhappy family is unhappy in its own way"), we 
  formulate the \emph{Anna Karenina principle for option pricing}: in a good region of the $(K,T)$-space, all reasonable models and pricing methods are alike; close/far from maturity and far in the tails, models and pricing methods perform differently\footnote{Disclaimer: we are not the first to formulate the principle. J. Diamond 
 \cite{Diamond1994}
formulated the principle in applications to biology: ``A deficiency in any one of a great number of factors can render a species undomesticable. Therefore, all successfully domesticated species are not so because of a particular positive trait, but because of a lack of any number of possible negative traits." }.
Even if the pricer is sound, in applications, the restrictions on the CPU cost imply that one cannot choose both
$N$ and $M$ very large. One is tempted to use the  same parameters of the numerical scheme
for all $(K,T)$ in the data set and all parameters of the model. Hence, typically, too small $N$ and/or $M$ are
used, and the error depends on the choice. If $N$ is small, the error of the numerical Fourier
inversion is large; if $N$ is large but $M$ is insufficiently large, then the characteristic function
is evaluated with large errors; the relative errors can be of the order of thousands percent and more.
We formulate   the \emph{ Uncertainty Principle of calibration}: using different parameters of the numerical scheme,
one can produce a host of different prices and volatility curves and surfaces, and choose shapes one likes better.
  
  In the paper, we demonstrate these effects in the case of the rough Heston model \cite{EuchRosenbaum2019}.
 We suggest new fast and accurate methods for pricing and calibration
of the rough Heston model and analyze  errors of several popular methods of numerical Fourier
inversion, which lead to the ghost and sundial calibration. 
The pricing method and analysis of other methods are quite general and applicable to wide classes of models where the characteristic functions can be calculated explicitly or numerically. 
In particular, modifications of pricing and calibration procedures constructed in the paper can be
easily adjusted to generalizations of the rough Heston model of Bates type. In application to the rough Heston model,  it is necessary  to numerically solve the fractional Volterra equation. We suggest modifications that increase accuracy of the solution. 
As it is explained in \cite{pitfalls}, even in classical affine models,
accurate numerical solution of a system of Riccati equations can be difficult 
and serious errors can result. In the numerical example for the two-factor CIR model considered
in \cite{pitfalls}, 
a numerically calculated trajectory for the cumulant $\phi(T,\xi)$, with a fairly small time step, blows up as $T\to \infty$ although
according to the theoretical analysis, the exactly calculated trajectory does not blow up. In the case of the fractional Volterra equation,
the theoretical results available in the mathematical literature are very far from being complete, in the case of large (in absolute value) spectral parameter $\xi$ especially. For design of accurate numerical schemes, the following difficulties must be
taken into account; similar difficulties arise for other models but the detailed analysis is model-specific.  Let time to maturity $T$ be fixed. For accurate Fourier inversion, $\phi(T,\xi)$ with large $\xi$ must be used
and the number of terms $N$ in the chosen quadrature must be large;
for large $\xi$, very large number of time steps $M$ is necessary to calculate $\phi(T,\xi)$
sufficiently accurately. We consider in detail several variations of the Adams method in applications
to the fractional Volterra equation.
The inherent instability of numerical solutions of the latter leads to the following
interesting effect.    Let $T$ and $M$ be fixed, and $M$ is insufficiently large. Then, if one variation
of the Adams method is used,
$\Re \phi(T,\xi)$ may start  to increase as $\xi\to \infty$ along the line of integration. Hence,
the pricing error becomes too large if large $\xi$ are used, and the calculated price falls outside the no-arbitrage bounds. 
If another variation of the Adams method is used,
 $\Re\phi(T,\xi)\to -\infty$, and very fast. Then the sequence of prices evaluated using longer $\xi$ 
 grids converges extremely fast, and one is tempted to conclude that the convergence of the method
is excellent - although the reason is that all additional terms become irrelevant. The excellent convergence is fallacious, and the residual non-negligible error remains. The residual error depends on the parameters of the model, maturity and strike, and can be rather large. 
To control the errors in difficult situations where the theoretically sound error control is unavailable
or too complicated for practical purposes, we suggest the \emph{Conformal Bootstrap principle} (CB principle).
We use the principle in applications to the Heston and rough Heston models; the principle can be used
in applications to other models as well. The main idea is to use different conformal deformations of the
contour of integration in the pricing formula, make the corresponding conformal changes of variables and apply a good quadrature to the resulting integral over $\bR$ or $\bR_+$. If the prices obtained with  different deformations agree up to the error $10^{-m}$, where $m\ge 5$, we surmise that each of the prices differ from
the one calculated with a perfect pricer by not more than $10^{-m+2}$. If a simply connected region of analyticity of 
the characteristic function where the initial contour and deformed one lie are known and the leading term of asymptotics of the cumulant at infinity is known as well, then admissible deformations are known. Furthermore,  for each deformation,  sufficiently accurate prescriptions
for the step $\ze$ and number of terms $N$ of the simplified trapezoid rule are easy to derive. The prescription being an approximate one, small adjustments by 10\% - 30\%  may be needed; in some cases,
up to 80\%.
The CB principle allows one to essentially guarantee that the pricing formula 
with the chosen $\ze$ and $N$ produces prices with the desired accuracy.
We discuss the reliability
of the ad-hoc CB principle in the case when the domain of the analyticity of the integrand in the pricing formula
is unknown.

In the first version of the paper, as an efficient family of deformations, we used the sinh-deformation 
defined by conformal maps of the form $\xi = i\om_1+b\sinh(i\om+y)$ and the simplified trapezoid rule.
In addition to the sinh-family of deformations, we also used (in different publications), the fractional-polynomial  
deformations and logarithmic deformations of $\bR$, and, in applications to evaluation of probability 
distributions of stable L\'evy processes, deformations of $\bR_+$. 
In the case of integrals over $\bR_+$, the most efficient family is the rotation of $\bR_+$. After the rotation, we apply the exponential change of variable and simplified trapezoid rule. 
The families of deformations of $\bR$ and $\bR_+$ are sufficiently flexible for efficient evaluation of distributions and pricing options
if the characteristic function admits a bound via  $\exp(-c|\xi|^\nu)$, where $\nu\in (0,2)$, in the union of a strip
and cone around the line of integration in the Fourier inversion formula. The families work well even if $\nu\in(0,1)$, hence, the characteristic function decays slowly at infinity.
In the present paper, we introduce  an alternative to the sinh-acceleration,  the rotation of the contour of integration
$\bR_+$ in the Gauss-Laguerre quadrature (GL quadrature). Assume that the integrand can be evaluated
very accurately at a small CPU cost and the error tolerance is not smaller than E-10, which is sufficient
for any practical application. 
We prove that if the sinh-deformation (or rotation of $\bR_+$)
is  applicable and either $\nu>1$ or $\nu=1$ and $c>1$,
then the efficiency of the rotated GL quadrature is comparable to the sinh-acceleration 
but the number of terms in the GL quadrature is 20-50\% larger. 
If $\nu<1$ or $\nu=1$ and $c<1$, then the theoretical error bound for the GL quadrature 
gives $+\infty$. If $f(y)$ is the integrand, the error bound is in terms of high order derivative of
$f_0(y)=e^yf(y)$. Hence, one expects that increasing number of terms, one obtains a diverging sequence.
Surprisingly, we observed that, for moderately large number of terms, the GL quadrature works well if $\nu=1$ and $c<1$ (which is the case
for all parameter sets of the Heston model and rough Heston models that we considered if time to maturity $T$ is not rather large)  but $c$ is
not too close to 0. Apparently, in cases when the GL quadrature is surprisingly accurate, there exists a sufficiently accurate
approximation of the integrand by a function which decays at infinity faster than $e^{-|\xi|}$ but
insignificantly differs from the characteristic function on the interval that contains all nodes of the GL quadrature with
a chosen number of terms $N$. However, relatively simple error bounds  that we managed 
to derive using this interpretation imply significantly larger errors than we observe. 

Thus, one can try to use the rotated GL quadrature instead of the sinh-acceleration but the results
will be not reliable, for short maturity options especially unless high precision arithmetic is used.  In the case of Heston model and rough Heston model,  $c=c(T)\to 0$ as $T\to 0$; heuristically, we hope
that if the results produced for 2-3 different rotations agree well, then the prices are sufficiently accurate,
as in the case when the sinh-acceleration is applied. This conclusion can be false
if the fallacious variation of the Adams method is used.  In this case several dozens of non-neglible terms
in the GL quadrature applied after different rotations are values of the same analytic functions. Hence,
the difference of the price calculated using an ideal pricer and 
the sums of the non-negligible terms in the simplified trapezoid rule must be identical (up to the machine error).

Using the methods developed in the paper and the CB principle, we address the outstanding problem
\emph{Markovian or non Markovian} - which class of models fits the empirical data best. 
Starting with the celebrated Heston model \cite{heston-model}, affine models have become one of the most popular classes of stochastic
volatility models, term structure models, and models in FX. The popularity is due to the fact that
the characteristic function in an affine model can be explicitly calculated solving an associated system
of generalized Riccati equations \cite{DFS}, hence, the Fourier transform technique allows
one to express prices of options of the European type as oscillatory integrals.  However, despite their analytical convenience, models based on classical affine diffusions such as Heston have well-documented limitations in reproducing key features of observed implied volatility surfaces. In particular, these models are unable to  capture adequately the pronounced short-maturity steepness of implied volatility smiles -- a feature consistently observed in equity and FX markets. In response to these shortcomings, a new class of models, that of \emph{rough volatility} models, has been proposed. These models replace the standard Brownian Motion (BM) drivers of volatility with fractional processes characterised by a Hurst index $H \in (0,1/2)$. 
Empirical evidence presented by Gatheral, Jaisson, and Rosenbaum~\cite{GatheralJaissonRosenbaum2018} demonstrated that, at the time, the log-volatility of financial assets behaved as a fractional Brownian motion with $H \approx 0.1$, both in historical and implied volatility data. This low Hurst index implies a ``rough'' volatility path that is far less smooth than Brownian motion, leading to stronger short-term memory and more accurate modelling of the volatility clustering and bursts which are often observed in financial markets.
For various aspects of rough volatility models, see \cite{BayerFritzGatheral2015,GatheralJaissonRosenbaum2018,JacquierMartiniMuguruza2018,EuchRosenbaum2019,FordeZhang2017,FordeSmithViitasaary2021,FrizGassiatPigato2021,FrizGassiatPigato2022,BayerFrizFukasawaGatheralJacquierRosenbaum2024} and the bibliographies therein. 

 By the law of the pendulum, several recent empirical publications (see, e.g., 
\cite{GuyonAmrani,Delemotte} and the bibliographies therein) cast doubts on the superiority of rough volatility models as compared to
Markovian ones. 
The following main shortcomings are formulated: the volatility skew is too high at the short end and
too low at the long end as compared to the skew inferred from the market data, the log-log plot of ATM skew
is a downward sloping straight line although the corresponding curve inferred from the market data
is concave, and the implied volatility surfaces in rough volatility models agree poorly with the market ones.
In the paper, we include a preliminary analysis of  these conclusions. A detailed analysis 
requires a systematic empirical study which is outside the scope of the paper.  In several examples of implied volatility curves calculated for the rough Heston model, which
we found in the literature, the errors are of the order of several percent or more, hence, the calculated model
surfaces are far from the correct ones.  Thus, the claim that the surfaces in rough Heston models fit the market ones very poorly
can be attributed to inaccuracy of the numerical methods used.  If the pricing method is sufficiently accurate, the model surface can be fairly close to the market one, even in a fairly difficult situation.
In the present paper, we calibrate 
the rough Heston model to the implied volatility surface of TSLA on 2 May, 2025\footnote{Tesla (TSLA) was among the most actively traded option underlyings in 2025, which makes it a natural test case.}. The trading activity
was hectic on this day, which makes the example interesting from the point of view of applications to the risk management, and the calibrated data set enjoys (rather, suffer) difficult properties for any two-factor model to fit well. (We discuss non-typical properties this set of data below.)
Nevertheless, the model surface is fairly close to the market one (see Fig.~\ref{3surfaces} and implied volatility curves in 
Sect.~\ref{ss:calib_sinh}).

If the model implied volatilities differ from the market ones by several percent or more, the accuracy of
the numerical calculation of the derivative of the implied volatility w.r.t. strike becomes very poor indeed. 
In many cases, if a more accurate pricing method is used, the shapes of the implied volatility surfaces, curves, and skews 
significantly change.
For instance,  in the numerical example in \cite{EuchRosenbaum2019}, 
the ATM skew at the short end drops by more than 100\% when we recalculate the skew using more accurate methods. We calculated the ATM skews for the sets of parameters used in the paper and several other sets; in all cases,
the log-log ATM skew is concave up starting from maturities 1-3 days (depending on model's parameters).
In the case of TSLA, 2 May, 2025, the concavity is small.

In our numerical experiments, we observed that as the accuracy of the  pricing procedure
improves, the skew becomes significantly  lower at the short end and may become higher at the long end.
 Sufficiently accurate calculations for the interval 1D-1W require $M=40000$ even if the new efficient modification
of the Adams method is used. We suspect that  all numerical results that can be found in the literature
are obtained with less accurate numerical schemes and time steps are larger than necessary. 
For instance, in \cite{Radoicic-Gatheral}, $M=200$ is mentioned, in other texts, $M=2000$ is used.
We found that to obtain good calibration results for the data set used in the present paper, even $M=4000$  
is too small if one wishes to obtain sufficiently accurate results. Hence, the calculated curves and surfaces that can be found in the literature are incorrect and correctly calculated ones
can fit the market ones significantly better.
The analysis of the sources of errors at the short end is one of the main objectives of the paper.\footnote{
To conduct a satisfactory analysis at the long end, where the errors of numerical solution of the fractional Adams equation accumulate, it is necessary to use hybrid methods with more efficient methods applied
outside a small neighborhood of 0. We leave this natural extension of the paper to the future. It is possible that the agreement of the model skew with the market one will improve in the region of long maturities as well.}

%
\begin{figure}[p]
	\centering
	\includegraphics[width=0.65\textwidth]{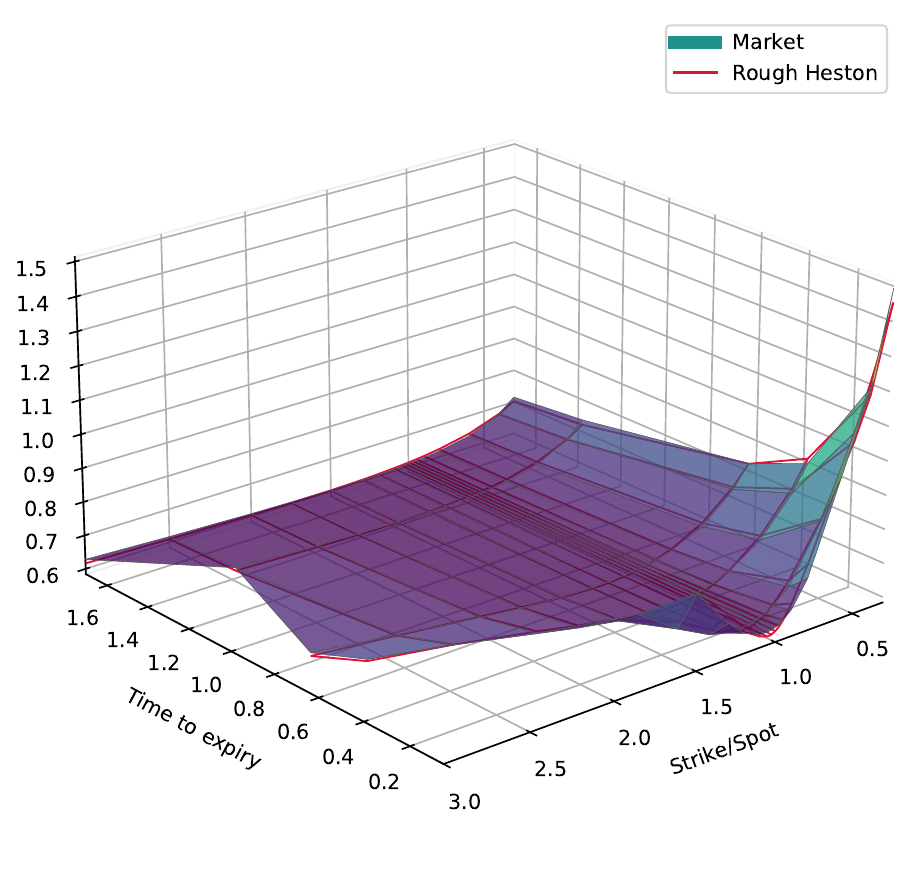}\\[1ex]
	\includegraphics[width=0.65\textwidth]{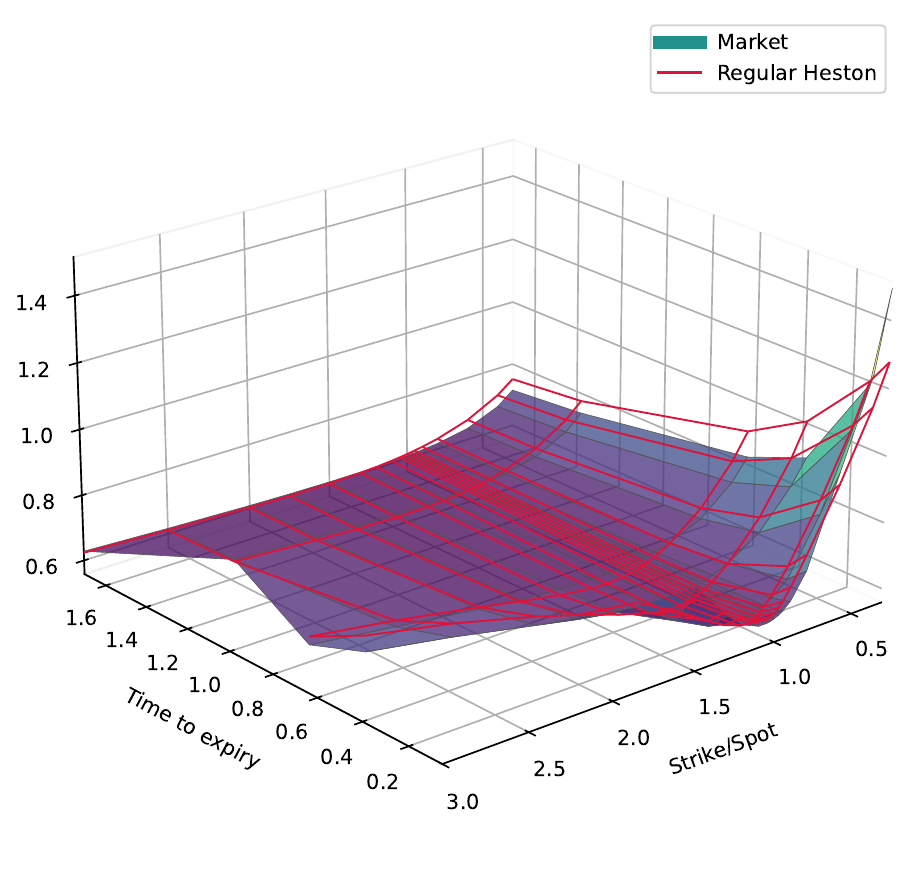}
	\caption{\small Upper panel: TSLA, 2 May 2025, market implied volatility surface (solid) versus the calibrated rough Heston model (red wireframe). Lower panel: market surface versus the regular Heston model. The rough Heston model is calibrated to implied volatilities between 3W and 1.7Y using  SINH acceleration and the hybrid BL-Adams method. Model parameters are in \eqref{params:rough_long}.}
	\label{3surfaces}
\end{figure}
\begin{figure}
\scalebox{0.8}
{\includegraphics{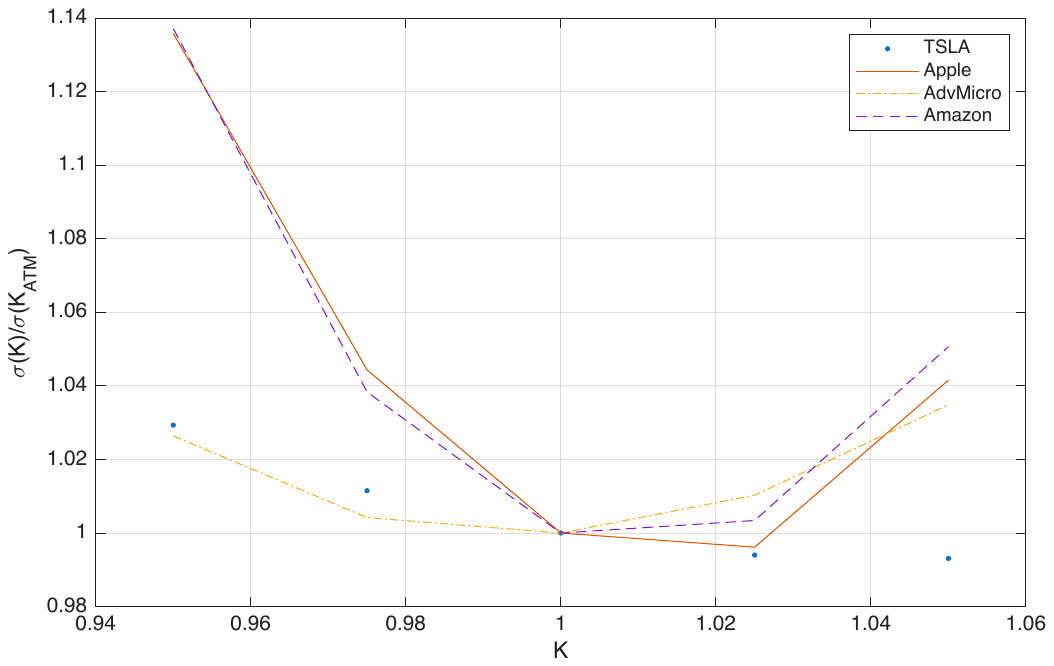}}
\vskip-5cm
\caption{\small 1W implied volatilities normalized by the spot. Apple, Advanced Micro Devices (AMD) and Amazon: 
8 Aug, 2025; TSLA: 2 May 2025}
\label{fig:normalized curves}
\end{figure} 
We were unable to find examples in the literature
where very fine grids were used, hence, the common belief about the shape of the log-log ATM skew
in the rough Heston model is based not on reliable calculations of model skews but  on the asymptotic result in \cite{Fukasawa2017}. But an asymptotic formula cannot be expected to give a good approximation outside a small neighborhood of zero.
An additional problem is that the observed implied volatility curves are just several  dots - for short maturities, about a dozen or less. Hence,  \emph{the  market skew} is not even a well-defined object.
If we define the market skew using finite differences (as in \cite{Fukasawa2017}), then, typically, the result strongly depends on the choice of the finite difference and the difference among different market skews blows up as $T\to 0$. To make a meaningful comparison of the market and model skews, it is
necessary to know how the market skew is calculated and apply the same approximate procedure
to the model. Typically, at the short end, the accurately calculated derivative differs significantly from
each finite difference. 
 An example using the real data TSLA on 2 May, 2025, obtained on a hectic date  of the financial 
shows an unexpectedly good performance of the rough Heston model. The rough Heston model
gives excellent fit
even when extrapolated across expiries. The out-of-sample performance of the Heston model 
is significantly worse. We have chosen the date for calibration as an important one from the point
of view of applications in risk management. Later, we observed the following interesting atypical feature of
this empirical parameter set: in a sizable  neighborhood of the spot, the implied volatility curve
is close to a segment of the downward sloping straight line.  See Fig.~\ref{fig:normalized curves},
where we compare 1W curves for Apple, AMD and Amazon on 
8 Aug, 2025, and the leading example TSLA on 2 May, 2025. The tails are steeper, especially the left one (see Fig. ~\ref{fig:normalized curves}).  
A shape of this kind can be reproduced in a one-factor model with a reasonably small AVE as we demonstrate; but then
the ATM skew  significantly differs from the market one. The example of TSLA on 2 May, 2025, 
seems to be an outlier. We leave to the future a study of models which can reproduce
the implied volatility curves with an approximately straight segment around the ATM strike.  Probably, one-factor models
are incapable of  producing  a small AVE \emph{and the ATM skew which is close to the market one}.
In other respects, the rough Heston model performs in this case very well. The detailed analysis
of the out-of-sample performance is in Sect.~\ref{s:model_calib}. 

A further motivation for rough volatility models comes from counterparty credit risk. One of us implemented the Heston model for risk-factor evolution and pricing within the counterparty credit risk (CCR) engine of a major bank, a setting in which the use of a stochastic volatility model had until a few years earlier seemed impractical, but became necessary for realistic exposures. The extreme  volatility clustering seen in recent episodes of market stress suggests that rough volatility models may, in time,  play a similar role in CCR and CVA, as well as in stress testing. This however requires pricing to be fast and reliable enough for the millions of revaluations such applications demand, which is  the use case our method targets.

The rest of the paper is organized as follows. The rough Heston model and several methods of
the solution of the fractional Volterra equations are introduced in Sect.~\ref{s:Rough-Adams}.
In Sect.~\ref{s:SINH-CB}, we recall the Flat iFT and  SINH-acceleration methods, give explicit recommendations
for the choice of the parameters of the SINH-acceleration, and consider the Conformal Bootstrap principle. 
In Sect.~\ref{s:FT}, we analyze several popular methods of the Fourier inversion. In particular, we outline the rotated version of the Gauss-Laguerre quadrature and a version of the Conformal Bootstrap principle
based on the the rotated version. Numerical examples are in Sect.~\ref{s:numer}. Pricing algorithms and calibration results are 
in Sect.\ref{s:fast_pricing}-~\ref{s:model_calib}. Sect. ~\ref{s:concl} concludes. The detailed pricing/calibration algorithms and additional tables and figures are in the appendix.

\section{Rough Heston model and the fractional Volterra equation}\label{s:Rough-Adams}

\subsection{Formulas for the characteristic function}\label{ss: char_function}
The rough Heston model \cite{EuchRosenbaum2019} is constructed by replacing the variance process in the Heston model
with the fractional square root process:
\beqa\label{dynS}\label{frSqrt}
d S_t&=& S_t\sqrt{V_t} dB_t, \\\label{frVdyn}
V_t &= &V_0 + \frac{1}{\Ga(H+1/2)}\int_0^t(t-s)^{H-1/2}(\ga(\theta-V_s)ds+\gamma \nu\sqrt{V_t}dW_s),
\eqa
where $S_0, V_0>0$ and $(B_t,W_t)$ is BM in 2D with the correlation coefficient $\rho\in [-1,1]$;
the components are standard BM in 1D. Denote $\al=H+1/2$, and 
let $\al\in (0,1)$, $v, \ga,\theta,\nu>0$. It is proved in \cite{EuchRosenbaum2019,EuchRosenbaum2017} that the (conditional) characteristic function 
of the log-price  $\Phi_\al(t,T,v,\xi):=\bE[e^{i\xi X_T}\ |\ X_t=0, V_t=v]$ in the rough Heston model is of the form
\bbe\label{chFRough}
\Phi_\al(t,T; v,\xi)=\exp[g_1(\xi,\tau)+vg_2(\xi,\tau)],
\ee
where $\tau=T-t$, 
\bbe\label{eq:g1g2}g_1(\xi,\tau)=\theta\ga\int_0^\tau h(\xi,s)ds,\ g_2(\xi,\tau)=I^{1-\al}h(\xi,\tau),\ee
and $h(\xi,\cdot)$ is the solution of the fractional Riccati equation
\bbe\label{RiccRough}
D^\al_t h(\xi, t)=-\frac{1}{2}(\xi^2+i\xi)+\ga(i\xi\rho\nu-1)h(\xi,t)+\frac{(\ga\nu)^2}{2}h(\xi,t)^2,
\ee
subject to $I^{1-\al}h(\xi,0)=0$. Recall that, for $\al\in (0,1)$, $I^\al$ and $D^\al$ are the fractional integral and differential operators:
\beqa\label{defIal}
I^\al u(t)&=&\frac{1}{\Ga(\al)}\int_0^t (t-s)^{\al-1}u(s)ds,\\
D^\al u(t)&=&\frac{1}{\Ga(1-\al)}\frac{d}{dt}\int_0^t (t-s)^{-\al}u(s)ds.
\eqa
Introduce the notation
\bbe\label{defF}
F(\xi,h)=-\frac{1}{2}(\xi^2+i\xi)+\ga(i\xi\rho\nu-1)h+\frac{(\ga\nu)^2}{2}h^2.
\ee
Equation \eq{RiccRough} subject to $I^{1-\al}h(\xi,0)=0$ is equivalent to the following  Volterra equation
\bbe\label{Volterra}
h(\xi,t)=I^\al F(\xi,t)=\frac{1}{\Ga(\al)}\int_0^t(t-s)^{\al-1}F(\xi,h(\xi,s))ds.
\ee
In \cite{EuchRosenbaum2019}, \eq{Volterra} is solved (numerically) using the fractional Adams method. It is 
not explained how $g_1$ and $g_2$ are evaluated. Presumably, using the piece-wise linear interpolation as in the fractional Adams method: the trapezoid rule and fractional trapezoid rule, respectively. Since $h$ is not smooth at 0 and an additional fractional integral needs to be evaluated, the errors increase. 
We use the following version of \eq{chFRough}, thereby avoiding additional errors.
\begin{prop}\label{prop:newRoughHestonChExp}
Let $\al\in (0,1)$, $v, \ga,\theta,\nu>0$, $\rho\in (-1, 1)$, and let $h(\xi,t)$ be
the solution of \eq{Volterra}. Then
\bbe\label{chFRough2}
\Phi_\al(t,T,v,\xi)=\exp\left[\int_0^\tau (\ga\theta h(\xi,s)+vF(\xi,h(\xi,s)))ds\right].
\ee
\end{prop}
\begin{proof} It suffices to note that $I^{1-\al}I^\al=I^1$.
\end{proof}

\subsection{Asymptotics of $\phi(\xi,\tau):=g_1(\xi,\tau)+vg_2(\xi,\tau)$ as $\tau\to \infty$}\label{AsymRough}
 The asymptotic formula is an analog of the formula rigorously 
derived in \cite{paraHeston} in the Heston model. Unfortunately, we were unable
to rigorously prove the formula in the case of the rough Heston model and the validity of the assumption that
$\phi(\xi,\tau)$ admits analytic continuation to an open cone $\cC$ containing $(0,+\infty)$. We verified the latter property
for a number of sets of the parameters of the model using the conformal bootstrap principle.
The second assumption that we make is that as $\xi\to \infty$ remaining in $\cC$,
\bbe\label{asphi}
h(\xi,\tau)=h_\infty\xi + O(1),
\ee
where $h_\infty\in \bC\setminus\{0\}$. 
 We also empirically verified this assumption in a number of examples. 
Substituting \eq{asphi} into \eq{Volterra}, dividing by $\xi$ and passing to the limit
$\xi\to \infty$, we observe that \eq{asphi} fails unless $h_\infty$ is a solution
of the equation 
\[
-\frac{1}{2}+i\rho\ga\nu h_\infty+\frac{(\ga\nu)^2}{2}h_\infty^2=0.
\]
As in the case of the Heston model, we need the solution in the left half-plane,
therefore, 
\bbe\label{hinfty}
h_\infty = -\frac{i\rho+\sqrt{1-\rho^2}}{\ga\nu}.
\ee
Now we can calculate the asymptotics of $\phi(\xi,\tau)$:
\bbe\label{eq:asphi}
\phi(\xi,\tau)/\xi=\ga\theta h_\infty \tau 
+\frac{v_0h_\infty}{\Ga(1-\al)}\int_0^\tau (\tau-s)^{-\al}ds+o(1)
=-c_\infty(\tau)+o(1),
\ee
where 
\bbe\label{cinf}
c_\infty(\tau) = \frac{i\rho+\sqrt{1-\rho^2}}{\ga\nu}\left(\ga\theta \tau +\frac{v_0\tau^{1-\al}}{\Ga(2-\al)}\right),
\ee
and  $v_0 := V_0$. Note that if $\al=1$, then the asymptotic formula \eq{eq:asphi} coincides with the formula for the Heston model
derived in \cite{paraHeston}. 
We observed
that  numerically calculated $\phi$ satisfies $\phi(\xi,\tau)/\xi\to c_{\infty}(\tau)$ with a good accuracy.
See Fig.~\ref{phi-curves}.

\begin{figure}
     \caption{Curves $\cL_{\om_1,b,\om}\ni \xi\mapsto-\phi(\xi,\tau)/\xi\in \bC$ for $\om=-0.2, 0.0, 0.2$ and
     $\tau=1/252$. Parameters
    of the rough Heston model are in \eq{parEuRos},
    $c_\infty(\tau)=0.1222 - 0.1136i$. We observe a small discrepancy in the imaginary part of $c_\infty(\tau)$.
    For the choice of $\La=N\ze$, the real part plays a more important role than the imaginary part.}
    \vskip-3cm
    \includegraphics[width=0.9\textwidth,height=0.7\textheight] {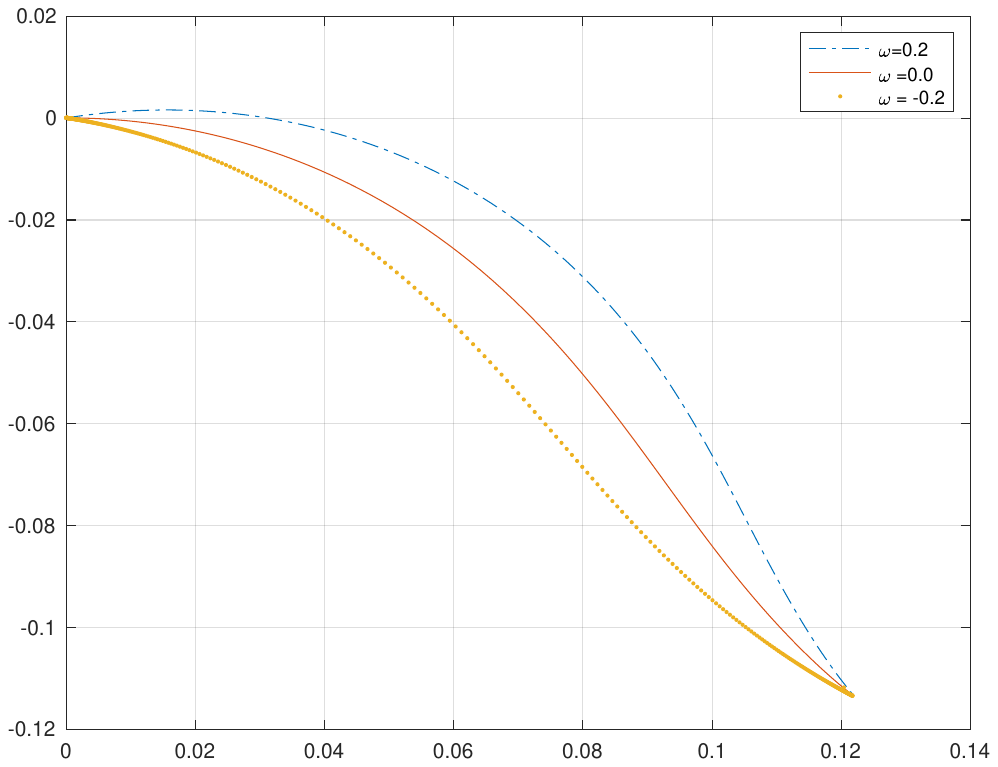}
    \label{phi-curves}
\end{figure}
However, for large $|\xi|$, the difference $\phi(\xi,\tau)-c_{\infty}(\tau)\xi$ can be large,
and the sign of $\Re (\phi(\xi,\tau)-c_{\infty}(\tau)\xi)$ can strongly depend on a numerical scheme chosen
to evaluate $\phi(\xi,\tau)$.

\subsection{Ad hoc recommendation for the choice of the cone of analyticity}
The choice of the cone $\cC_{\gam,\gap }$  is necessary to give a recommendation 
for the choice of the crucial parameters $\om$ and $d$, and then $\ze$ and $\La$ of the SINH-acceleration method.
We conjecture that the asymptotic formulas \eq{hinfty}-\eq{cinf} are valid and assume that the spot
is 1. We set $c^{re}_\infty (\tau)= \Re c_\infty(\tau)$, $c^{im}_\infty(\tau)=c_\infty(\tau)$, and  
 find $\cC_{\gam,\gap }$ as a cone such that for any ray $e^{i\ga}\bR_+\subset \cC_{\gam,\gap}$,
 the real part of
$
-[i\xi\ln K + i c^{im}_\infty(\tau)+c^{re}_\infty(\tau)]e^{i\ga}$ is positive. Define
\[
\ga^0(K, c_\infty(\tau))=\arctan((\ln K+c^{im}_\infty(\tau))/c^{re}_\infty(\tau)). 
\]
The admissible $\ga$ must satisfy  $\ga-\ga^0(K, c_\infty)\in (-\pi/2,\pi/2)$, hence, \[
\gam:=\gam(K,\tau)=-\pi/2+\max\{0,\ga^0(K, c_\infty(\tau))\},\
\gap: =\gap(K,\tau)= -\pi/2+\min\{0, \ga^0(K, c_\infty(\tau))\},
\]
and approximately optimal recommendation for the choice of $\om, d$ are as follows.
\begin{enumerate}[I.]
\item
For one strike $K$,  $\om(K,\tau) = (\gap(K,\tau)+\gam(K,\tau))/2$,  
$d(K,\tau)=k_d(\gap(K,\tau)-\gam(K,\tau))/2$, where $k_d<1$ is close to 1, e.g., $k_d=0.9$.
\item
Given a set of strikes $\cK=(K_j)_{j=1}^m$, the same $\om$ and $d$ can be used for all $K\in \cK$
only if $\max_K\gam(K,\tau)=:\gam<\gap:=\min_K\gap(K,\tau)$, and 
then we recommend the choice
$\om(\tau) = (\gap(\tau)+\gam(\tau))/2$ and $d(\tau)=k_d(\gap(\tau)-\gam(\tau))/2$.
 \end{enumerate}

\subsection{Fractional Adams method and its modifications}\label{ss:Adams_modif}

\subsubsection{Adams method} 
 One 
 fixes a uniform grid $(t_j)_{j\in \bZ_+}$, $t_j=j\De$, and 
calculates the approximations $\hh(\xi,t_k), k=1,2,\ldots,$ in two steps. First, the predictor
$\hh^P(\xi,t_k), k=1,2,\ldots,$ is calculated, and then the more accurate approximation $\hh(\xi,t_k), k=1,2,\ldots.$
To calculate the former, the rectangular rule is used while the latter is evaluated using the trapezoid rule. At the first steps, in the region of large $|\xi|$ and small $t_j$,
significant errors appear. The errors are especially clearly seen
at the first step of the induction procedure, which we write explicitly:
\bbe\label{hhP10}
\hh^P(\xi, t_1)=b_{0,1}F(\xi,\hh(\xi,t_0)=b_{0,1}(-0.5(\xi^2+i\xi)),
\ee
where $b_{0,1}=\De^\al/\Ga(\al+1)$. The RHS of \eq{hhP10} is of the order of $\De^\al|\xi|^2$, however,
\bbe\label{Ash0}
\hh(\xi, t)= \frac{-0.5(\xi^2+i\xi)}{\Ga(\al+1)}t^\al(1+o(1)), 
\ee
uniformly in $(\xi,t)$ in the region $\{(\xi,t)\ |\ 0\le t^\al |\xi|^2< c\}$. See \cite{RoughNotTough},
where the full asymptotic expansion is calculated. We construct modifications of the Adams method
changing the prediction step so that the asymptotics \eq{Ash0} is taken into account.
We use the same coefficients $a_{j,k}$ as in the fractional Adams method. 
Explicitly, for $k=0,1,\ldots, M-1$, set 
\[
a_{k+1,k+1}=\frac{\De^\al}{\Ga(\al+2)}, \
a_{0,k+1}=\frac{\De^\al}{\Ga(\al+2)}(k^{\al+1}-(k-\al)(k+1)^\al),\]
and, in the cycle $j=1,2,\ldots, k$, calculate
\[
a_{j,k+1}=\frac{\De^\al}{\Ga(\al+2)}((k-j+2)^{\al+1}+(k-j)^{\al+1}-2(k-j+1)^{\al+1}).
\]
\subsubsection{Modification BL in \cite{RoughHestonWeMarco2025} a.k.a. Modification III in \cite{RoughHestonWe2024}}
Introduce the scaled unknown and its leading asymptotic part
\[
\tilde\hh(\xi,t):=(1+|\xi|)^{-1}\hh(\xi,t),\qquad
\tilde\hh_{as}(\xi, t):=(1+|\xi|)^{-1}\hh_{as}(\xi, t),
\]
where
\[
\hh_{as}(\xi, t)= -\frac{1}{2}\,\frac{(\xi^2+i\xi)}{\Gamma(\al+1)}\,t^\al.
\]
Define the scaled remainder $\tilde\hh^1(\xi,t):=\tilde\hh(\xi,t)-\tilde\hh_{as}(\xi, t)$, and use, in place of the function
$F(\xi,h)$ given by \eq{defF}, the transformed version
\bbe\label{defFas1_tilde}
\tilde F_{as1}(\xi,\tilde h_{as},\tilde h^1)
=\ga(i\xi\rho\nu-1)(\tilde h_{as}+\tilde h^1)+(1+|\xi|)\frac{(\ga\nu)^2}{2}(\tilde h_{as}+\tilde h^1)^2.
\ee
Set $\tilde\hh^1(\xi, 0)=0$, and then, in a cycle $k=0,1,\ldots, M-1$,
\begin{enumerate}[1.]
	\item Calculate $\tilde\hh_{as}(\xi, t_{k+1})$, then evaluate the predictor
	\beqa\label{tAsk}
	\tilde\hh_0(\xi,t_{k+1})&=&\sum_{0\le j\le k}a_{j,k+1}\,
	\tilde F_{as1}\big(\xi,\tilde\hh_{as}(\xi,t_j),\tilde\hh^1(\xi,t_j)\big),\\
	\label{thh1k}
	\tilde\hh^1(\xi,t_{k+1})&=&\tilde\hh_0(\xi,t_{k+1})+a_{k+1,k+1}\,
	\tilde F_{as1}\big(\xi,\tilde\hh_{as}(\xi,t_{k+1}),\tilde\hh^1(\xi,t_{k+1})\big).
	\eqa
	\item For $m=1,\ldots,n$, perform the Picard correction
	\bbe\label{titerk}
	\tilde\hh^1(\xi,t_{k+1})=\tilde\hh_0(\xi,t_{k+1})+a_{k+1,k+1}\,
	\tilde F_{as1}\big(\xi,\tilde\hh_{as}(\xi,t_{k+1}),\tilde\hh^1(\xi,t_{k+1})\big).
	\ee
	\item Set the unscaled value
	\bbe\label{recon_h}
	\hh(\xi,t_{k+1}):=(1+|\xi|)\big(\tilde\hh^1(\xi,t_{k+1})+\tilde\hh_{as}(\xi, t_{k+1})\big).
	\ee
	\item Calculate the integral on the RHS of \eq{chFRough2} using the trapezoid rule.
\end{enumerate}

\begin{rem}{\rm 
\begin{enumerate}[(a)]\item
A pseudo-code implementation scheme for the BL Modification can be found in Appendix. 
\item
	One can use asymptotic expansions of higher orders but in our numerical experiments
	with the two-term asymptotic expansion, the latter brings no advantages.\item The accuracy of the numerical scheme can be significantly increased using the time-grid depending on $\xi$, and non-uniform grids.
	\item
	The accuracy of calculations can be increased using  grids $\{t_k\}$ depending on $\xi$. If calculations
	of $\phi(\xi, T)$ for each $\xi$ in the pricing formula are parallelized, the CPU time decreases.
Non-uniform grids can be indispensable for pricing options of long maturities.
	\item
	The CPU time of the BL-modification can be decreased sizably and accuracy somewhat improved using the standard prediction-correction method starting with some $t^*_k(\xi)$ and larger step starting at $t^*_k(\xi)$.
	Indeed, outside a neighborhood of 0 (which depends on $\xi$ and the parameters of the model),
	the solution is sufficiently regular and not small, hence, there is no advantage to use the leading term of asymptotics. 
We leave to the future the study of an approximately optimal choice of $t^*_k(\xi)$.
\item
In the numerical examples, we use $t^*_k(\xi)=0.02$ for all $\xi$,
and call the resulting hybrid method the \emph{BL-Adams method}.
\item
Outside a neighborhood of zero, alternative methods can be used. See \cite{RoughNotTough,Radoicic-Gatheral} and the bibliographies therein. These methods can be more efficient outside a small vicinity of zero than the \emph{BL-Adams method}. We leave the study of improvements of our technique along this  line to the future. 
Improvements of this kind can be very useful for option pricing and study of the ATM skew in the long run.
 	\end{enumerate}}
\end{rem}

\begin{rem}\label{gen} {\rm BL-modification can be applied to more general kernels $K(t)$. 
Assume that (after the rescaling, if necessary) $K(t)=t^{\al-1}(1+o(1))$ as $t\to 0+$. Then BL-modification
is applied exactly as in the case $K(t)=t^{\al-1}$. Naturally, the coefficients of the Adams method must be recalculated
for the given $K$. If more than one term of the asymptotics of $K(t)$ is available, additional asymptotic terms in BL-modification can be added. Additional terms can be added if $K(t)=t^{\al-1}$ as well but we found that adding one or two additional terms does not improve the performance of BL-modification. 
If the kernel is regularized in the vicinity of 0, e.g., $K(t):=K(\eps, t)=(\eps+t)^{\al-1}(1+o(t))$, where $\eps>0$ is a small parameter, then, in the formula for the leading term of $h(t)$, one replaces $t^\al$ with $(\eps+t)^\al$. }
\end{rem}

\subsubsection{Modification SQ}\label{mod SQ}
We use the label SQ  to denote seemingly natural modification of the Adams method: due to
a special form of the integrand, $\hh(\xi,t_{k+1})$ can be calculated explicitly solving  the quadratic equation for  $\hh(\xi,t_{k+1})$. The errors of a numerical method for the solution
of an integral equation strongly depend on the regularity of the solution. 
In the BL modification, we solve the equation for the auxiliary function which is more regular
than the initial one. 
In Modification SQ, the unknown function remains the same, hence, the errors in the region of small $t$ remain large.
There are two additional disadvantages. First, the solution of the quadratic equation is more time
consuming that the simple iteration method at each step of the Adams method. For the same
$t$-grids, the CPU time is 10-20\% larger; the cumulative error of the simple iteration
with 5-10 steps is smaller than the error of the Adams method (or its BL modification) itself, hence,
a more accurate evaluation of $\hh(\xi,t_{k+1})$ brings no significant advantages.
Secondly, the expression for the square root in the formula for $h_{k+1,k+1}(\xi,t)$
may cross the cut $(-\infty,0]$ as $t$ increases and then the pricing algorithm produces an incorrect result.
Finally, in the application to evaluation of the characteristic function in the rough Heston model,
we observe that the errors of the SQ modification are larger than the errors of the Adams method
and the BL modification.

\subsubsection{Errors of the evaluation of $h(\xi,\tau)$}
We evaluate $h(\xi,\tau)$ for $\tau\le 1M$ and $\xi$ on the line $\{\Im\xi = -0.5\}$  using the standard Adams method (Adams1 in Fig.~\ref{Re_vs_BLfine_SET1_T002}-\ref{Re_Phi_vs_BL})
and the BL modification. In both cases, $M=40000$ is fairly large, and the results agree extremely well. Then we take a moderate $M=1000$
and calculate errors of the BL modification, Adams1, and SQ modification (labelled Adams2 in the figures).
It is seen that the BL modification is significantly more accurate than the standard Adams method, and
the SQ modification is significantly less accurate than the standard Adams method.

\begin{figure}
\begin{tabular}{cc}
 \begin{subfigure}[h]{0.45\textwidth}
 \centering
    \includegraphics[width=0.9\textwidth,keepaspectratio]{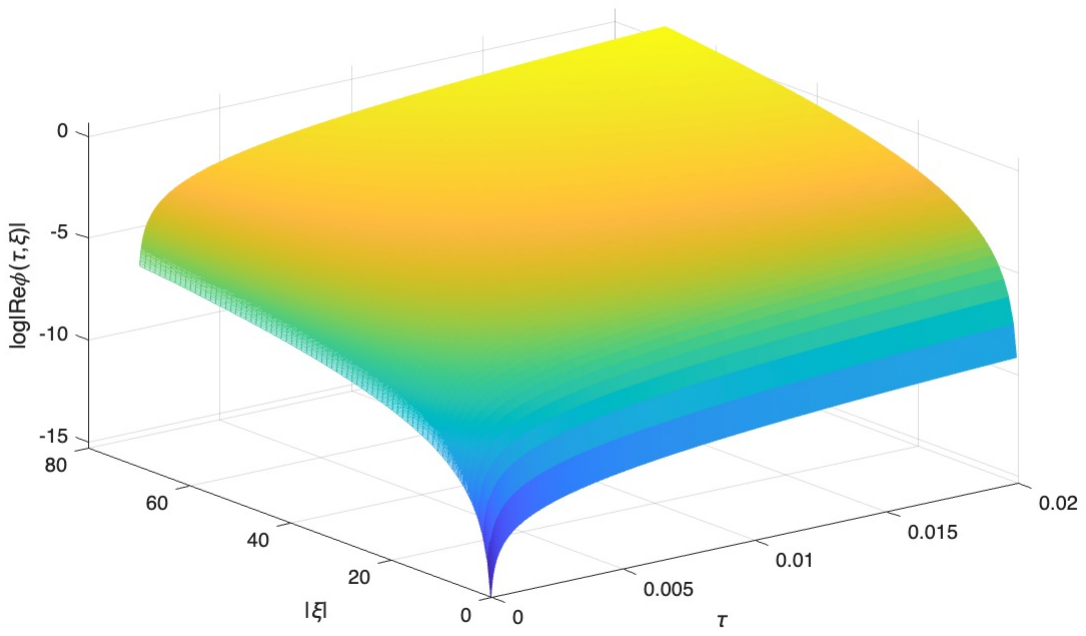}
    \vskip-2cm
    \subcaption{}
    \label{SET1_T002_ReBLfine}
\end{subfigure}
&
\begin{subfigure}[h]{0.45\textwidth}
\centering
    \includegraphics[width=0.9\textwidth,keepaspectratio]{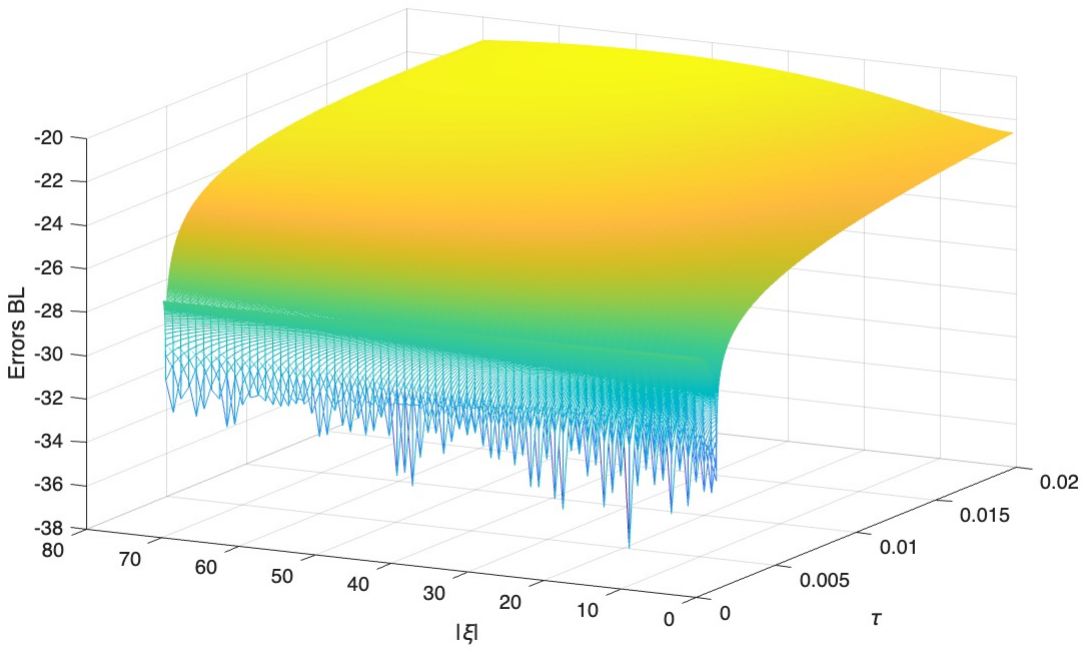} 
     \vskip-2cm
    \subcaption{} 
    \label{SET1_T002_BLfine_vs_BL.pdf}
\end{subfigure}
\\ \\
\begin{subfigure}[h]{0.45\textwidth}
 \centering
    \includegraphics[width=0.9\textwidth,keepaspectratio]{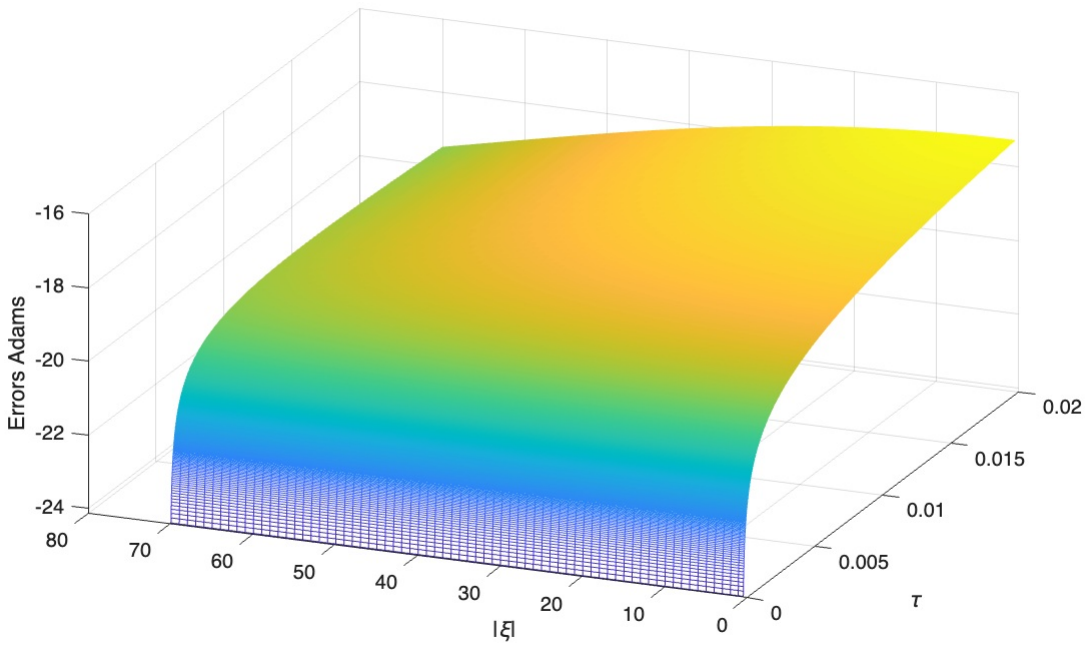}  
      \vskip-2cm  
    \subcaption{} 
    \label{SET1_T002_BLfine_vs_Adams.pdf}
\end{subfigure}
&
\begin{subfigure}[h]{0.45\textwidth}
\centering
    \includegraphics[width=0.9\textwidth,keepaspectratio]{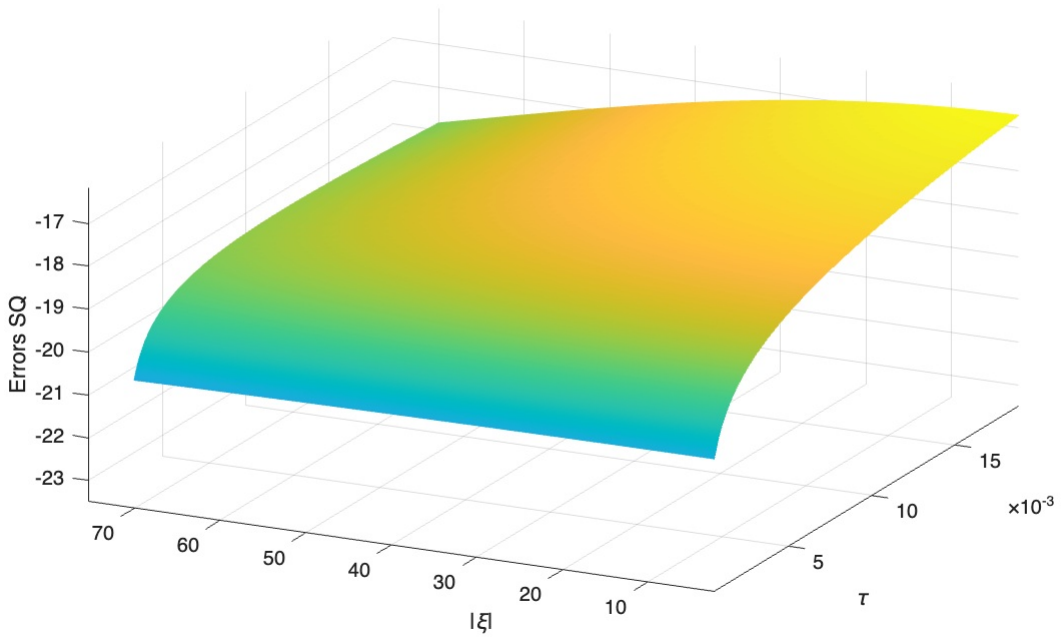}  
     \vskip-2cm 
    \subcaption{} 
    \label{SET1_T002_BLfine_vs_Adams_vs_QS}
\end{subfigure}
\end{tabular}
\caption{\small (A) Benchmark values $\Re\phi(\xi,\tau)$  are  evaluated 
using the BL modification of the Adams method with 20000 time steps. 
(B)-(D): logs of absolute values of errors  of $\Re \phi$ if $\phi(\xi,\tau)$ is  evaluated 
using the BL modification, simplified Adams method and SQ-modification, respectively, with 1000 time steps.
The number of iterations in cases (B),(C) is 5. Time to maturity $\tau\in[0.002,0.02]$, $\xi$ are on the line $\{\Im\xi=-0.5\}$.\\ 
Parameters  $(\al,\ga,\rho,\nu,\theta,v_0)=
(0.62,0.1,-0.681,0.3156,0.331,0.0392)$ }
\label{Re_vs_BLfine_SET1_T002}
\end{figure}

\begin{figure}
\begin{tabular}{cc}
 \begin{subfigure}[h]{0.45\textwidth}
 \centering
    \includegraphics[width=0.9\textwidth,keepaspectratio]{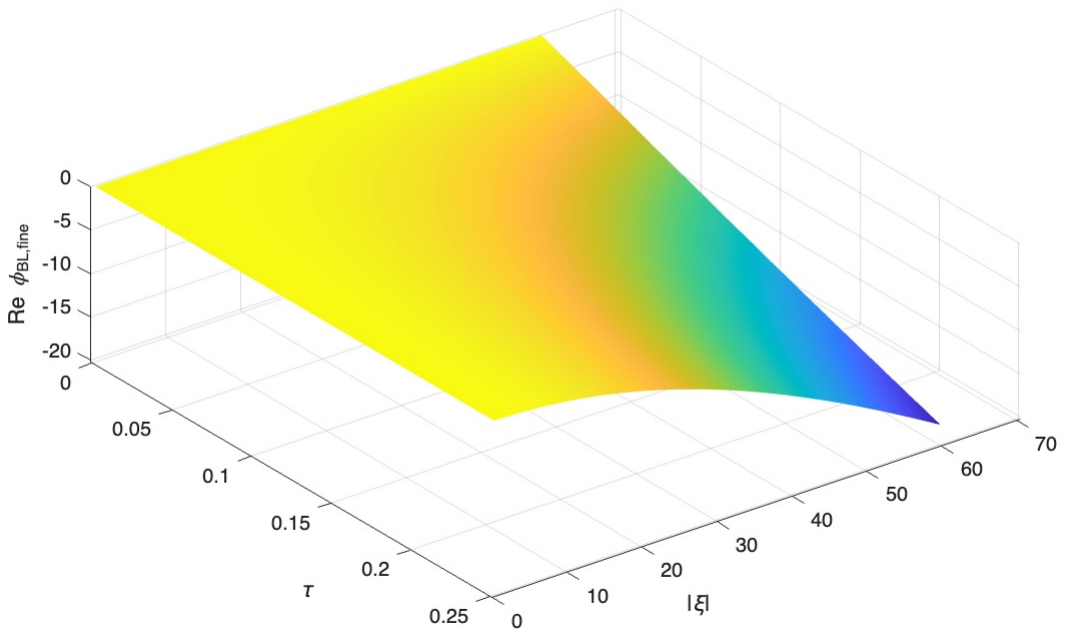}
      \vskip-2cm
    \caption{}
    \label{BLfine-1}
\end{subfigure}
&
\begin{subfigure}[h]{0.45\textwidth}
\centering
    \includegraphics[width=0.9\textwidth,keepaspectratio]{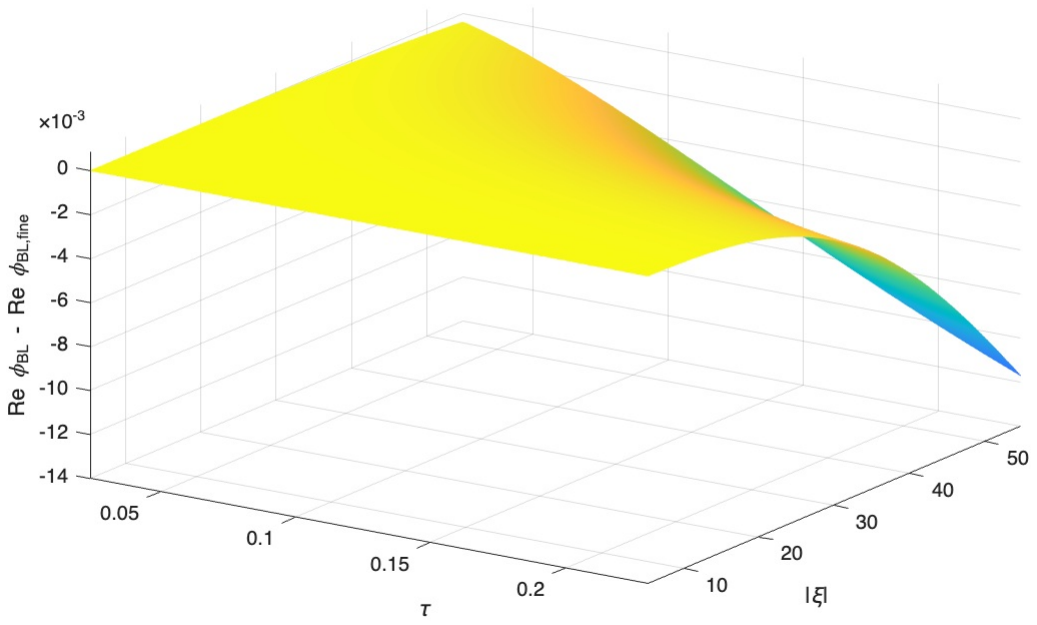} 
     \vskip-2cm
    \caption{} 
    \label{PhiBL-BLfine.pdf}
\end{subfigure}
\\ \\
\begin{subfigure}[h]{0.45\textwidth}
 \centering
    \includegraphics[width=0.9\textwidth,keepaspectratio]{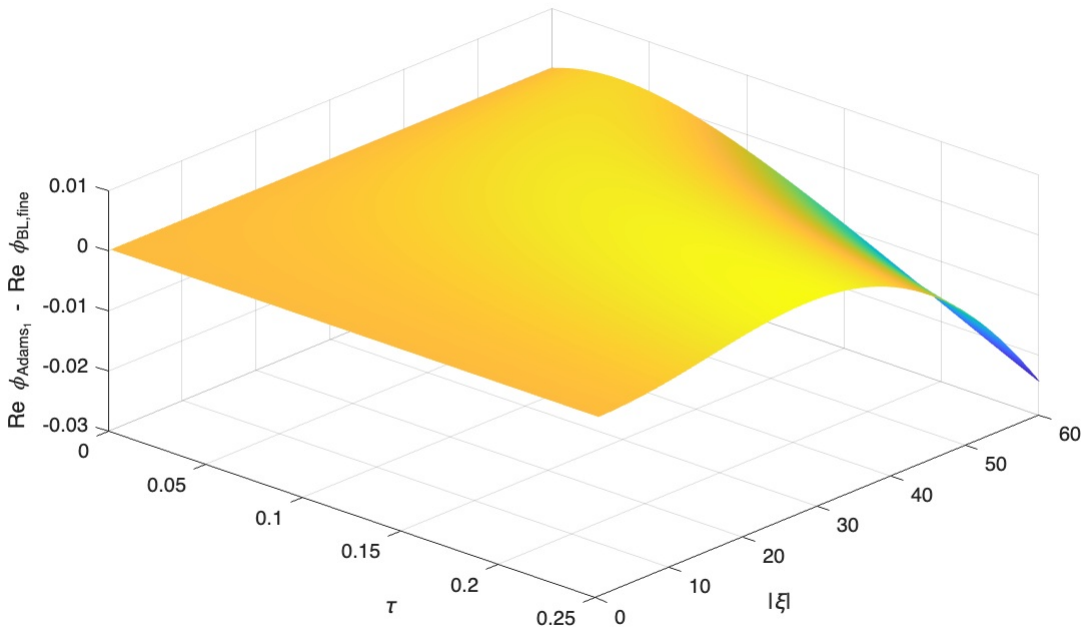}  
     \vskip-1.5cm  
    \caption{}
    \label{Adams1_vs_BLfine.pdf}
\end{subfigure}
&
\begin{subfigure}[h]{0.45\textwidth}
\centering
    \includegraphics[width=0.9\textwidth,keepaspectratio]{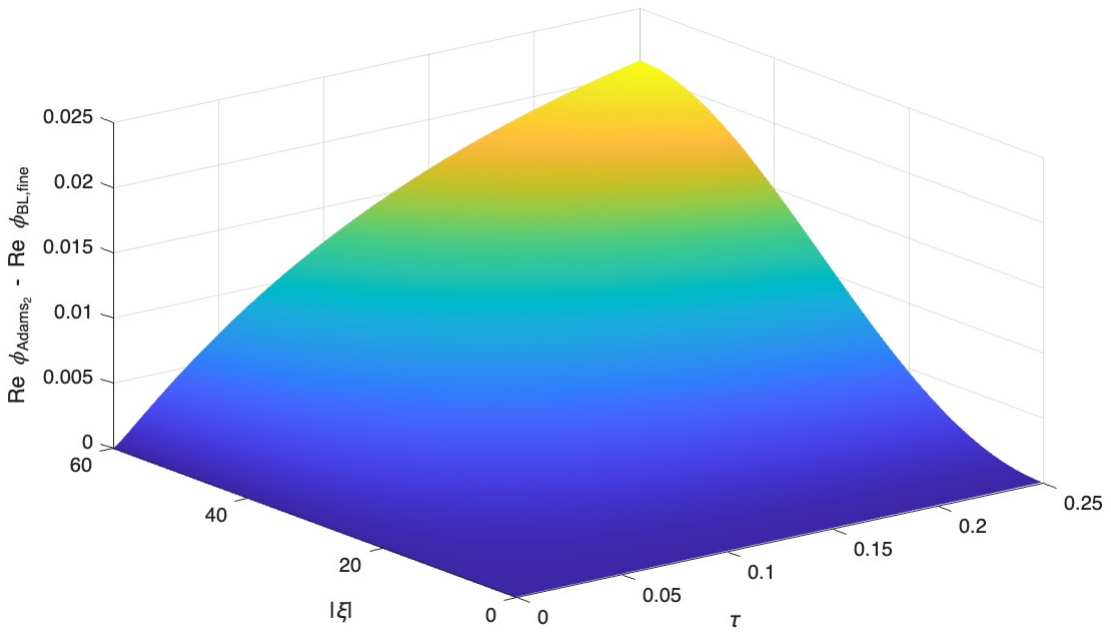}   
      \vskip-1.5cm
    \caption{} 
    \label{Adams2_vs_BL1fine}
\end{subfigure}

\end{tabular}
\caption{(A) $\Re \phi_{BL,fine}(\xi,\tau)$: $\phi(\xi,\tau)$ is  evaluated 
using the BL modification of the Adams method with 40000 time steps. 
  (B)-(D): differences $\Re\phi(\xi,\tau)-\Re \phi_{BL,fine}(\xi,\tau)$, where $\phi(\xi,\tau)$ is evaluated 
using the BL modification, Adams method and SQ modification, respectively, with 1000 time steps.
Time to maturity $\tau\in[0.0062,0.25]$, $\xi$ on the line $\{\Im\xi=-0.5\}.$\\ 
Parameters  $(\al,\ga,\rho,\nu,\theta,v_0)=
(0.62,0.1,-0.681,0.3156,0.331,0.0392)$ }
\label{Re_Phi_vs_BL}
\end{figure}

\begin{figure}
\begin{tabular}{cc}
 \begin{subfigure}[h]{0.45\textwidth}
 \centering
    \includegraphics[width=0.9\textwidth,keepaspectratio]{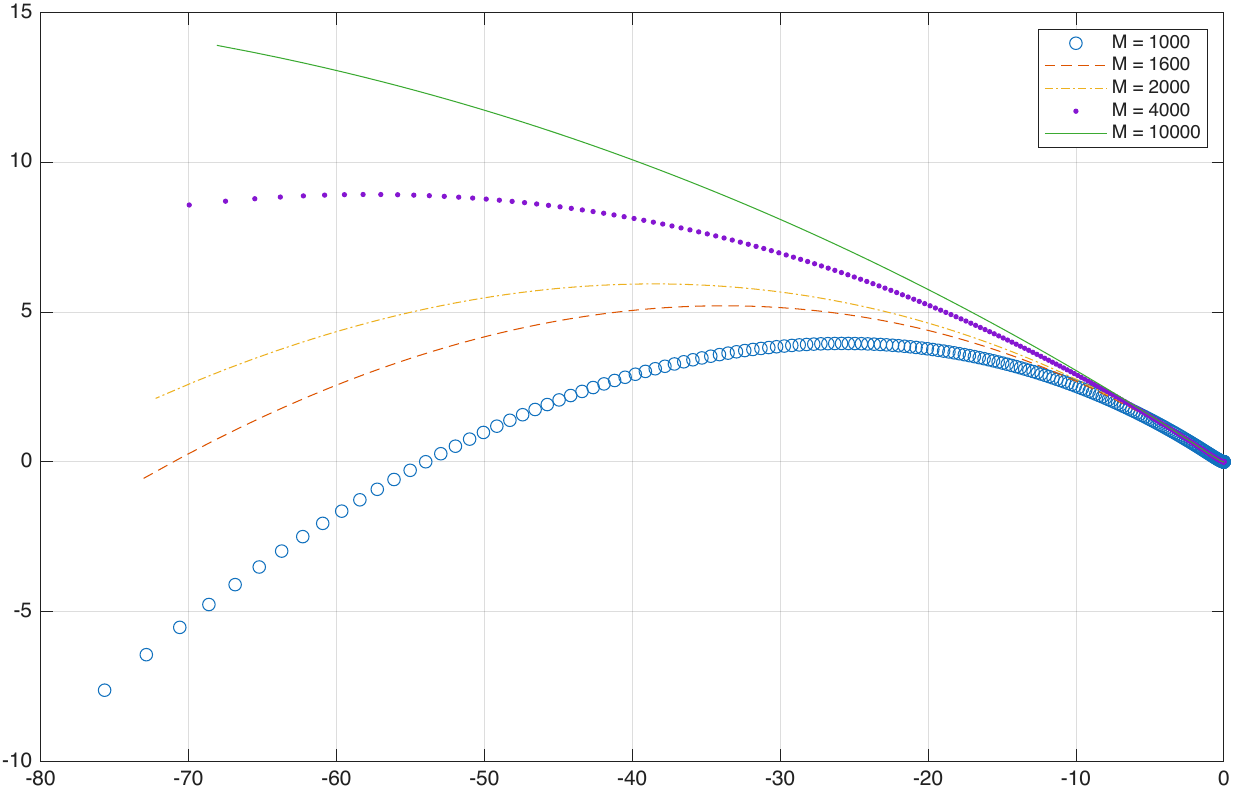}
    \label{QL200SQ}
\end{subfigure}
&
\begin{subfigure}[h]{0.45\textwidth}
\centering
    \includegraphics[width=0.9\textwidth,keepaspectratio]{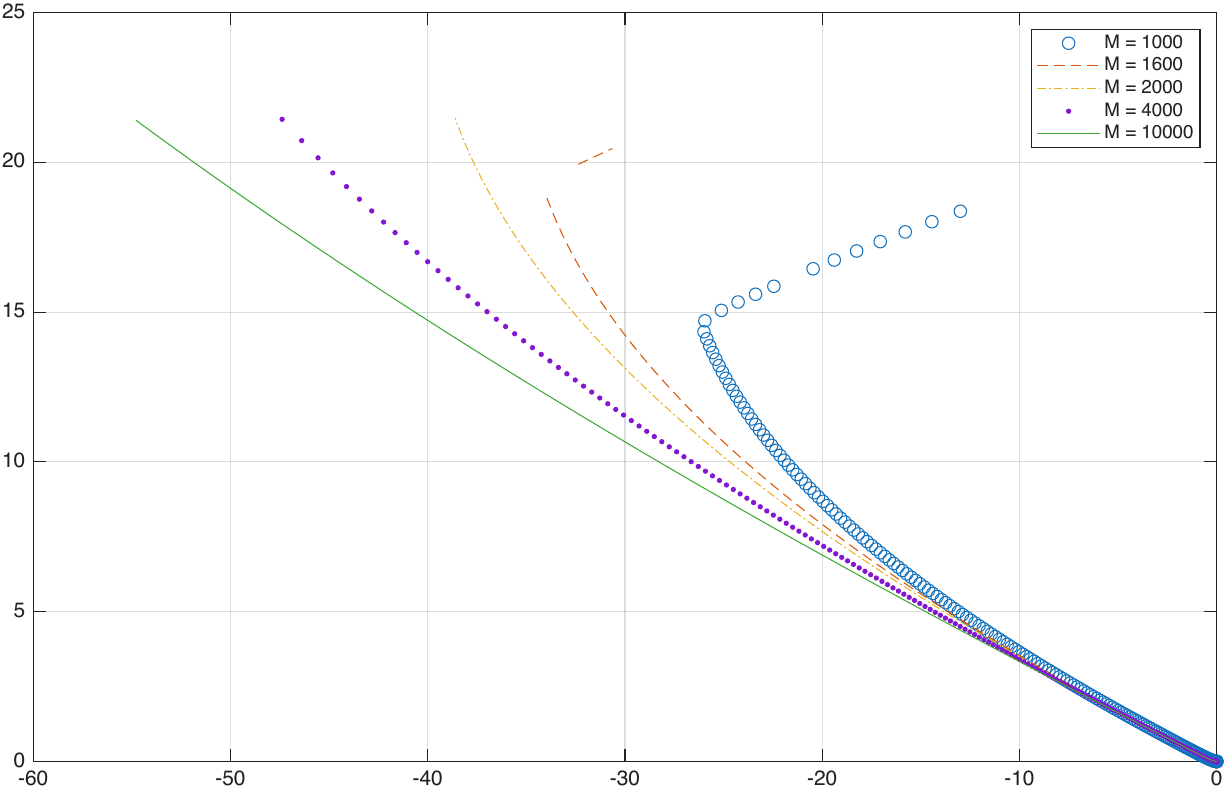} 
    \label{GapErrTSLA_1WLG200}
\end{subfigure}
\label{GL200BL}
\end{tabular}
\caption{$\phi(\tau,\xi_j)$ for $\tau=1W$ and $\xi_j=-0.5i+y_j$, where $y_j, j=1,2,\ldots, 200,$ are the nodes of the Gauss-Laguerre quadrature
of order $N=200$. Left and right panels: SQ and BL modifications are used, respectively.
Parameters  $(\alpha,\gamma,\theta,\sigma,\rho,v_0)=
(0.587271,\,3.22767,\,0.219608,\,1.49494,\,-0.310089,\,0.552303)$; $\nu=\sg/\ga$.}
\label{GL200}
\end{figure}  
In  Fig.~\ref{GL200}, we show the values of $\phi(\tau, \xi_j)$, $\tau=1/52$,
$\xi_j=-0.5i+y_j$, where $y_j$ are nodes of the Gauss-Laguerre quadrature with $200$ terms.
The curves on the left (resp., right) panel are obtained using the SQ (resp., BL) modifications.
Assuming that the conjecture about the  asymptotics of $\phi(\xi,\tau)$ as $\xi\to \infty$ is correct,
for very large $M$, the curve $(\bR_+\ni y)\to \phi(-0.5i+y, \tau)$ must have an (approximate) asymptote.
If this is the case, the curves evaluated using the QS-modification approach the asymptote from below
as $M\to\infty$, and the ones evaluated using the BL-modification - from above. One observes that 
\begin{enumerate}[(a)]
\item
the errors of the former are larger than the ones of the latter;
\item
if the SQ modification is used, $\Re\phi(-0.5i+y,\tau)\to -\infty$ as $y\to +\infty$ significantly faster than 
the correct $\Re\phi(-0.5i+y,\tau)$ which we expect to be closer to the asymptote. Hence, if $M$ is not very large, the Gauss-Laguerre quadratures with 100 nodes and more will produce seemingly rapidly converging sequence of prices;
but the errors - small but not very small - remain of the same order of magnitude;
\item
if the BL modification is used and $M$ is insufficiently large, then the sequence $\Re\phi(-0.5i+y_j,\tau)$
may start to increase, hence, the errors in prices increase and become quite substantial  as $N$ increases.
This explains the calibration/pricing errors in the example shown on \cite[Fig.~7]{RoughHestonWeMarco2025}, where the Gauss-Laguerre
quadrature with 200 nodes and BL-modification with $M=1000$ are used.
\end{enumerate}

\subsection{Markovian approximation of rough volatility}\label{sss:markov_approx_desc}

The rough Heston model is characterized by a fractional integral in the dynamics of the volatility process. The key idea of the Markovian approximation is to replace the non-Markovian fractional process with a high-dimensional Markovian process that has similar dynamics \cite{MarkovianGG}. One observes that the fractional kernel of the rough Heston model $K(t) = c_{\alpha} t^{\alpha-1}$ is a complete monotone function, hence, by Bernstein's theorem, can be represented as an integral of exponential functions:
\[
K(t) = \int_0^\infty e^{-yt} \mu(dy),
\]
where $\mu(dy)$ is a positive measure. One approximates $\mu(dy)$ by a weighted sum of
atoms, thereby approximating $K(t)$ by a finite weighted sum of exponentials:
\begin{equation}\label{e:markov-expansion}
K^n(t) = \sum_{j=1}^n c_j e^{-y_j t}.
\end{equation}
The weights $c_i$ and exponents $y_i$ are chosen in order to match the original kernel $K(t)$ as closely as possible. 
This transforms the original non-Markovian rough Heston model into a higher-dimensional, but Markovian, system which is more amenable to standard pricing techniques. Following \cite[\S 1]{MarkovianGG}, we can define
the approximation $(S^n, V^n)$ of $(S, V)$ as follows
\begin{align}
	d S^n_t &= S^n_t \sqrt{V^n_t} dB_t ,\qquad S^n_0 = S_0,\\
	V^n_t &= V_0 + \int_0^t K^n(t-s)\, \gamma (\theta -  V^n_s) \, d s + \int_0^t K^n(t-s)\gamma \nu\sqrt{V^n_s} \, d W_s.
\end{align}
In \cite{AbiJaberEuch2019} it was shown that $V^n$ solves an $n$-dimensional SDE. In \cite{MarkovianGG}, different schemes for the approximations $K^n$ are compared, including those proposed in earlier studies, e.g. in \cite{AbiJaberEuch2019}, \cite{AlfonsiKebaier2024}, \cite{MarkovianBB2021}. Of these, the scheme called ``BL2'' was shown in \cite{MarkovianGG} to be the fastest and most accurate. This rule minimizes the $L^2$ error between $K^N$ and $K$, while penalizing large nodes. The corresponding algorithm can be found in Appendix F of \cite{MarkovianGG}. 
\sbr
The  SINH-CB method can be applied in this setting to price not only European options, but barrier and
lookback options as well, after the approximation by a regime-switching L\'evy model or time discretization, and using
the schemes in \cite{EfficientLevyExtremum,EfficientDoubleBarrier2,AltFX2,EfficientDiscExtremum2025}.

\section{Flat iFT, SINH-acceleration and Conformal Bootstrap principle}\label{s:SINH-CB}

\subsection{Flat iFT and simplified trapezoid rule}\label{ss:Flat iFT and simpl. trap}
 In popular models, the (conditional) characteristic function admits  analytic continuation to a strip around the real axis.
This implies that 
the following scheme (standard from the viewpoint of  Analysis) suggested in \cite{BL-FT,genBS,KoBoL} is more efficient than the scheme in 
\cite{heston-model} based on the L\'evy inversion formula.
Let the riskless rate $r\ge 0$ be constant, and let $S_T=S_0e^{X_T}$ be the price of the underlying non-dividend
paying asset (or index) at time $T$. Let $\Phi(\xi,T)=\bE[e^{i\xi X_T}] $ be the characteristic function of $X_T$ under a no-arbitrage measure $\bQ$ chosen for pricing (the expectation is conditioned on the spot values of additional factors as in SV models).
Then $\Phi(0,T)=1$, and if $\bE^\bQ[e^{X_T}]<\infty$, $\Phi(-i,T)=e^{rT}$. 
Assume
that there exist $\mum(T)<-1<0<\mup(T)$ s.t. for $\be\in (-\mup(T),-\mum(T))$, the exponential
moments $\bE^\bQ[e^{\be X_T}]$ are finite. 
Equivalently, $\Phi(\xi,T)$ admits analytic continuation to a strip 
$S_{(\mum(T),\mup(T))}:=\{\xi\ |\ \Im\xi\in (\mum(T),\mup(T))\}$.  Then the price of the call option
with strike $K$ and maturity $T$ can be calculated as follows. 
The payoff function $G(S_0,K,x)=(S_0e^x-K)_+$ admits a representation
\bbe\label{hG}
G(S_0,K; x)=\frac{1}{2\pi}\int_{\Im\xi=\om_1}e^{ix\xi}\hG(S_0,K;\xi)d\xi,
\ee
where  $\om_1\in (\mum(T),-1)$ is arbitrary, and $\hG(S_0,K;\xi)=-Ke^{i\xi\ln (S_0/K)}/(\xi(\xi+i))$ is the Fourier transform of
$G(S_0,K; x)$ w.r.t. $x$. We  substitute the integral representation \eq{hG} of $G(S_0,K; X_T)$ into the pricing formula $V(S_0, K;T)=e^{-rT}\bE[(S_0e^{X_T}-K)_+]$, and change the order of integration and summation
(the use of the Fubini theorem can be justified in all popular models). The result is 
\bbe\label{EuroPrice}
V(S_0,K;T)=-\frac{Ke^{-rT}}{\pi}\Re\int_{\Im\xi=\om_1}\frac{e^{i\xi\ln (S_0/K)}\Phi(T,\xi)}{\xi(\xi+i)}d\xi.
\ee
Similarly, the price of the put is given by the RHS of \eq{EuroPrice} with arbitrary $\om_1\in (0,\mup(T))$
(repeat the proof for the call starting with $G(x)=(K-S_0e^x)_+$ or use the put-call parity on the LHS of \eq{EuroPrice} and
the residue theorem on the RHS). The price of the covered call is given by the RHS of \eq{EuroPrice} with $\om_1\in (-1,0)$.
Since $\overline{\Phi(\xi,T)}=\Phi(-\bar{\xi},T)$ and $\overline{\hG(\xi)}=\hG(-\bar\xi)$, an equivalent form of
\eq{EuroPrice} is
\bbe\label{EuroCallSym}
V(S_0,K;T)=-\frac{Ke^{-rT}}{\pi}\Re\int_{\Im\xi=\om_1}\frac{e^{i\xi\ln (S_0/K)}\Phi(\xi,T)}{\xi(\xi+i)}d\xi.
\ee
After truncation, the integral on the RHS of \eq{EuroPrice} (or \eq{EuroCallSym}) can be calculated using either trapezoid rule or Simpson rule.

However, since the integrand on the RHS of \eq{EuroPrice} is analytic in a strip
$S_{(\lm,\lp)}$ around the line of integration ($\lm=\mum(T), \lp=-1$ in the case of calls, $\lm=-1, \lp=0$ in the case of the covered call,
and $\lm=0, \lp=\mup(T)$ in the case of puts),
it is significantly more efficient to use the infinite trapezoid rule and then truncate the sum. The reason is
an exponential decay of the discretization error of the infinite trapezoid rule as the function of $\ze$, where $\ze$ is
the step. In Mathematical Finance, Lee \cite{Lee04} and Feng and Linetsky \cite{feng-linetsky08} were the first  to use this important property of the infinite trapezoid rule; the truncation of the sum results in
the simplified trapezoid rule. As it is stated in the review paper \cite{TrefethenWeidmanTrapezoid14}, the excellent properties of the simplified trapezoid rule had been noticed since Poisson but rigorously proved in the middle of the last century only.  
Let $H^1(S_{(\lm,\lp)})$ denote the space of functions analytic in the strip $S_{(\lm,\lp)}$ such that
\[
\int_{\lm}^{\lp} |f(\eta+i\om)|d\om \to 0\quad {\rm as}\ (\bR\ni)\eta\to\pm\infty\]
and the following analog of the Hardy norm 
is finite:
\begin{equation}\label{Hardynorm}
||f||_{S_{(\lm,\lp)}}:=\lim_{\om\uparrow \lp}\int_\bR|f(\eta+i\om)|d\eta+
\lim_{\om\downarrow \lm}\int_\bR|f(\eta+i\om)|d\eta<\infty.
\end{equation}
Fix $\om_1\in (\lm, \lp)$, and denote $d(\om_1)=\min\{\om_1-\lm, \lp-\om_1\}$. For $\ze>0$, construct a grid $\xi=i\om_1+\ze\bZ$,
and denote by $E_{\rm disc}(\zeta, \infty)$ the error of the infinite trapezoid rule
\[
\int_{\Im\xi=\om}f(\xi)d\xi\approx \ze\sum_{j\in\bZ}f(\xi_j).
\]
The following bound is proved in \cite{stenger-book} using the heavy machinery of the
sinc-functions (a simple proof can be found in \cite{paraHeston}, and several other elementary
bounds and proofs in \cite{TrefethenWeidmanTrapezoid14}):
\begin{equation}\label{discerrbound}
 \left|E_{\rm disc}(\zeta, \infty)\right|\le \frac{e^{-2\pi d(\om_1)/\zeta}}{1-e^{-2\pi d(\om_1)/\zeta}}||f||_{S_{(\lm,\lp)}}.
 \end{equation}
Let the error tolerance $\eps>0$ for the discretization error be small, and let $|\mu_\pm|$ be not too large.
 Then we 
choose $\om_1=(\lm+\lp)/2$, set $d(\om_1)=k_d(\lp-\lm)/2$, where $k_d<1$ is close to 1, e.g., $k_d=0.95$,
and use the following approximate recommendation: 
\bbe\label{ze_inf}
\ze=2d(\om_1)/\ln(100/\eps).
\ee
If the strip of analyticity is very wide, we choose a substrip around the line of integration with moderately large $|\la_\pm|$ and apply the prescription above. 

Once $\ze$ is chosen and the sum is truncated, we have the pricing formula. In the case of \eq{EuroPrice}, 
\bbe\label{EuroCallSimp}
V(S_0,K; T)=-\frac{Ke^{-rT}\ze}{2\pi}\sum_{|j|\le N}\frac{e^{i\xi_j\ln (S_0/K)}\Phi(\xi_j,T)}{\xi_j(\xi_j+i)},
\ee
where $\xi_j=i\om_1+j\ze$. The number of terms can be decreased almost two-fold: similarly
to  \eq{EuroCallSym},
\bbe\label{EuroCallSimpSym}
V(S_0,K; T)=-\frac{Ke^{-rT}\ze}{\pi}\Re\sum_{0\le j\le  N}(1-\de_{j0}/2)\frac{e^{i\xi_j\ln (S_0/K)}\Phi(\xi_j,T)}{\xi_j(\xi_j+i)},
\ee
where $\de_{jk}$ is the Kronecker symbol. We call this method Flat iFT (flat inverse Fourier transform) method.
To choose $N$ so  that the truncation error is sufficiently small, it is necessary to know the rate of decay
of $\Phi(\xi,T)$ as $\xi\to\infty$ along the contour of integration. 
Let $\Phi(\xi,T)=\exp[\phi(\xi,T)]$, 
 and let 
an upper bound for $\Re \phi(\xi,T)$ be known:
\bbe\label{upperphi0}
\Re \phi(\xi,T)<-g(|\xi|,T),
\ee
where $g(|\xi|,T)$ is a monotonically increasing function of $|\xi|$.
Then the truncation of the series
at $|\xi|=\La_0$ introduces the error of the order of  $e^{-g(\La_0,T)}/\La_0$.
If an analytic formula for $\phi(\xi,T)=\ln\Phi(\xi,T)$ is available,
then an efficient bound \eq{upperphi0} can be derived. See \cite{paraHeston,pitfalls}. In the case of the rough Heston model, an analytic formula is not available. In Sect.~\ref{ss:asymp_psi}, we use an informally proved asymptotic formula 
(see Sect.~\ref{AsymRough}) to formulate a prescription for the choice of $\La_0$.  In the case of Flat iFT, the asymptotic formula and the prescription 
should be used with $\om=0$. After $\La_0$ is calculated, we
set $N=\mathrm {ceil}\,\La_0/\ze$.


\subsection{SINH-acceleration}\label{ss: SINH}
The SINH-acceleration is applicable if the (conditional on the spot value of the underlying
and additional factors, as in the case of SV models) characteristic function
$\Phi(\xi,T)$ of $X_T$ admits analytic  continuation to the union of a strip 
$S_{(a,b)}:=\{\xi\ | \ \Im\xi\in (a,b)\}$, where $a\le -1<0\le b$, and a cone around $-0.5 i+\bR$.
For the Heston model and rough Heston model, conditions on the parameters of the model and $T$ which guarantee 
analyticity of $\Phi(\xi,T)$ in the strip $S_{[-1,0]}$ are known \cite{Lee04,paraHeston,GerholdGersteneckerPinter2019}. If $\Phi(\xi,T)$ admits analytic continuation to $S_{[-1,0]}$, then the existence of a cone of analyticity around $-0.5 i+\bR$ is proved in \cite{paraHeston}
for the Heston model, and in \cite{pitfalls} for wide classes of SV models and models with stochastic interest rates. We surmise that the (conditional) characteristic function in the rough Heston model
enjoys the same property but we were unable to prove this fact. In Sect.~\ref{ss:CP principle}, we explain how the CB principle
can be used to heuristically establish this property for any set of the parameters of the rough Heston model.


In the Heston model and rough Heston model, the integral \eq{EuroPrice} can be calculated sufficiently accurately for applications using the Flat iFT method with 200-400 terms if the time to maturity is not too short. The errors depend on the choice of the line of integration. In the real-analytic interpretation \cite{carr-madan-FFT},  choices of different lines
of integration are choices of different {\em dampening factors}. In Complex Analysis, one observes that the Fourier transform $\hf$ of a sufficiently regular function $f$ is an analytic function in a wide region $\cU_0$ of the complex plane and meromorphic function in a wider 
region $\cU$.  We choose $\cU$ so that $\hf(\xi)\to 0$ sufficiently fast as 
$\xi\to \infty$ remaining in $\cU$. The inverse Fourier transform can be calculated deforming the line of integration into any sufficiently regular curve in $\cU_0$; crossing poles, one can reduce
to the integral over any sufficiently regular curve in $\cU$ (plus residues at the poles crossed
in the process of deformation).

In all popular models bar stable  L\'evy models different from BM, $\Phi(\xi,T)$ admits analytic continuation to a region of the form
 $\cU(\gam,\gap; \mup, \mum):= i(\mup,\mum)+(\cC_{\gam,\gap}\cup(-\cC_{-\gap,\gam})\cup\{0\})$,
where $\mum<-1<0<\mup$, $\cC_{\gam, \gap}:=\{\xi\in \bC\ | \mathrm{arg}\,\xi\in (\gam,\gap)\}$.
In the case of L\'evy models, $\cU(\gam,\gap; \mup, \mum)$ is independent of $T$. In the Heston model and other SV models, the domain of analyticity depends on $T$. In the case of the Heston model, under additional restriction on
the parameters, it is  proved in \cite{Lucic}  
that $\Phi(\xi,T)$ is analytic in $\cU=\bC\setminus i((-\infty,\mum(T)]\cup [\mup(T),+\infty))$, where $\mum(T)<-1<0<\mup(T)$; in  \cite{paraHeston}, this fact is proved for jump-diffusion generalizations of the Heston model, with more than one factor driving the dynamics of the volatility process,
and algebraic equations for $\mum(T)$ and $\mup(T)$ were derived. 
For wide classes of affine jump-diffusion processes, it is proved in \cite{pitfalls} that $\Phi(\xi,T)$ is an analytic 
function on the union  $\cU_0(\mum(T),\mup(T),\gam,\gap)$ of a strip
$S_{(\mum(T),\mup(T))}$, where $\mum(T)<-1<0<\mup(T)$, and a cone
$\cC_{\gam,\gap}:=\{\xi=\rho e^{i\varphi}\ |\ \varphi\in (\gam,\gap) \vee 
 \varphi\in (\pi-\gam, \pi-\gap) \}$, where $\gam\in (-\pi/2,0), \gap\in (0,\pi/2)$
 (typically, $\ga_\pm=\pm\pi/4$; in the case of the Heston model, $\ga_\pm=\pm\pi/2$), and decays as $\xi\to \infty$ remaining in the cone. 
 Once the existence of such a strip and cone is established, we choose a deformation of the contour of integration  into a contour $\cL_{\om_1,b,\om}:=\chi_{\om_1,b,\om}(\bR)$, where $\om_1\in \bR$, $b>0$, $\om\in (\gam,\gap)$,
 and the conformal map $\chi_{\om_1,b,\om}$
 ({\em sinh-deformation}) is defined by 
 \bbe\label{eq:sinh}
 \chi_{\om_1,b,\om}(y)=i\om_1+b\sinh(i\om+y).
 \ee
 The parameters of the deformation
 are chosen so that in the process of deformation, the contour remains in  $\cU(\gam,\gap; \mup, \mum)$,
 the oscillating factor becomes a fast decreasing one and the poles at $\xi=0, -i$ are not crossed.  
 The deformation being made,
 we change the variable $\xi=\xi(y)=\chi_{\om_1,b,\om}(y)$ in \eq{EuroPrice}
 \bbe\label{EuroPricesinh}
V(S_0, K;T) = -\frac{bKe^{-rT}}{2\pi}\int_{\bR}\frac{e^{i\xi(y)\ln(S_0/K)}\Phi(\xi(y),T)}{\xi(y)(\xi(y)+i)}\cosh(i\om+y)dy,
\ee
and apply the simplified trapezoid rule: 
\bbe\label{EuroPricesinhtrap}
V(S_0, K;T) = -\frac{b\ze Ke^{-rT}}{\pi}\Re\sum_{j=0}^N e^{i\xi(j\ze)\ln(S_0/K)}g(j\ze,T)(1-\de_{0j}/2),
\ee
where
$
g(y,T)=\frac{\Phi(\xi(y),T)}{\xi(y)(\xi(y)+i)}\cosh(i\om+y)$. 

The change of variables $\xi=\chi_{\om_1,b,\om}(y)$ makes the integrand a fast decaying one in a strip of the form $S_{(-d,d)}$, hence, the error of the truncation of the infinite sum is easy to control.
 Fig.\ref{fig:two deformations} illustrates the sinh-acceleration. 
 In essentially all cases we tried, an absolute error tolerance of the order of $E-09$ can be satisfied with 20-60 terms of the simplified trapezoid rule. 

\begin{figure}
\scalebox{1.}
{\includegraphics{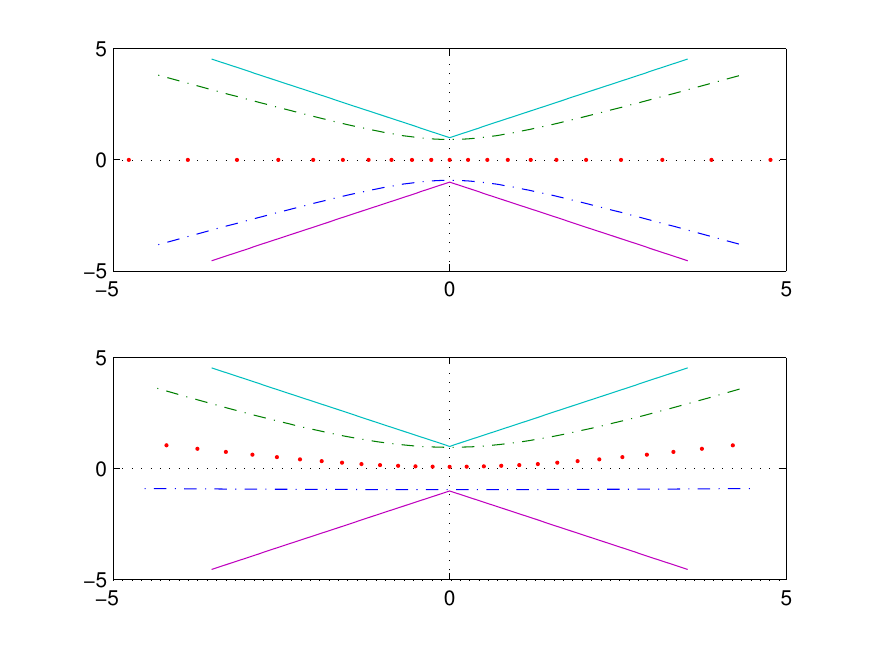}}

\caption{\small Solid lines: boundaries of the domain of analyticity $S_{(-1,1)}+\cC_{-\pi/4,\pi/4}$
in the $\xi$-coordinate.
Dots: points $\xi_j=\chi_{\om_1,\om;b}(y_j)=i\om_1+b\sinh(i\om +y_j)$ used in the simplified trapezoid rule.
Dots-dashes: boundaries of the image $\chi_{\om_1,\om;b}(S_{(-d,d)})$ of the strip of analyticity $S_{(-d,d)}$.
 Upper panel: $\om_1=\om=0$, $d=\pi/4$, $b=1/\sin(\pi/4)$. Lower panel: $\om_1=-1, \om=d=\pi/8$, $b=2/\sin(\pi/8)$.
 For the calculations represented in the lower panel, only a smaller domain $S_{(-1,1)}+\cC_{0,\pi/4}$ matters.}
\label{fig:two deformations}
\end{figure}  
Explicit recommendations for the choice
of the parameters of the deformation $\om_1,b,\om$ and parameters $\ze,N$ of the simplified trapezoid rule
are derived in \cite{SINHregular}. We add several useful details for the case of the rough Heston model.
\begin{enumerate}[I.]
\item
Find $\mu_\pm(T)$ and $\ga_\pm$. 
\item
Calculate $z_T$ using \eqref{zT} below. If $z_T\le 0$, use $\om\le 0$ and calculate the price of either the call or covered call; 
otherwise, use $\om\ge 0$ and  calculate the price of either the put or covered call.

\item
\begin{enumerate}[(a)]
\item
 If the call is priced, set $\lm=\mum(T), \lp=-1$, $\om=\gam/2$, $d_0=-\om$.
\item
If the put is priced, set $\lm=0, \lp=\mup(T)$, 
$\om=\gap/2$, $d_0=\om$.
\item
If the covered call is priced, set $\lm=-1,\lp=0$. If $S_0< K$, set
$\om=\gam(T)/2$, $d_0=-\om$.
If $S_0> K$, set
$\om=\gap(T)/2$, $d_0=\om$.
\item
For ATM options, it is optimal to set $\om=(\gam+\gap)/2$, $d_0=(\gap-\gam)/2$.
\end{enumerate}
\item
Choose $k_d<1$ close to 1, e.g., $k_d=0.9$, and set $d=k_dd_0$, $\ze=2\pi d/\ln(100/\eps)$,
\[
b=\frac{\lp-\lm}{\sin(\om+d)-\sin(\om-d)},\ \om_1=\frac{\lm\sin(\om+d)-\lp\sin(\om-d)}{\sin(\om+d)-\sin(\om-d)}.
\]

\item
As in the case of Flat iFT, to choose $N$ so  that the truncation error is sufficiently small, it is necessary to know the rate of decay
of $\Phi(\xi,T)$ as $\xi\to\infty$ along the contour of integration. 
Let $\Phi(\xi,T)=\exp[\phi(\xi,T)]$, and let 
an upper bound \eq{upperphi0} for $\Re \phi(\xi,T)$ be known.
 In the $y$-coordinate, the series decays as
$(K\ze b/\pi)e^{-g(|\xi(y_j)|,T)}/|\xi(y_j)|$. Since $|\xi(y)|$ increases as an exponential function of $y$ as $y\to\pm\infty$,  the truncation error is smaller than the last term of the truncated sum if $g(|\xi(y_j)|,T)$ is large.
We find the positive solution $\La_0$ of the equation $
e^{-g(\La_0,T)}/\La_0 = b\pi\eps/(K\ze)$, and set $\La=\ln(2\La_0/(Kb))$,  $N=\mathrm {ceil}\,\La/\ze$.
\end{enumerate}
\begin{rem}\label{rem:choice of ze}{\em The recommendation $\ze=2\pi d/\ln(100/\eps)$
 presumes that $||f||_{S_{(\lm,\lp)}}$, the analogue \eq{Hardynorm} of the Hardy norm of the integrand, is bounded by 100.  A safer alternative which we used in several publications is to use the approximation 
$||f||_{S_{(\lm,\lp)}}\approx |f(i(\om+d))|+|f(i(\om-d)))|$.
}
\end{rem}

\subsection{Ad-hoc bound for $\phi$ and choice of $N$ in the rough Heston model}\label{ss:asymp_psi}
To choose $\La:=N\ze$, we use \eq{eq:asphi}-\eq{cinf}.
The leading term of asymptotics of the expression under the exponential sign in the pricing formula is
\[
(-c_\infty(T) + i\ln(S_0/K))\xi = i z_T\xi - \Re c_\infty\xi,
\]
where 
\beqa\label{zT}
z_T&=& \ln (S_0/K)-\frac{\theta\rho}{\nu}T-\frac{v_0\rho}{\Ga(2-\al)}T^{1-\al},\\\label{Recinf}
\Re c_\infty&=&\left(\ga\theta T+ \frac{v_0 T^{1-\al}}{\Ga(2-\al)}\right)\frac{\sqrt{1-\rho^2}}{\ga\nu}.
\eqa
If we use the sinh-deformation with the parameters $\om_1,b,\om$, then, as $\xi\to\infty$ in the right half-plane
along the contour $\cL_{\om_1,b,\om}$, the absolute value of the integrand admits a bound via
$
H |\xi|^{-2}\exp(-c_\infty(\om)|\xi|)$, where $H$ is a constant, and 
\bbe\label{cinfom}
c_\infty(\om) = z_T\sin(\om)+\Re c_\infty \cos(\om).
\ee
Therefore, for a given error tolerance, 
an approximately optimal $\om$ is found as the maximizer of  $c_\infty(\om)$,
and then $\La$ is chosen solving (approximately, because high accuracy is unnecessary)
the equation
\[
H \exp[-(c_\infty(\om)b/2)e^\La]= \eps,
\]
which gives
\bbe\label{eq:La}
\La= \ln[2\ln(H/\eps)/(bc_\infty(\om))]
\ee
and $N= \mathrm{ceil}\,\La/\eps$.  If $\tau$ is very small or Flat iFT is used, this prescription
results in an unnecessary large $\La$ and $N$. Then $\La$ can be decreased solving
approximately the equation $
H \exp[-(c_\infty(\om)b/2)e^\La]/\La= \eps.
$ We find (an approximation to) $\La_1=e^\La$ solving the equation
\[
\La_1 = \frac{2}{bc_\infty(\om)}(\ln \La_1 + E),
\]
where $E=\ln (H/\eps)$.
The following approximation suffices: 
\[
\La_{10}:= 2E/(bc_\infty(\om)), \ \La_1 = \frac{2}{bc_\infty(\om)}(\ln \La_{10} + E), \La_1:=\max\{1.2, \La_1\}.
\]
Then we set $\La = \log(\La_1)$, $N=\mathrm{ceil}\, \La/\ze$.

\subsection{Conformal bootstrap principle}\label{ss:CP principle}
For the error control, we calculate the prices using two deformations.
  The probability (in the colloquial sense) that the difference of the two weighted sums of the values of the integrand calculated at different nodes on different curves is significantly larger than the difference of either sum and the true price is essentially zero.
This is an ad-hoc principle which we call Conformal Bootstrap principle. 
Fig.~\ref{contours-integrands} illustrates the idea behind the Conformal Bootstrap principle.
The graphs of integrands shown on the right panel 
 are real parts of the same analytic function on different contours shown on the left panel and the integrals must be
 equal if the numerical quadrature is ideal. However, if the errors of the quadrature
 do not define an analytic function, the quadrature ``treats"  integrands as functions
 with different analytic properties, hence, the probability (understood in the colloquial sense) that the errors of integration agree
 with accuracy $\eps=10^{-m}$ but the correct value of the integral differs from either of 4 values
 by more than $10^{-m+2}$ is negligible. The SINH-algorithm uses grids with spacing and number of nodes that depend on the parameters of deformation, and the integrands at nodes are calculated using
 the BL modification which does not define an analytic function. Additional advantage
 of the simplified trapezoid rule and sinh-accelerations: 
 \begin{enumerate}[1)]
 \item the discretization error decays
 exponentially as the function of $1/\ze$. Hence, if the difference of results obtained with
  $\ze$ and $\ze/1.2$ is of the order $\eps$, the discretization error remains not larger if $\ze$
  is decreased further. If one of the contours is close to a singularity
  and $\ze$ is insufficiently small, this agreement cannot be observed; 
  \item the truncation error decreases faster than exponentially or exponentially
  if the initial integrand decreases only as a power of $|\xi|$. Hence, if the difference of results obtained with
  $N$ and $1.2N$ is of the order $\eps$, the truncation error remains not larger if $N$
  is increased further.
  \end{enumerate}
  
\begin{figure}
 \caption{Contours of integration and graphs of $\Re I(y)$, where $I(y)=\hV(\tau, \chi(y))$ and $\chi(y)=i\om_1+b\sinh(i\om+y)$, as functions
of $y\in \bR_+$. Parameters of the rough Heston model: $(\al,\ga,\rho,\nu,\theta,v_0)=
(0.62,0.1,-0.681,0.3156,0.331,0.0392)$. Time to maturity $T=1/252$. Number of time steps $M=1000$ in all 4 cases.
Parameters of SINH-acceleration (rounded): \\
(1) $\om_1=-2.056, b=5.591, \om = 0.1$; 
(2)	$\om_1=-0.0230, b=	2.795, \om = 0.1$;\\
(3) $\om = -2.057, b=2.850, \om=0.2$; (4) $\om_1=-0.129, b=1.309,
\om=	0.3.$
  }
\begin{tabular}{cc}
 \begin{subfigure}[h]{0.45\textwidth}
 \centering
    \includegraphics[width=0.9\textwidth,height=0.3\textheight]{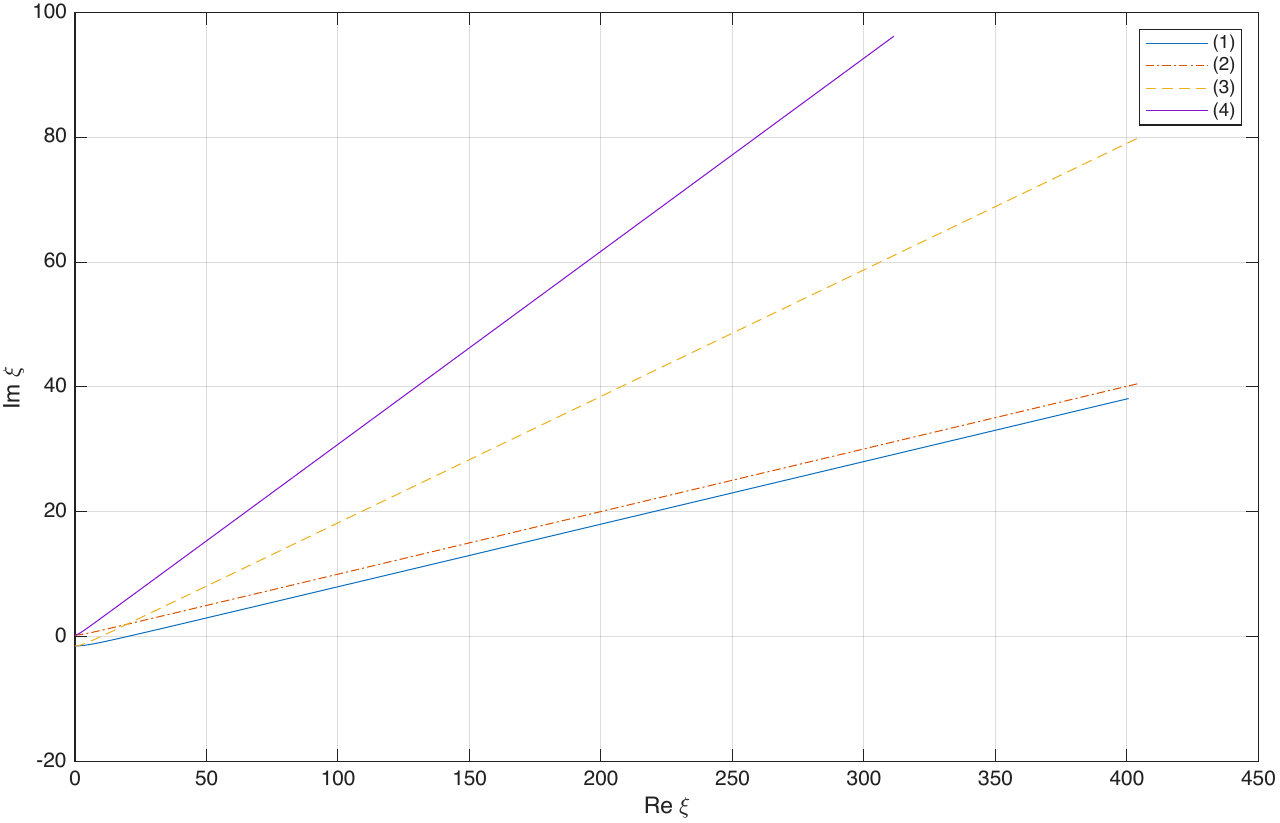}
    \caption{}\label{contours}
\end{subfigure}
&
\begin{subfigure}[h]{0.45\textwidth}
\centering
    \includegraphics[width=0.9\textwidth,height=0.3\textheight]{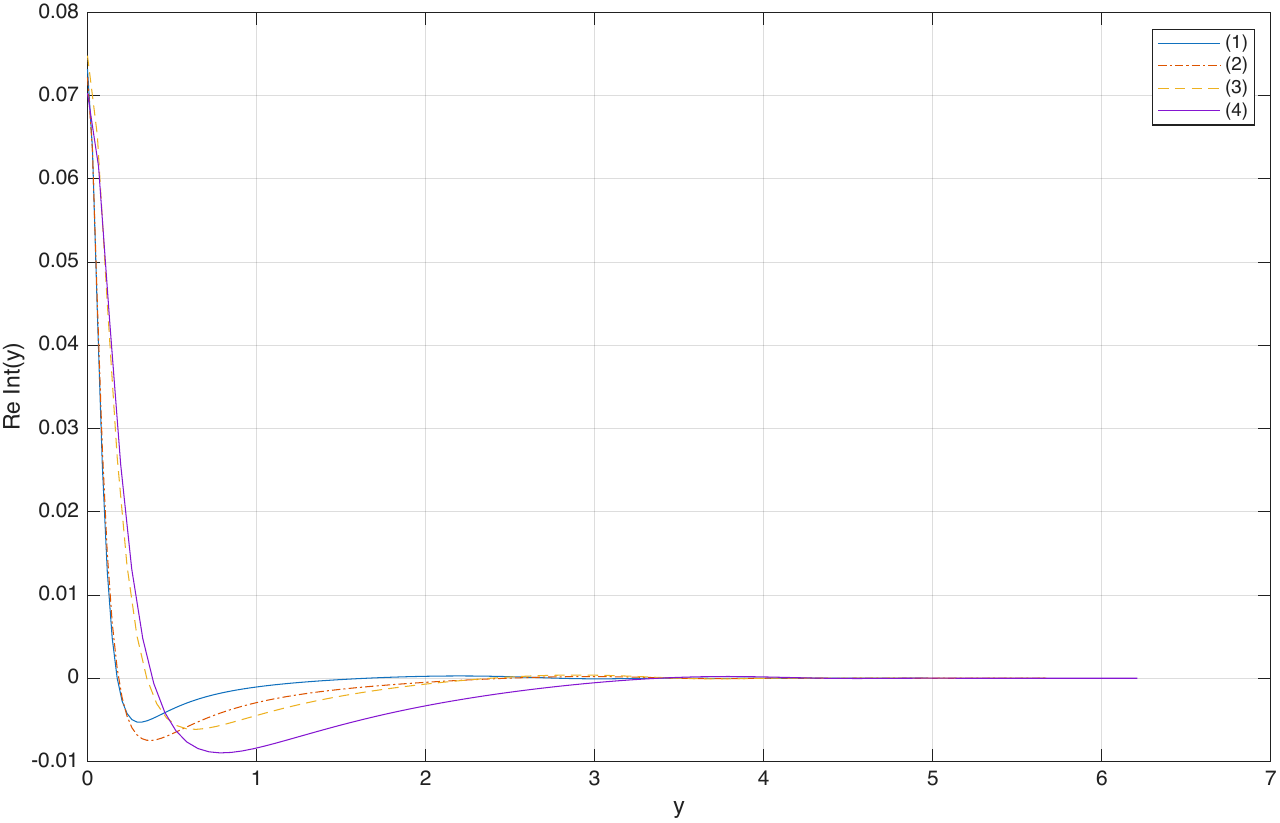}\caption{} 
    \label{logAp}
\end{subfigure}
\end{tabular}
\label{contours-integrands}
\end{figure}   
However, if each term in the chosen quadrature is an analytic function, then the difference between the \emph{true price} (price evaluated using a perfect pricer) and the sum of terms in the chosen quadrature is an analytic function,
and the CB principle will not notice the error. We observed this effect using Modification II of the Adams method in \cite{RoughHestonWe2024}: according to the CB principle, the error was of the order of E-13 in some cases; when we switched to Modification III (the BL modification), the error jumps up to E-09.
The same effect is observed if we use the rotated version of the Gauss-Laguerre quadrature and evaluate
the integrand using the SQ-modification. By design, the values at moderately large nodes become essentially zeros, and each term in the rotated Gauss-Laguerre procedure is an analytic function. We observe
 an extremely fast albeit fallacious convergence.

 In applications to option pricing, 
we use a pair of deformations for OTM puts, and a pair for OTM calls; for a given error tolerance,  
the step $\ze$ is calculated using a universal prescription. If two prices for calls do not agree well for a chosen $N$, a small number of additional terms can be easily added, which is a serious advantage as compared with adaptive quadratures. A larger  than necessary $N^*$ can be used, and prices calculated with $N^*$ and smaller $N$ compared to make sure that the truncation error is small. The resulting program is only a few lines long, and the block for the evaluation of the parameters $\om_1,b,\om, \ze$ (and $N$ if the rate of decay of $\Phi(\xi,T)$ is known) is likewise only a few lines long. Note that SINH-CB method 
can be used in all situations where  numerical Fourier inversion is applied,
and to numerical evaluation of  complicated integrals of a different nature.  
If the domain of analyticity is unknown, we 
use 2-3  deformations and compare the results as in the case when the domain is known, and,
to ensure that in the process of deformation, no pole or singularity has been crossed,
calculate the integral using $\om=0$.

Table \ref{Table_CB} illustrates how the SINH-CB works. One takes several sets of
the parameters of the SINH deformation, and calculates the option prices. If the agreement is insufficiently
good, the step $\ze$ and number of terms $N$ are decreased. 
Following the suggestion by the referee,
we included the sets of the parameters of deformations and prices obtained. The parameters are calculated
using the general prescription in the paper for $\om=\pm 0.1$ and $\om = \pm 0.3$ (``+" for puts, and ``-" for calls), in two versions: 1) strip $S_{(0,1)}$ is used, puts are calculated, and then calls using the put-call parity;
2) strip $S_{(-1,0)}$ is used to evaluate the covered call, and then OTM and ATM puts and calls are calculated. The reader observes that even for options of short maturity and moderately large number of time steps, the relative errors  of prices in (A) - (B) are smaller than $E-05$, hence, either of these prices
can be used as a benchmark to evaluate the errors of the Gauss-Laguerre and SINH quadratures with smaller number of terms. The  errors in (E)-(G) are sizably larger than the relative errors of (A)-(D). The errors of  the Gauss-Laguerre quadrature are especially large farther in the tails, and the SQ modification increases errors further still.

\begin{table}
\caption{\small Prices (rounded) of OTM and ATM options of short maturity $T=1/252$, spot $S_0=1$,
and relative errors of several numerical schemes. Parameters of the rough Heston model $(\al,\ga,\rho,\nu,\theta,v_0)=
(0.62,0.1,-0.681,0.3156,0.331,0.0392)$. (A)-(E): SINH-CB. (F), (G): the Gauss-Laguerre quadrature
with $125$ terms, BL- and SQ-modifications, respectively, are used. In cases (A)-(E), the characteristic function is evaluated using the BL modification with $10000$ time steps and different parameters of SINH-acceleration.  Cases (F2), (G2) are with $M=1000$. Comparison of (E) with (A)-(B) illustrates the error of SINH with rougher and shorter grid,
comparison of (F)-(G) with (F2)-(G2) - the error caused by modifications of the Adams method with smaller number of time intervals, and comparison of (F), (F2) with (G), (G2) - the relative inefficiency of the SQ modification.}
{\tiny \begin{tabular}{c|ccccc}
\hline\hline
$K$ & 	0.95	& 0.975 &	1	&1.025 &	1.05	\\\hline
(A) & 
 2.45438E-07 &1.291108123E-04 
& 5.011158079064E-03 
& 9.162353308E-05
 & 3.3080E-08 \\
 (B) & 2.45431E-07 & 1.291107993E-04 &   5.011158100825E-03 &
 9.162333101E-05 & 3.3042E-08\\
 (C) & 2.45439E-07 & 1.291108120E-04 & 5.011158067024E-03 & 
 9.162347644E-05 & 3.3083E-08\\
 (D) & 2.45433E-07&  1.291107987E-04 & 5.011158091131E-03 &
 9.162333785E-05 & 3.3041E-08\\\hline
 (E) & 2.49019E-07 &   1.290969205E-04 & 5.011200898524E-03 &
 9.161466866E-05 & 3.6990E-08\\\hline\hline
 (F) & 2.45154E-07 &   1.291105240E-04 &  5.011157708667E-03 & 
 9.162296909E-05 & 3.2791E-08\\
 (G)   &    2.43087E-07 &  1.288850170E-04 & 5.010932731825E-03 & 9.176001781E-05
 & 3.3146E-08\\\hline
 (F2) &  2.44753E-07 &  1.290926508E-04 & 5.011196601749E-03 & 9.161027741E-05 &
    3.2708E-08
   \\
 (G2) & 2.36502E-07	& 1.281845611E-04 & 5.010290744068E-03&	9.216234972E-05 &	3.4157E-08
  \end{tabular}}
  {\tiny \begin{tabular}{c|ccccc}
\hline\hline
 & $\om_1$ & $b$ & $\om $ & $\ze$ & $N$ \\\hline
 (A), puts & -2.258176465 & 2.618405789 & 0.3 & 0.029365304& 195\\ 
 (A), calls & -0.741823535 & 2.618405789 & -0.3 & 0.029365304& 195\\
 (B), puts & -2.055909011 & 5.590994333 &0.1 &	0.013133034 &378\\
 (B), calls &-0.944090989 & 5.590994333 &-0.1 &0.013133034 &378\\
 (C), puts & -0.129088232 &	1.309202894 &	0.3 &	0.029365304 & 219\\
 (C), calls & 0.629088232 & 1.309202894 &	-0.3 &	0.029365304 & 219\\
 (D), puts & -0.027954506 & 2.795497166	 &0.1 &0.013133034 &431\\
 (D), calls & 0.527954506 & 2.795497166	 &-0.1 &0.013133034 &431\\
 (E), puts & -2.258176465& 2.618405789&	0.3& 0.06537276&	88 \\
 (E), calls & -0.741823535& 2.618405789&	-0.3& 0.06537276&	88
 \end{tabular}}
 \label{Table_CB}
 \end{table}
  In  (A) - (D), the number of steps $M=10000$ in the BL modification is used.
 Parameters of SINH schemes are chosen using the universal procedure in the paper
 for the error tolerance of the FT inversion $\eps=10^{-15}$, with $\ze$ 10\% smaller and $\La$ 20\% larger,
 and for the strips of analyticity $S_{(0,1)}$ and $S_{(-2,-1)}$ for puts and calls, respectively. 
 The sets of prices can be used to constitute a benchmark. The reader can either select the first digits that appear in all four cases or take the average. Differences of prices in (E) (SINH with a smaller number of terms $N=88$)
  and (F), (G) (Gauss-Laguerre with 100 and 125 terms) and prices in   (A)-(D) are significantly larger  than
  the differences between prices (A)-(D), therefore, the differences between prices (E)-(F) and the average
  of prices in (A)-(D) can be used as proxies of the errors of (E)-(F).
 In (E),  $M=1000$, SINH is used with $\eps=10^{-10}$, and $\ze= 70\%$ and $\La=120\%$ of the recommended. The number of terms is 88. The absolute error is smaller than E-08 for all $K$ but the relative error for OTM options is sizable. 
 
 We formalize the Conformal Bootstrap principle in two forms. 
  \vskip0.1cm
\noindent
{\sc Conformal bootstrap principle I.} {\em Let a union $\cU$ of a strip and cone of analyticity
of the characteristic function $\Phi(\xi)$ be known, and  $\Phi(\xi)$ can be calculated with an (almost) machine precision.
Construct  at least two admissible conformal deformations of
$\chi_j$, $j=1,\ldots, n$, of
the line of integration $\cL^0$ such that  the contours $\cL_j=\chi_j(\cL^0), j=1\ldots, n,$ are not close
and diverge at infinity, and calculate the approximations $V^j$ to the price  
using   
the corresponding changes of variables and simplified trapezoid rule
with several dozens of terms and more.

If  $|V^j-V^k|<10^{-m}$ for $j,k\in 1,\ldots, n$,
where $m$ is not too small, e.g., $m\ge 5$, 
then , \emph{as a practical heuristic}, the common value can be adopted with a conservative tolerance of order $10^{-m+2}$.}

 Lack of agreement across either contours or procedures should be interpreted as evidence of unresolved numerical/analyticity issues (e.g., contour crossing a latent singularity, insufficient decay, or bias in $\Phi_{\mathrm{ap},r}$), in which case one should refine the discretization, modify the deformations, or increase precision.
 This is an a posteriori consistency check across independent admissible contours. Agreement across such deformations is taken as evidence that quadrature, tail truncation, and roundoff errors are collectively small; disagreement flags the need to refine the quadrature, enlarge $\cU$, or adjust the deformation.

It is possible that a region of analyticity $\cU$ 
and the rate of decay of $\Phi(\xi)$ as $\xi\to \infty$ remaining in $\cU$ are unknown as well (this is the case for the rough Heston model).
Then we use
\vskip0.1cm
\noindent
{\sc Conformal bootstrap principle II.}   {\emph Assume that we have two or more numerical procedures for
evaluation of $\Phi(\xi)$ for $\xi$ in a union $\cU$ of a strip and cone. Let $\Phi_{ap,j}(\xi), j=1,2, $ be the approximations.
At least one of the functions $\Phi_{ap,j}$ may not be an analytic function.

Then, if we use different  $\Phi_{ap,j}$ to evaluate the integrals over different contours $\cL_j$, and after   
the corresponding changes of variables and application of the simplified trapezoid rule
with several dozen of terms and more, 
the results agree with the accuracy $10^{-m}$ where $m$ is not small, e.g., $m\ge 7$, 
then, \emph{as a practical heuristic}, we  
\begin{enumerate}[(1)]
\item
accept that $\Phi$ is analytic in a simply connected region $\cU_0\subset \cU$ containing the chosen contours;
\item
the common value can be adopted with a conservative tolerance of order $10^{-m+2}$.
\end{enumerate} }
 Lack of agreement across either contours or procedures should be interpreted as evidence that either one of the deformed contours or both are either outside the domain of analyticity or too close to the boundary, or unresolved numerical/analyticity issues (e.g., contour crossing a latent singularity, insufficient decay, or bias in $\Phi_{\mathrm{ap},r}$), in which case one should modify the deformations. If there is no improvement, the deformations must be changed. When a moderately good agreement is reached, the discretization must be refined and number of nodes increased. If the improvement is observed, the deformed contours are within a domain of analyticity, and we refine the discretization, increase the number of nodes  or increase precision to verify that the agreement improves.
 
 \section{Several popular methods for Fourier inversion}\label{s:FT}
\subsection{Carr-Madan method}\label{ss:CM method}
The implementation of Flat iFT is very simple, and can be easily parallelized if the option prices for  several dozens, hundreds of strikes or even thousands of pairs $(K,T)$ need to be calculated. Furthermore, if the strip of analyticity is not too narrow and the characteristic function
decays not too slowly, which is the case for the Heston model and options of not very short maturity,
then $N$ of the order of 2-3 hundreds suffices to satisfy the error tolerance of the order of E-07 (assuming that the strike $S_0=1$). 

Nevertheless, in the noughties, the unnecessarily complicated (and slower and less accurate)  CM method
 became popular, and it is still used in the quantitative finance literature as {\em one of the standard methods}. 
 Hence, the accurate analysis of the drawbacks of the CM method seems to be useful. The main idea of the method is to use the Fast Fourier transform (FFT) to evaluate
the option prices at several strikes. However, FFT produces the results at points of uniformly spaced 
grids in the $\ln(K)$-space. Therefore, in order to evaluate the option prices for given strikes,
an interpolation procedure needs to be employed. To satisfy even a moderate error tolerance,
a fine grid $x_j=x_0+j\De$, $j=1,2,\ldots, M=2^m$, with $\De\ll 1$ is necessary; to make an accurate Fourier inversion,
a small step $\ze$ in the dual space must be used (in \cite{carr-madan-FFT}, $\ze=0.25$ or $\ze=0.125$ are recommended). The Nyquist  relation $\De\ze=2\pi/M$ requires $M$ to be of the order of several thousand. In \cite{carr-madan-FFT},
the basic recommendation is $M=4,096$ and it is mentioned that larger $M=8,192$ or $M=16,384$ may be needed. Hence, the calculations become computationally many times costlier as compared to
Flat iFT, and an unnecessary interpolation error is introduced. Table \ref{table: iFT-FFT}
illustrates the adverse impact of the interpolation errors on the quality of calibration: the number of
strikes for which the calculated prices are outside the no-arbitrage bounds  increases because of the interpolation.
In the case of the rough Heston model, accurate evaluation of $\Phi(\xi,T)$ for $\xi$ large in absolute value
is especially difficult and time consuming.
The implied volatility surface produced by CM method 
can significantly differ from the correct one 
(see Fig.~ \ref{Set1ImpVolsurfacesXiT152}).
In particular, essentially flat volatility curves
can become nice volatility smiles, and changing the dampening factor (the line of integration) in the CM method, while keeping the same step size
 and the grid size recommended in the CM method one can significantly change the smiles and surface. The  implied volatility surface can significantly change as one changes the step and/or grid size. Furthermore, the errors are systematic,
 and, in many cases, prices of deep OTM options produced by CM method are ``prices" of the systematic errors of the method, which can be ``useful" to produce the implied volatility curves and surfaces one wants to produce.

\subsection{Gaussian quadratures}\label{ss:Lewis}
The specific choice $\om_1=-1/2$ was suggested  by A.~Lewis and A.~Lipton 
\cite{Lewisbook,liptonFX}, and the formula for the covered call was rewritten in the  form
\bbe\label{EuroPriceLL}
V(S_0, K; T) = -\frac{(K/S_0)^{1/2}}{\pi}\Re\int_0^{+\infty}
\frac{e^{iy\ln(S_0/K)}\Phi(T,-i/2+y)}{y^2+0.25}dy.
\ee
In the Lewis method \cite{lewisFT}, it is recommended to change the variable in order to reduce to the integral over $(0,1)$, and then apply the Gauss-Legendre quadrature.
Numerical examples 
show that, in applications to the Heston model and rough Heston model, for the same number of nodes,
the errors of the Gauss-Legendre quadrature are larger than the errors of the Gauss-Laguerre quadrature.
Finally, note that the performance of the Lewis method strongly depends on the choice of the change of variables.
In our numerical experiments, the errors of Gauss-Legendre quadrature increased greatly when we used $y=-\ln u$ instead of $y=u/(1-u)$.
We also observe that 1) given the error tolerance, the SINH-CB method requires the number of nodes 2-5 times smaller than the Gauss-Laguerre quadrature; 2) Gauss-Kronrod method is significantly slower even for the built-in error tolerance,
and, typically, produces errors larger than E-08 whereas SINH-CB satisfies this error tolerance with 20-40 terms, depending on the maturity. 

We finish the discussion about the performance of Gaussian quadratures with the following general observations.
If the integrand is sufficiently regular, then the convergence of Gaussian quadratures are the best ones.
However, the general error bounds are in terms of derivatives of high order, hence,
sizable errors are possible. The error of a  Gaussian quadrature for analytic functions can be expressed as a contour integral in the complex plane. This representation allows for the analysis of the error in terms of the behavior of the integrand in the complex domain. The decay of the error is then related to the distance of the contour from the interval of integration and the analytic properties of the integrand. For instance, a
review paper  \cite{DjukicDjukicPejcevSpalevic2020} starts with ``Let $\Ga$ be a simple closed curve in the complex plane encompassing the interval $[-1,1]$
and let $\cD$ be its interior. Suppose $f$ is a function that is analytic in $\cD$ and continuous on $\cD$."
 However, after the reduction to a finite interval as the Lewis method recommends,
the derivatives become highly irregular and, apparently, very large. At the same time, the integrand does not admit
analytic continuation to a domain containing $[0,1]$, hence, there is no theoretical reason to expect
that the Gauss-Legendre and Gauss-Kronrod quadratures should perform well. In the examples that we consider
both perform moderately well although the latter is too slow and the former insufficiently accurate 
close to maturity and far in the tails. In the same examples, the Gauss-Laguerre quadrature performs much better
although the number of terms needed to satisfy moderately small tolerance  is 2-5 times larger than
the number of terms that SINH-CB method required; even Flat iFT-BM method required smaller number of terms.

\subsection{The Gauss-Laguerre quadrature}\label{ss:GLag}
The situation with the Gauss-Laguerre quadrature is rather peculiar. The theoretical error bound
gives infinity when applied to the same examples; for pricing in the Heston model with small
volatility of variance and far from maturity, in the NIG model far from maturity and the KoBoL model of order $\nu>1$, the same theoretical bound guarantees the excellent convergence
if high precision arithmetic is used and the integrand can be evaluated very accurately. In more detail, 
for the integral
\bbe\label{LagInt}
I(f_0)=\int_0^{+\infty} f_0(y)dy,
\ee 
the error admits a representation in terms of the function $f(y):=e^yf_0(y)$:
\bbe\label{G-Lag-err}
Err_{GL}(I(f_0);n) = \frac{n!}{2(2n)!}f^{(2n)}(y),
\ee
for some $y>0$ (see \cite[25.4.45]{handbook}).
Hence, the general error bound is 
\bbe\label{G-Lag-err-bound}
|Err_{GL}(I(f_0);n)| \le \frac{n!}{2(2n)!}\sup_{y>0}|f^{(2n)}(y)|,
\ee
therefore, the bound \eq{G-Lag-err-bound} is applicable only if $f^{(2n)}$ is uniformly bounded. Furthermore, the very
proof of the convergence of the quadrature is valid in this case only. 
In the case of the Heston model with the following SDE for the variance process, 
\bbe\label{volHeston}
dV_t = \ka(m-V_t)dt + \sg_0\sqrt{V_t}dW_t,\quad V_0=v0,
\ee
the logarithm of $\Phi$ on the RHS of \eq{EuroPriceLL} obeys the asymptotics
\bbe\label{asHeston}
\Re\Phi(T,-i/2+y)= -\frac{(\ka mT+v_0)\sqrt{1-\rho^2}}{\sg_0} y(1+O(y^{-1})).
\ee
Therefore, if $a:=\frac{(\ka mT+v_0)\sqrt{1-\rho^2}}{\sg_0}<1$, there is no reason to expect that
the Gauss-Laguerre quadrature should work, as it seen from Table ~\ref{ErrLaguerre}, where we show
the errors of the  Gauss-Laguerre quadrature applied to $f_0(y)=e^{-ay}$ for various $a$
and various number of terms. However, if $a\ge 1$, and $f(y)=e^y f_0(y)$ admits analytic continuation to a strip $S_{(-d,d)}$ around the real axis, and it is uniformly bounded in the strip, then, applying the Cauchy theorem, one easily proves that there exists $H>0$ such that
$
\sup_{y>0}|f^{2n}(y)|\le Hd^{2n},
$
and the bound \eq{G-Lag-err} can be simplified
\bbe\label{G-Lag-err_anal}
|Err_{GL}(I(f_0);n)| \le H\frac{n!}{2(2n)!}d^{2n}.
\ee
Hence, the rate of convergence of the quadrature is excellent.

Miraculously, if $a<1$ is not very  small, the observed
errors are small; but for very small $a$, the errors become really huge. The rate of decay of the characteristic function very close to maturity is very small, and this is the region where the differences in the performance of
different models are most clearly seen. Hence, the comparative results of the performance of different models using
the Gauss-Laguerre quadrature are unreliable, at the short end especially. \footnote{ For KoBoL processes (a.k.a CGMY model)  of order $\nu<1$ and Variance Gamma model, the performance of the Gauss-Laguerre quadrature is much worse; for KoBoL processes of order $\nu>1$,
the performance is excellent. See Tables \ref{table:relerrKBLBad} and \ref{table:relerrKBLGood}.}
The unexpected good performance can be explained as follows. Assume that $f_0$ admits an upper bound 
\bbe\label{eq:boundf_0}
|f_0(y)|\le Ce^{-a y}
\ee
 where $C>0$ and $a\in (0,1)$ are independent of $y$.  We replace $f_0(y)$
with $f_\eps(y)=f_0(y)e^{-\eps y^2}$, where $\eps>0$ is small. As $\eps\to 0$, the error of the replacement
tends to zero  but the RHS in the error bound  \eq{G-Lag-err-bound} for
$f_\eps$ tends to $+\infty$. An attempt to minimize the sum of the two errors gives an error bound
which is significantly larger than empirically observed errors.

\begin{table}
	\caption{\small Evaluation of $\int_0^\infty e^{-ay}dy$. Relative errors of the Gauss-Laguerre quadrature with $N$ terms. }
	\begin{tabular}{c|ccccc}
		\hline\hline
		$a$ &0.001 &	0.005 &	0.01 &	0.02\\\hline
		$N=100$ & -0.6798	& -0.1440 &	-0.0207 &	-0.00043\\
		$N=125$ & -0.6149 & -0.0878& -0.0077 &	-5.88E-05\\
		$N=150$ & -0.5569 &	-0.0535 &	-0.0029 &	-8.10E-06\\
		$N=175$ & -0.5044 & -0.0326 &	-0.0011&	-1.11E-06\\\hline
	\end{tabular}
	\label{ErrLaguerre}
\end{table}
In a market regime characterized by small $v_0$ and large 
$\sigma_0$, e.g. on a post-earnings day in a calm bull market, the Gauss-Laguerre approach to Fourier inversion for Heston, and especially rough Heston, can become numerically fragile, especially at short expiries and for far-OTM strikes. The combination of very small  initial variance and large volatility of variance causes the characteristic function to oscillate rapidly and decay slowly along the integration contour, especially for short maturities, where the effective spectral parameter is large, and if the sinh contour deformation is not used. In addition, since without the sinh deformation the integrand is highly oscillatory and only weakly damped, the fractional Adams method will typically require many more steps for the same level of accuracy. Since the Adams method is the main numerical bottleneck in the rough Heston price calculation, this will typically result in longer computation times. Even in the example of the calibration to TSLA implied vols, described in Sect.~ \ref{s:model_calib}, where $v_0$ is not small,  one needs at least $N = 125$ Gauss-Laguerre nodes and  $M = 500$ in order to obtain a relative error lower than 0.1\% at 2D expiry ($T=2/365$), for prices higher than $10^{-4}$, whereas SINH needs on average $N=67$ and $M = 180$, respectively, for the calibration to 1W and 2W described in  Sect.~ \ref{ss:calib_sinh}.

\begin{rem}\label{rem:Laguerre_large_xi}
{\em If the integrand has to be evaluated numerically, and errors of 
$\xi$ large in absolute value are sizable, the SINH-acceleration has an additional advantage as compared to the Gauss-Laguerre quadrature. If the SINH-acceleration is used, the average density of nodes decreases
exponentially with $|\xi|$, hence, only a small number of terms in the simplified trapezoid are expected
to make a significant contribution to the total error. If the Gauss-Laguerre quadrature is used,
the density of points in the region of large $|\xi|$ is significantly higher, hence, the total error
is larger.}
\end{rem}

\subsection{Rotated Gauss-Laguerre quadrature and the CB principle}\label{ss:rotGLag}
Assume that $f_0$ admits analytic continuation to a cone $\cC_{-\ga,\ga}$, where $\ga\in (0,\pi)$, and admits a bound
\bbe\label{eq:boundf_0_om}
|f_0(e^{i\om}y)|\le Ce^{-a y},
\ee
where $C>0$ and $a>0$ are independent of $\om\in (-\ga,\ga)$ and $y>0$. Then we can take any
$\om\in (-\ga,\ga)$ and rotate the ray of integration in \eq{LagInt}:
\bbe\label{LagInt_om}
I(f_0)=e^{i\om}\int_0^{+\infty} f_0(e^{i\om}y)dy.
\ee 
Choosing different $\om$, we can apply the CB principle replacing the SINH-acceleration with the rotated version of the Gauss-Laguerre quadrature. If $a\in (0,1)$, the reliability
of this version of the CB principle is questionable.

\subsection{COS method  \cite{COS}}\label{ss:COS}
Both  the pricing density and  payoff are truncated and approximated by linear combinations
of cosines. Thus, two unnecessary truncation errors are introduced. The
error control becomes very difficult: the numerical scheme is characterized by 3 parameters.
The truncation errors are controlled by the choice of two parameters, and, assuming that these  
are chosen sufficiently accurately, the geometric convergence of the method
is illustrated by increasing the third parameter,  the number of terms. The recommendations (in the literature, one can find 
several versions) are formulated in terms of the first 4 or 6 moments of the cumulant.
In view of the exponential growth of the payoff of the call, any recommendation
for the choice of the truncation parameters in terms of the moments cannot be reliable, and the authors of the COS method
explicitly and strongly recommend to apply the method to price puts but not calls (see Remark 5.2 in \cite{COS}).
As numerical examples in  \cite{iFT0,iFT} demonstrate, typically, given the error tolerance, Flat iFT requires a smaller number of terms than COS; in addition, Flat iFT is free from the unnecessary errors and restrictions of COS.

\subsection{SINC method}\label{SINC}
The SINC method \cite{Baschetti2022} is a modification of COS, with a more accurate approximation
of the integrand after the truncation. The authors claim that 1) the new truncation recommendation is more efficient
than the one in the COS method; 
2) the truncation being made,
approximation using sinc-functions is superior to the approximation used in the COS method.
They also state that the number of terms required by the SINC method is 4 times smaller, at best, and, in some cases,
COS is more accurate for the same number of terms.
The theoretical error bounds in \cite{Baschetti2022} are rather complicated, not explicit and essentially impossible  to apply in practice. The numerical example in \cite[Table 2]{Baschetti2022} (pricing put
in the Heston model) demonstrates the resulting errors. Naturally, to hide the errors, \cite[Table 2]{Baschetti2022} shows relative errors
of put options in both  OTM and ITM regions, the errors for ITM puts being excellent. However, when we use the numbers shown
in \cite[Table 2]{Baschetti2022} to calculate the relative errors of the corresponding OTM calls, the errors become quite substantial.
Next, the excellent performance of the infinite trapezoid rule is also explained in \cite{stenger-book} using
the approximation by linear combinations of SINC functions. As a  result, the method
 in \cite{Baschetti2022} uses approximately the same number of terms or even larger than even Flat iFT, to say nothing
 of the Flat iFT-BM and SINH-CB methods constructed below.
Finally, note that in the example shown in  \cite[Table 2]{Baschetti2022}, the strip of analyticity of $\Phi(\xi,T)$ is very wide,
as in the example for the KoBoL (a.k.a. CGMY) model in the same paper. In COS, SINC and  Flat iFT methods,
the number of terms is approximately inversely proportional to the width of the strip of analyticity.
If the strip were narrow (as it is the case for the Heston model for $T$ close to the moment explosion),
both COS and SINC methods would have required several times more terms; very close to the explosion, thousands times more.
The SINH-CB method is much less sensitive to the width of the strip of analyticity and requires small or moderate numbers of terms in all cases.

\subsection{Flat iFT-BM and Flat iFT-NIG methods}\label{ss:Flat iFT-BM} 
The additional errors of  COS are partially compensated by the increase of the width of the strip of analyticity
around the line of integration: instead of one of the three strips $S_{(\mum(T),-1)}, S_{(-1,0)}, S_{(0,\mup(T))}$,
the strip $S_{(\mum(T),\mup(T))}$ can be used. In this section, we demonstrate that the same effect is achievable without introducing additional errors. We use the same straightforward idea as in \cite{ConfAccelerationStable}, where
we eliminated the zero at $\xi=0$ of the integrand in the formula for the cumulative probability distribution function of a stable L\'evy process. In the current setting, we eliminate two zeros, at $\xi=0$ and $\xi=-i$.
Let $\Phi_{ad}(\xi,T)$ be the characteristic function in
a model, where vanilla prices can be calculated faster than in the initial model. Denote by $V_{call}(\Phi; S_0,K;T)$
the call price in the model with the characteristic function $\Phi$; as above, the asset pays no dividends
and interest rate $r$ is constant. 
\begin{prop}\label{prop: Flat iFT-BM}
Let $\Phi(\xi,T)$ and $\Phi_{ad}(\xi,T)$ admit analytic continuation to a strip $S_{(\mum(T),\mup(T))}$,
where $\mum(T)<-1<0<\mup(T)$,
and let $\Phi(-i,T)=\Phi_{ad}(-i,T)=e^{rT}$. 

Then, for any $\om_1\in (\mum(T),\mup(T))$,
\bbe\label{eq:Flat iFT-add}
V_{call}(\Phi; S_0,K;T)=V_{call}(\Phi_{ad}; S_0,K;T)-\frac{Ke^{-rT}}{2\pi}\int_{\Im\xi=\om_1}\frac{e^{i\xi\ln (S_0/K)}
(\Phi(\xi,T)-\Phi_{ad}(\xi,T))}{\xi(\xi+i)}d\xi.
\ee
The equality \eq{eq:Flat iFT-add} is valid for put and covered call as well.
\end{prop}
\begin{proof} Let $\om_1\in (\mum(T),-1)$. Then \eq{eq:Flat iFT-add} is valid. The apparent singularities of the integrand are removable because $\Phi(\xi,T)-\Phi_{ad}(\xi,T)$ is analytic in the strip $S_{(\mum(T),\mup(T))}$
and $\Phi(\xi,T)-\Phi_{ad}(\xi,T)=0$ at $\xi=0,-i$. Hence, the integrand on the RHS of \eq{eq:Flat iFT-add}
is analytic in the strip, and one may move the line of integration to any line 
$\{\xi\ |\ \Im\xi=\om_1\}, \om_1\in (\mum(T),\mup(T))$.
The proof for puts and covered calls is essentially the same. 
\end{proof}
The integral on the RHS of \eq{eq:Flat iFT-add} is calculated using Flat iFT.
If $\mup(T)-\mum(T)\gg 1$, and $\om_1(T)=(\mup(T)+\mum(T))/2$ is chosen, the half-width
of the strip of analyticity used to derive the recommendation for the choice of the step $\ze$ and $N$
increase significantly, and the number of terms of the simplified trapezoid rule and CPU time decrease,
also significantly.

Natural choices for $\Phi_{ad}$ are the characteristic functions in the following models: 
\begin{enumerate}[(1)]
\item
 the BM with the characteristic exponent $\psi(\xi)=\sg^2\xi^2/2-i\mu\xi$; $\sg>0$,
 $\mu=r-\sg^2/2$;
 \item
 Normal Inverse Gaussian process  (NIG) \cite{B-N} or the generalization of NIG 
 (tempered stable L\'evy processes (NTS) constructed in \cite{B-N-L}), with the same or wider strip of analyticity;
 \item
 in applications to rough Heston model, it is feasible that the use of $\Phi_{ad}$
 in the Heston model with the same parameters $\ga, \theta, \nu, \rho$ can be advantageous.
 \end{enumerate}
 We call the resulting method with the choices (1) and (2) {\em Flat iFT-BM} and {\em Flat iFT-NIG} 
 (more generally, {\em Flat iFT-NTS}) methods.
 In the numerical examples in the paper, we use the simplest variant: Flat iFT-BM. 
 In the numerical examples that we considered, the analogs: the Legendre-BM, Laguerre-BM and SINH-BM methods do not bring advantages as
 compared with the Lewis, Gauss-Laguerre and SINH-methods.
 We leave to the future the study of possible advantages of choices (2) and (3).

\subsection{Summation by parts in the infinite trapezoid rule}\label{ss: summation by parts}
For the explicit formulas, see \cite{Contrarian}. The summation by parts significantly decreases
the product $\ze N$ necessary to satisfy the given error tolerance if the strike is not close to
the spot.  Hence, it is natural to separate the region of strikes into two regions: close to the spot,
where Flat iFT-BM (or Flat iFT-NIG) is used, and the region farther from the spot, where, in addition, the
summation by parts is used.

  \section{Numerical examples}\label{s:numer}

The calculations in this section were performed in MATLAB 2024b-academic use, on
a MacBook Pro with an Apple M1 Max chip, 10-core CPU, 24-core GPU, 16-core Neural Engine, 32GB unified memory, 1TB SSD storage.
 \subsection{SINH-acceleration vs Gauss-Laguerre quadrature} The first example is  Table
 \ref{Table_CB} in Sect.~\ref{ss:CP principle}.
 In Fig. \ref{ErrorsSINHvsSQ},
    we show pricing errors of the SINH-acceleration and the Gauss-Laguerre quadrature for options of short maturity
   $T=1/52$ with the number of terms $N$ and number of time steps $M$ chosen as approximately the smallest ones that ensure the accuracy better than E-03.
    In this example, $M=1000$ is sufficient to achieve a good accuracy because $\Re\xi\le 60$.
If the Gauss-Laguerre quadrature of high order, e.g., 200, is used, then $\Re\xi$'s close to 800
appear, and then the errors become very large even if significantly larger $M$ are used. 
The curves obtained with the Gauss-Laguerre quadrature
    with $N=125$ and $N=200$ (in both cases, $M=2000$ and the BL-modification is used)
    are indistinguishable. 
If $M$ is not very large, e.g., $M=1000$, the sequence may oscillate as we observe in \cite[Fig.~2]{RoughHestonWeMarco2025}.

\begin{figure}
\centering
    \includegraphics[width=0.9\textwidth,keepaspectratio]{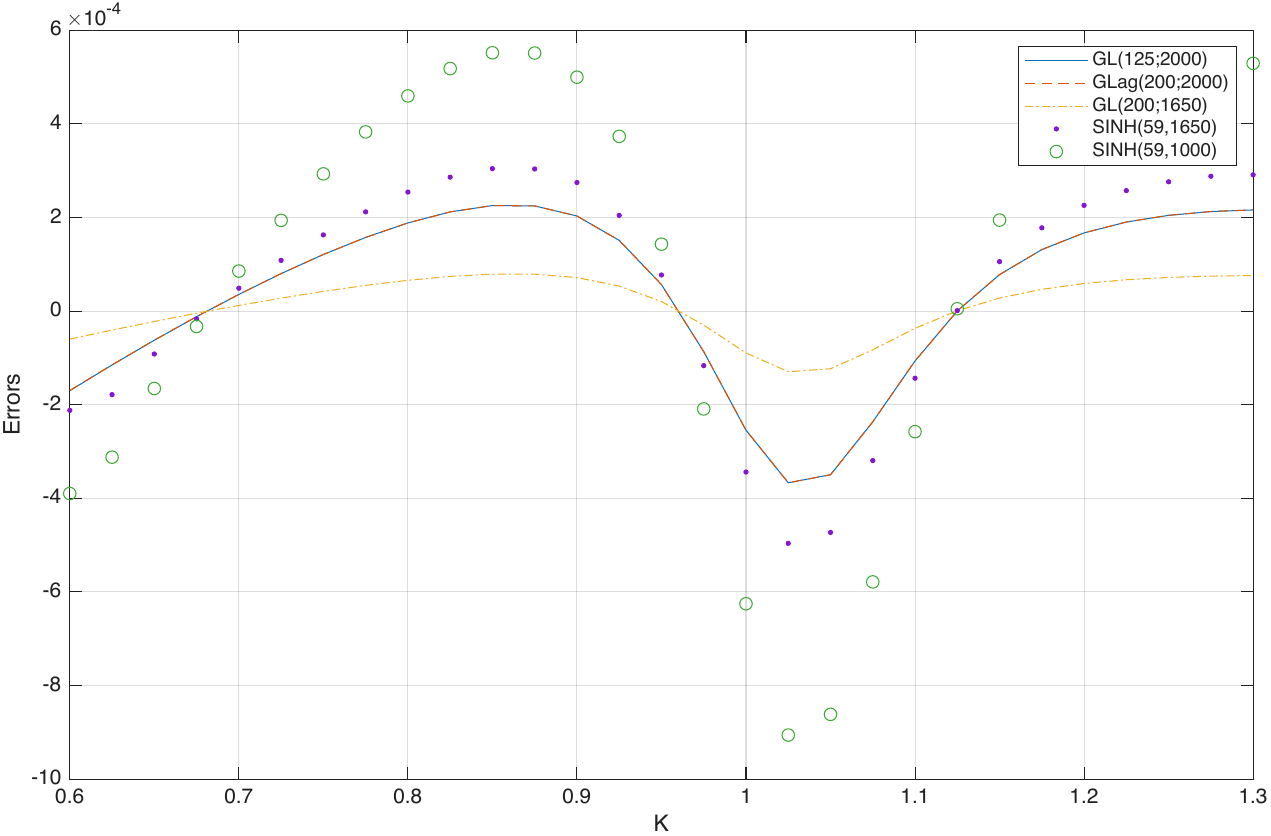}
    \caption{\small Errors of $SINH(N,M)$ and Gauss-Laguerre quadratures $GL(N,M)$. $N$ is the number of nodes, $M$ is the number of time steps in the BL-modification. Step $\ze$ is chosen
    so that the discretization error is smaller than E-05. Time to maturity $\tau=1/52$, parameters of the rough Heston are $(\alpha,\gamma,\theta,\sigma,\rho,v_0)=
(0.587271,\,3.22767,\,0.219608,\,1.49494,\,-0.310089,\,0.552303)$; $\nu=\sg/\ga$.
     \label{ErrorsSINHvsSQ}
    }
    \end{figure}

\subsection{Comparison of the efficiency of different methods for pricing in the Heston model}
\label{ss:comparison_Heston}
To avoid the analysis of the potential impact of errors  of the Adams method and its modification
on the final results, 
we start with the Heston model, where the integrand can be calculated with machine precision, and, therefore,
the errors shown are the errors of the Fourier inversion method used. In the tables, we show the parameters of
each numerical scheme used so that the reader can check that the errors are as shown in the table. 

We use   \cite[Table 2]{Baschetti2022} and prices and relative errors of SINC method shown in 
\cite[Table 2]{Baschetti2022}. Typically, authors tend to select examples to show advantages of their method,
hence, we may presume that the authors of \cite{Baschetti2022}  believe that the results shown are good.
However, we  use small relative errors
of ITM puts shown in \cite[Table 2]{Baschetti2022} to calculate the relative errors of the corresponding OTM call options, which are needed for calibration. The errors of OTM calls turn out to be quite significant, which makes SINC method unsuitable
for calibration purposes. In  \cite[Table 2]{Baschetti2022}, the maturity is $T=0.1$.
The results shown in  Table~\ref{table:SINC_T=0.1} clearly demonstrate that, in addition to the fact that the SINC method
is unnecessarily complicated and uses non-explicit and unreliable recommendations for the choice of the truncation parameters, it is also significantly slower than
Flat iFT with BS correction, the Gauss-Laguerre quadrature and SINH-acceleration. We also observe that  the Gauss-Legendre 
quadrature is significantly
less accurate than Flat iFT with BS correction and SINH-acceleration bar in a small vicinity of the spot. Therefore, in the following two tables 
Table~\ref{table:SINC_T=0.5} and \ref{table:SINC_T=2} for $T=0.5$ and $T=2$,
we do not calculate the errors of SINC method. We also do not show the errors of the Gauss-Legendre quadrature because the errors are systematically and significantly larger than the errors of  the Gauss-Laguerre quadrature.
The parameters in Tables~\ref{table:SINC_T=0.1} - \ref{table:SINC_T=2}  are 
the ones in \cite[Table 2]{Baschetti2022},
and the riskless and dividend rates $r=0$, $q=0$.
\begin{table}
\caption{\small Benchmark prices \cite[Table 2]{Baschetti2022} of OTM and ATM put (panel A) and call  (panel B) options
calculated using  the SINC method 
  in the Heston model with parameters 
 $\ka = 1.5768,
m = 0.0398,
\sg = 0.5751,
\rho = -0.5711,
v_0 = 0.0175$,
 and relative errors of several methods w.r.t. to the BB prices calculated using a more accurate
 SINH-CB method (shown separately). 
 Maturity $T=0.1$, $N$ is the number of terms
of the quadrature.
 }
{\tiny \begin{tabular}{c|ccccc|c} 
 \hline\hline
A &  \\
$K$ & 	0.6	& 0.7 &	0.8	&0.9  &	1	& N\\\hline
$BB_{SINC}$ & 1.1E-09 &	2.363E-07	 &1.98699E-05	& 8.057899E-4 &	0.0163700005 & 1024\\\hline
SINH rel.err. &   -4.77E-03 &	8.73E-06 &	-9.47E-09 &	-7.72E-11	& -1.65E-11
         &  27-44  \\ 
Flat iFT-BS & -2.79E-04 & 	4.17E-07	& -5.13E-09 &	1.07E-09	& 1.08E-10 & 60 \\
Gauss-Laguerre & 8.15E-03 &	4.33E-05 &	5.48E-07 &	1.43E-08	 & 7.39E-10    & 175\\
Flat iFT & 2.52E-04 &	7.88E-07 &	5.82E-08 &	-3.34E-09 &	-3.51E-09 & 200  \\\hline
SINC & -6.16E-02 &	-2.32E-04	& -4.62E-06 &	-5.84E-08	& -2.02E-09 & 384 \\
 & 5.36E-01 &	1.88E-03&	5.07E-05	& -2.17E-06	& 8.96E-08 & 256\\\hline
Gauss-Kronrod & -1.42E-01 &	-4.97E-04 &	2.65E-05 &	4.57E-05 &	-4.17E-07 & \\
Gauss-Legendre & 4.09E+05 &	-4.04E+02 &	3.89E+00 &	1.01E-02	& -9.03E-08 & 200
\\\hline
B & \\
$K$ & 	1	& 1.1 &	1.2	& 1.3  &	1.4	& N\\\hline 
$BB_{SINC}$ & 0.0163700005 &  6.85530575637E-05 & 1.223E-07 & 2.0E-10 & 0 &1024    \\\hline
SINH rel.err. &  -1.65E-11 &  -2.10E-09	& 1.87E-05 &	1.17E-03	& 7.22E+00
 &  44 - 30 \\
Flat iFT-BS &  1.08E-10	& -3.05E-08	&7.74E-06 &	3.25E-02	&1.63E+02
 & 60\\
Gauss-Laguerre & 7.39E-10 &	1.85E-07 &	1.09E-04 &	5.48E-02	& 2.09E+01
 & 175\\
Flat iFT &  -3.51E-09	 & 7.68E-07 &	-7.79E-04	&-5.38E-01 &	-1.24E+02 & 200
  \\\hline
SINC & -2.02E-09 &	-8.40E-07	& -6.36E-04	& -2.12E-01 &	-1.00E+00 &384 \\
       &  8.96E-08 &	2.69E-05	&-1.62E-02	& -8.09E+00 &	2.73E+03   &256 \\\hline
Gauss-Kronrod &  -4.17E-07 &	-3.89E-05 &	-6.54E-01 &	3.81E-02 &	7.54E+04  
\\
Gauss-Legendre & -9.03E-08 & 2.48E-02 &	1.63E+02 &	3.43E+05 &	3.26E+08 & 200  \\\hline
       
\end{tabular}
}
\begin{flushleft}{\tiny 
SINH-CB benchmark prices for \(K=0.6{:}0.1{:}1.4\):
\(1.17218\times10^{-9}\);\allowbreak\ \(2.36354837\times10^{-7}\);\allowbreak\ \(1.9869991862\times10^{-5}\);\allowbreak\
\(8.057899470805\times10^{-4}\);\allowbreak\ \(1.63700005331343\times10^{-2}\);\allowbreak\ \(6.855305756\times10^{-5}\);\allowbreak\
\(1.22377846\times10^{-7}\);\allowbreak\ \(2.538235\times10^{-10}\);\allowbreak\ \(6.968\times10^{-13}\).
\\
Benchmark prices are
calculated using $N=70-110$ terms,
absolute errors  are smaller than E-15.\\
The prices and relative errors of ITM put options of SINC method presented in \cite[Table 2]{Baschetti2022}
are recalculated for the corresponding OTM call options. \\
Parameters of SINH are chosen using the universal scheme for the error tolerance
$E-10$ with 
$\gap=\pi/2, \gam =0, \mup =0, \mum=-1$ for puts and
$\gap=0, \gam =-\pi/2, \gam =0, \mup =0, \mum=-1$ for calls.\\
Flat iFT-BS prices are calculated using $\sg = 0.15, \om_1=-0.1, \ze = 6.7, N=60$\\
Flat iFT prices are calculated using $\om_1=9, \ze = 1, N=200$.
}
\end{flushleft}
\label{table:SINC_T=0.1}
\end{table}

\begin{table}
\caption{\small Benchmark prices of OTM and ATM put (panel A) and call  (panel B) options
calculated using  
 SINH-CB quadrature, in the Heston model,
 and relative errors of several methods.   Parameters are as in Table~\ref{table:SINC_T=0.1}, maturity $T=0.5$, $N$ is the number of terms
of the quadrature.
 }
{\tiny \begin{tabular}{c|cccc|c} 
 \hline\hline
A &  \\
$K$ & 	0.4	& 0.6 &	0.8	&1  	& N\\\hline
$BB_{SINH}$ & 6.867676571E-06 & 2.88352018707E-04 & 4.1468390508486E-03 & 0.0381474566373446
& 37-59\\\hline
$SINH$ & 2.39E-06	&	2.63E-07	&	-6.01E-09	&	-1.52E-11	& 23-37\\
Flat iFT-BS & 3.89E-07	&	1.79E-07	&	2.97E-09&	-1.37E-08 & 70\\
Gauss-Laguerre & 1.20E-06	&	3.37E-08	&	2.63E-09	&	3.17E-10 & 175 \\
Flat iFT & -3.75E-08	&	5.07E-09	&	-1.20E-09	&	-4.24E-09 & 200\\
Gauss-Kronrod & -1.01E-04	&	2.82E-07	&	-3.16E-08	&	-4.66E-07\\\hline
&B& \\
$K$ & 	1	& 1.2 &	1.4	&1.6  	& N\\\hline
$BB_{SINH}$ & 0.0381474566373446 & 8.340111339346E-04 & 3.28092085119E-05
& 2.04002697E-06 & 59-45\\\hline
$SINH$ &-1.52E-11&1.46E-09	&	-1.97E-08	&	1.10E-05 & 37-28\\
Flat iFT-BS & -1.37E-08	&	7.90E-07	&	1.02E-05	&	1.01E-03 & 70\\

Gauss-Laguerre & 3.17E-10	&	1.60E-08	&	4.42E-07	&	7.75E-06 & 175 \\
Flat iFT & -4.24E-09	&	-2.10E-07	&	1.62E-05	&	4.59E-04
 & 200\\
Gauss-Kronrod & -4.66E-07	&	2.26E-03	&	1.87E-03	&5	6.41E-04\\\hline
\end{tabular}
}
\begin{flushleft}{\tiny 
Benchmark prices are
calculated using $N=70-110$ terms,
absolute errors  are smaller than E-15.\\
Parameters of SINH are chosen using the universal scheme for the error tolerance
$E-10$ with 
$\gap=\pi/2, \gam =0, \mup =0, \mum=-1$ for puts and
$\gap=0, \gam =-\pi/2, \gam =0, \mup =0, \mum=-1$ for calls.\\
Flat iFT-BS prices are calculated using $\sg = 0.15, \om_1=-0.1, \ze = 2.5, N=70$\\
Flat iFT prices are calculated using $\om_1=5, \ze = 0.95, N=200$.
}
\end{flushleft}
\label{table:SINC_T=0.5}
\end{table}

\begin{table}
\caption{\small Benchmark prices of OTM and ATM put (panel A) and call  (panel B) options
calculated using  
 SINH-CB quadrature, in the Heston model,
 and relative errors of several methods.   Parameters are as in Table~\ref{table:SINC_T=0.1}, maturity  $T=2$, $N$ is the number of terms
of the quadrature.
 }
{\tiny \begin{tabular}{c|cccc|c} 
 \hline\hline
A &  \\
$K$ & 	0.4	& 0.6 &	0.8	&1  	& N\\\hline
$BB_{SINH}$ & 1.31922212162344E-03 & 7.65194031130601E-03 & 0.0290086131558373 &
0.0886812708686885 & 34-45 
\\\hline
$SINH$ & -1.24E-06	& 	4.46E-07 &		-2.32E-07	 & -6.36E-08 & 21-28\\
Flat iFT-BS & 2.56E-08	&	1.18E-08	&	1.59E-08	&	9.86E-09
 & 65\\
Gauss-Laguerre & 7.02E-09	&	1.19E-09	&	3.59E-10&	1.31E-10
 & 175 \\
Flat iFT & 6.56E-10 &		-2.95E-11	&	2.178E-11	&	2.61E-11
 & 200\\
Gauss-Kronrod & 3.64E-05&	-1.17E-06	&	-1.38E-08	&	-1.07E-09
\\\hline
$K$ & 	1	& 1.2 &	1.4	&1.6  	& N\\\hline
$BB_{SINH}$ & 0.0886812708686885 & 0.0198570250501392 & 
0.00387696016670591 & 9.364368682739E-04 & 65-61\\\hline
$SINH$ & -6.36E-08	&	-1.82E-07	&	3.06E-06	&	-2.30E-06
 & 28-26\\
Flat iFT-BS & 9.86E-09	&	9.883E-08	&	9.62E-07	&	7.43E-06 & 70\\
Gauss-Laguerre & 3.17E-10	&	1.60E-08	&	4.42E-07	&	7.75E-06
 & 175 \\
Flat iFT &2.61E-11 &		2.71E-10	&	2.31E-09	&	1.66E-08
 & 200\\
Gauss-Kronrod & -1.07E-09	&	-1.40E-08	&	-3.00E-07	 &	-6.27E-08\\hline
\end{tabular}
}
\begin{flushleft}{\tiny 
Absolute errors of the benchmark-sinh prices are smaller than E-15\\

Parameters of SINH are chosen using the universal scheme for the error tolerance
$E-10$ with 
$\gap=\pi/2, \gam =0, \mup =0, \mum=-1$ for puts and
$\gap=0, \gam =-\pi/2, \gam =0, \mup =0, \mum=-1$ for calls.\\
Flat iFT-BS prices are calculated using $\sg = 0.15, \om_1=-0.1, \ze = 1.1, N=65$\\
Flat iFT prices are calculated using $\om_1=3, \ze = 0.25, N=200$.
}
\end{flushleft}
\label{table:SINC_T=2}
\end{table}

\subsection{Performance in ``good regions" of the $(K,T)$-plane}
Tables~\ref{table:rel_errors_moderate}, \ref{table:rel_errors_short} and \ref{table:implvol_short} in Sect.~\ref{s:addition}
 demonstrate that even in a rather difficult for accurate pricing rough Heston model,
in regions not close to maturity and rather close to the spot, a moderately small error tolerance can be satisfied using
essentially any reasonable method with a small number of terms, hence, if the data set contains points in this region only, then, for practical purposes, the Gaussian quadratures, Flat iFT, Flat iFT-BM and SINH-CB are essentially equally good. However, since practically useful data sets do contain points in inconvenient regions,
significant calibration errors result if either an insufficiently accurate method is used or the parameter choice
is not good; if the same parameters are used to calculate option prices for all $(K,T)$ and all parameters of the model,
serious errors are inevitable. For similar examples in the context of pricing in KoBoL (a.k.a.) CGMY model,
see Tables \ref{table:relerrKBLBad} - \ref{table:relerrKBLGood}.

\subsection{Examples of incorrect shapes}
In Fig.~\ref{Set1Skews} we show the correct ATM skew 
for the model with parameters \bbe\label{parEuRos}
\al=0.62,\ \ga=	0.1, \	 \rho=-0.681,\  \theta=0.3156, \ \nu=0.331, \ v_0=0.0392,
\ee
calibrated to the S\&P implied volatility surface as of 7 January 2010 in\footnote{The parameters in \eqref{parEuRos} can be found in the published version of \cite{EuchRosenbaum2019}, but not in the preprint.} \cite[\S 5.2]{EuchRosenbaum2019}, and re-used in other studies, e.g. \cite{RoughNotTough, KamuranEmreErkan2020,WangCuiRHINAR2025}. It is clearly seen that
the skew is more than 2 times lower than the one shown on \cite[Fig.~5.1]{EuchRosenbaum2019}.
The correct implied volatility curves shown on Fig.~\ref{Set1Curves} are essentially straight lines,
the slope depending on the maturity, whereas the curves \cite[Fig.~5.2]{EuchRosenbaum2019} are not so flat, 
which is expected, and agree with the empirical data well. Recall that in \cite{EuchRosenbaum2019},
the Lewis method and standard fractional Adams method are used; both are inaccurate. 
We have an example of ghost calibration. For the same parameter set, playing with the parameters of the CM method and 
using interpolation into
the bargain, one can produce implied volatility surfaces of different shapes. See Figure ~\ref{Set1ImpVolsurfacesXiT152}.
In Fig.~\ref{ImperialCurves} and \ref{Set2Curves}, we show the correct implied volatility curves
for two sets of parameters calibrated to the real data in \cite{Imperial2020} and \cite{KamuranEmreErkan2020}.
The curves shown in \cite{KamuranEmreErkan2020} differ by several percent and more, hence, we have an additional pair of ghost calibration
examples.

\section{Fast pricing}\label{s:fast_pricing}
\subsection{Pricing algorithms}\label{s:pricing_algo}
We give a detailed description of the pricing algorithm based on the Conformal Bootstrap principle, the sinh-deformation of the contour, and the modified Adams method. There are two versions of the algorithm: one which is used to calculate the benchmarks, and a faster one to be used on the fly, e.g. during calibration or for live pricing. The detailed description of the benchmark pricing algorithm can be found in Appendix \ref{a:pricing_algo_bm}.

\subsection{Calibration pricer}\label{ss:calib_pricer}
The on-the-fly pricing algorithm used during the calibration is similar to the one outlined in Appendix \ref{a:pricing_algo_bm}, except that the time-consuming optimization is not used, and the flat contour price $V_{LL}(T,K)$ (cf. \eqref{EuroPriceLL}) is only calculated when necessary. We proceed as follows
\begin{enumerate}[1.]\item Take a strip of analyticity, e.g. $(0, \pi/4)$ (see Step II in section \ref{ss: SINH} for the choice of the strip). 
	\item In a loop, price all OTM puts or calls using $\omega = 0.1, 0.2, \ldots$, at each step, e.g. using an initial number of $M = 100$ time steps in the Adams method  for $T>1$, and $M = 300$ for $T < 1$. The  hybrid BL-Adams modification  is used.
	\item The  procedure described in Appendix \ref{sss:data_generation_V_calc} is used, which successively adjusts the numerical parameters (number of time steps $M$, truncation parameter $\La$, mesh $\ze$, number of iterations $n$ in the modified Adams method) by adjusting each until further refinement has negligible effect. 
		
	\item Exit the loop as soon as any two prices have relative difference under e.g. $2 \cdot 10^{-5}$, and return the last price.
	\item Otherwise, try a similar loop with a larger initial value of $M$, e.g. $M = 500$.
	\item If no convergence is observed, then calculate  the price along a flat contour with $\Im \xi = -0.5,$  large initial value of $M$, e.g. $M = 1000$, and no sinh deformation (i.e. using the Lewis-Lipton formula), and check the relative differences between this and any of the previous set of prices. 
\end{enumerate}

\subsection{Performance times}\label{ss:perf-times}
\subsubsection{Hardware and software environment}
All benchmarks were executed on a dual-socket AMD EPYC 7H12 server  
(2 × 64 physical cores, 256 hardware threads, max boost 2.60 GHz) running  
Ubuntu 22.04.5 LTS with the Linux 5.15.0-130-generic kernel.  
The machine was equipped with 256 GB DDR4-3200 RAM. Primary storage comprised a 447 GB RAID-1 Intel SSD system volume. 
No GPU or other hardware accelerators were employed; all timings reported in this paper refer to this CPU-only configuration.
\subsubsection{Implementation details}
The rough-Heston pricer uses a fully vectorised implementation of the Adams method that is \emph{just-in-time} compiled with \texttt{numba}.  
The vanilla-Heston benchmark, by contrast, is a pure Python/NumPy implementation of algorithm in \cite{HestonCalibMarcoMeRisk}, does not employ \texttt{numba}, and sets the roughness parameter to the classical value \(\alpha = 1\) (i.e.\ \(H = 1/2\)). 
\subsubsection{Measured timings}
Table~\ref{tab:perf-times} reports the mean wall-clock time required to price a
single ATM European put option with expiry\footnote{In this section and in the following one, we use calendar days (e.g. $T=2/365$), as commonly done by practitioners, since expiry is a calendar date, carry (rates/dividends/borrow) accrues in calendar time, and weekend theta/P\&L is realized. For short maturities the difference can be material. The academic literature often uses 252 trading days.}
 $T = 2/365$ under each parameter set.  
We used $\om = 0.1$. 
Times were obtained with
Python's \texttt{cProfile}, using the high-resolution \texttt{perf\_counter}
timer.  For each parameter set we performed one warm-up call, to trigger
\texttt{numba} compilation where applicable, followed by 1000 pricing calls;
the value shown is the profiler's cumulative time divided by the number of
calls.  The results indicate near-parity between the two models for the
El Euch-Rosenbaum (EuRos) and SPY sets (24 ms vs.\ 23 ms and 30 ms vs.\ 19 ms,
respectively), whereas for the higher-volatility TSLA and MSTR sets the rough
model is roughly 2–5 times slower.
\subsubsection{Expected performance in vectorised \textsc{C++}}

If both pricers were re-implemented in high-performance, vectorized \textsc{C++} with identical numerical tolerances, the Python overhead would disappear and both methods would be expected to complete in a few milliseconds per price evaluation.  
Since the rough-Heston characteristic function involves the additional application of the Adams method, we would expect the rough model to  run approximately \(1\text{–}5\) times slower than the vanilla Heston pricer.

\subsubsection{Comparison with Markovian approximation}\label{sss:markov_times} 
Table~\ref{tab:markov-times} reports the total pricing times for the same parameter sets, for an ATM put option and for maturities of 2~days and 1~week. 
Prices were obtained using the BL2 method, identified in \cite{MarkovianGG} as the most efficient and accurate one among the proposed approaches. 
In \eqref{e:markov-expansion}, 
the number of nodes $n$ was selected so that the ATM volatility error remained under 1\%.
The timings were obtained by running the Python implementation made available by the authors of \cite{MarkovianGG} on GitHub \cite{breneis2025}. Less than $2\%$ of the total runtime is spent on node computation, leaving little room for acceleration through pre-caching. A full description of the algorithm can be found in Appendix \ref{a:BL2}. We used a tolerance of $\varepsilon = 10^{-3}$ for the relative error.

{\small
\begin{table}[htbp]
	\centering
	\caption{Average wall-clock time per contract for the rough and vanilla Heston pricers on the hardware described above.\protect\footnotemark}
	\label{tab:perf-times}
	\begin{tabular}{lcc}
		\toprule
		Set & Rough Heston (ms) & Vanilla Heston (ms) \\
		\midrule
		EuRos & 24 & 23 \\
		TSLA  & 50 & 19 \\
		MSTR  & 91 & 20 \\
		SPY   & 30 & 19 \\
		\bottomrule
	\end{tabular}
\end{table}
}
\footnotetext{
	Vanilla Heston prices were computed with the method of \cite{HestonCalibMarcoMeRisk} (no \texttt{numba} acceleration) and use \(\alpha = 1\;(H = 1/2)\). The TSLA set was calibrated to options on this name as of 2 May 2025 (cf. Section \ref{s:model_calib}), while MSTR and SPY were calibrated to the corresponding names as of  2 June and 31 March, 2025, respectively. 
	The parameter sets are as follows, with $\sg = \ga \nu$; \((\alpha,\gamma,\theta,\sigma,\rho,v_0)\):  
	EuRos \((0.62, 0.10, 0.3156, 0.0331, -0.681, 0.0392)\) from \cite[\S 5.2]{EuchRosenbaum2019};  
	TSLA \((0.5119, 2.3661, 0.4249, 1.3684, -0.1785, 0.5275)\);  
	MSTR \((0.6254, 2.2046, 1.1908, 3.9948, -0.4078, 0.3458)\);  
	SPY  \((0.7151, 1.8967, 0.03848, 1.1654, -0.6704, 0.06246)\).}

{\small

\begin{table}[htbp]
	\centering
	\caption{Calculation times for the BL2 Markovian approximation method \cite{MarkovianGG}, based on the minimal number of nodes $n$ in \eqref{e:markov-expansion} needed to keep the ATM volatility error below 1\%}
	\label{tab:markov-times}
	\begin{tabular}{lccc}
		\toprule
		Set & Expiry & Nodes & Time (sec.) \\
		\midrule
		TSLA & 2D & 2 & 51.89 \\
		TSLA & 1W & 2 & 16.77 \\
		MSTR & 2D & 1 & 3.90 \\
		MSTR & 1W & 1 & 1.58 \\
		EuRos & 2D & 1 & 4.11 \\
		EuRos & 1W & 1 & 2.79 \\
		SPY & 2D & 1 & 1.57 \\
		SPY & 1W & 1 & 0.83 \\
		\bottomrule
	\end{tabular}
\end{table}
}

\section{Calibration results}\label{s:model_calib}

\subsection{Calibration using SINH-CB}\label{ss:calib_sinh}
This section includes an example of how our new pricing method can be applied to calibrate the rough Heston model, using Tesla (TSLA) option data from Bloomberg. We perform the calibration on TSLA implied volatility smiles as of 2 May 2025, fitting on short-dated maturities (1-week and 2-week expiries), by minimising the sum of squared differences between model and market implied volatilities. Figure \ref{fig:insample} below shows the in-sample fit of the model to market implied volatilities for these maturities. The rough Heston model is able to closely reproduce the observed smiles at 1-week and 2-week expiries.
\sbr 

To verify the reliability of the fast pricer used in our calibration procedure, we conducted a ``reverse calibration'' test. In this test, we used the benchmark pricer described in Appendix \ref{a:pricing_algo_bm} and the set of calibrated parameters, i.e.
\begin{equation}\label{params:rough_short}
(\alpha,\gamma,\theta,\sg,\rho,v_0)_{rough}^{short} = (0.511913,\,2.36609,\,0.424949,\,1.36839,\,-0.178493,\,0.527527)\, ,
\end{equation}
where $\sigma = \gamma \nu$, to generate option prices, and hence implied vols, at expiries corresponding to 4, 11, 17 and 25 days, respectively, and moneyness levels between 0.6 and 1.6 for the first two expiries, and between 0.4 and 1.75 for the others. Treating these as ``market'' quotes, we then recalibrated the model with our fast pricer (described in section~\ref{ss:calib_pricer}). The latter recovered virtually identical parameters:
\[
(\alpha,\gamma,\theta,\sigma,\rho,v_0)_{rough;fast}^{short} = (0.512399,\,2.38011,\,0.425275,\,1.37226,\,-0.178501,\,0.527526)\,.
\] 
These match the benchmark values within about $0.2\%$ on every parameter. The maximum absolute deviation in any parameter is only $1.4\times10^{-3}$ (occurring in the mean-reversion rate $\gamma$, which is the hardest to calibrate), and the average relative error is approximately $0.08\%$. These results confirm that  our fast pricer, based on conformal bootstrapping with sinh-deformation, is sufficiently accurate and robust for calibration, essentially reproducing the original model parameters.

\sbr 
Aside from the pricer's performance, we also find that the model calibration extrapolates well across time. Using the parameters calibrated to the 1W--2W expiries, we priced options at longer maturities that were not included in the calibration (3W and 7W expiries). The resulting implied volatility smiles, shown in Figure~\ref{fig:outsample}, indicate that the model's predictions remain close to the actual market smiles for these longer expiries. In other words, the rough Heston model calibrated to short-term options is able to capture the term structure of volatility out to about one month without any re-calibration. This is especially useful for applications to market making, since broker dealers or market makers often need to provide quotes for illiquid expiries.

\begin{figure}[t]
	\centering
	\includegraphics[width=0.48\textwidth]{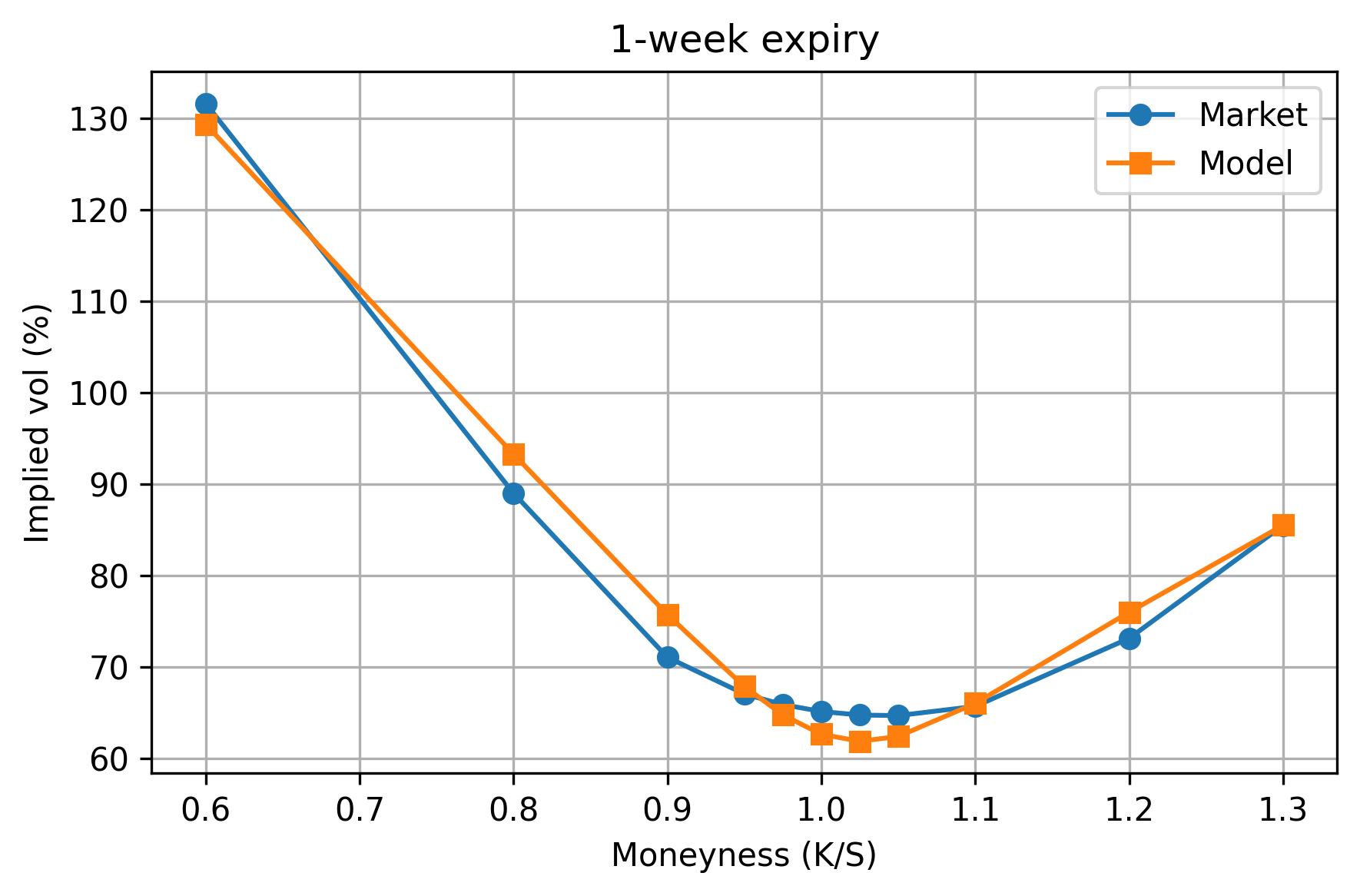}\hfill
	\includegraphics[width=0.48\textwidth]{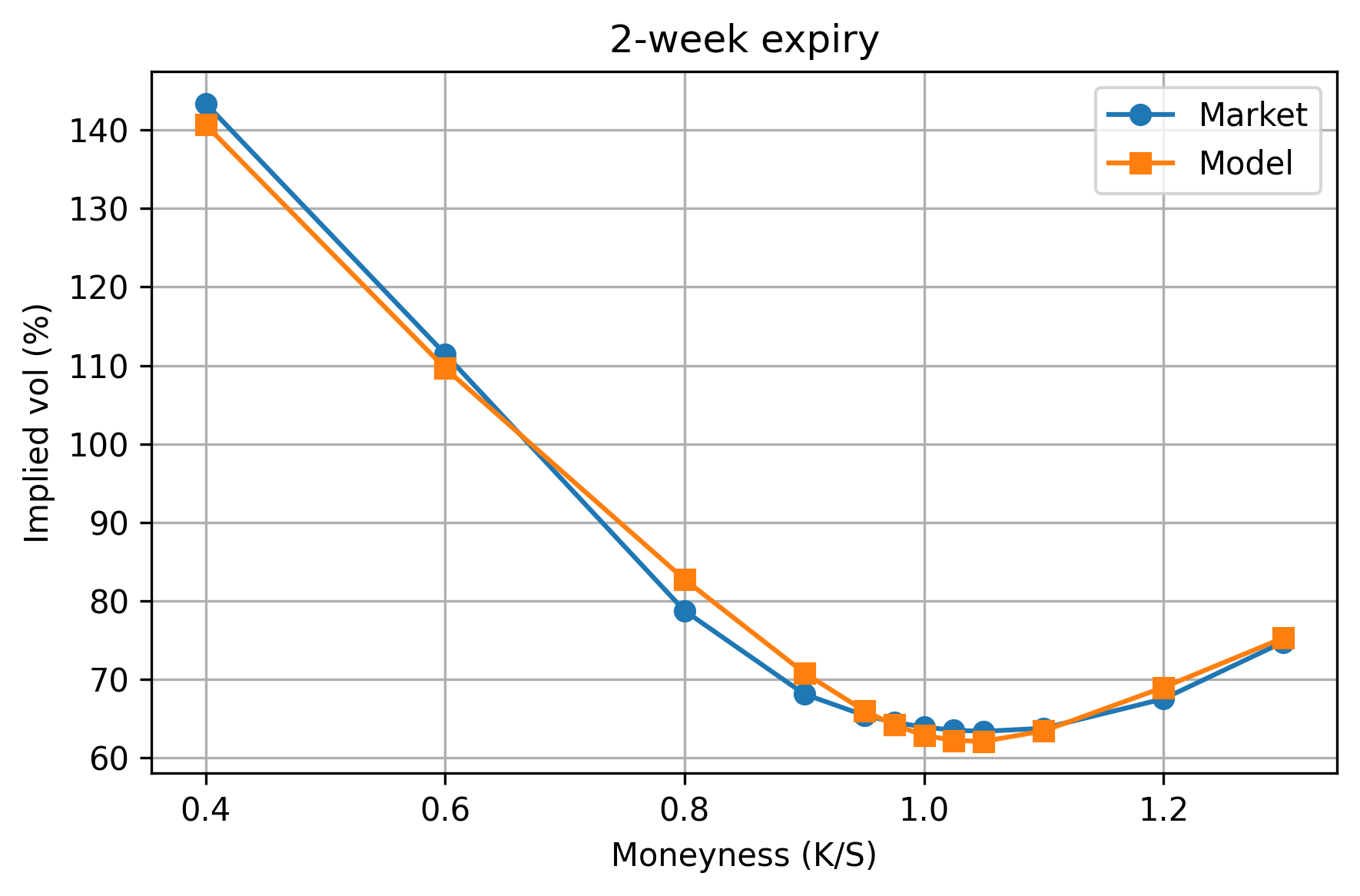}
	\caption{In-sample calibration of the rough Heston model to TSLA option smiles on 2~May~2025. Implied volatility (IV) smiles for the 1-week expiry (left panel) and 2-week expiry (right panel) are shown. Market IVs (blue) are closely fitted by the model IVs (orange). Model parameters are in equation \eqref{params:rough_short}}
	\label{fig:insample}
\end{figure}

\begin{figure}[h!]
	\centering
	\includegraphics[width=0.48\textwidth]{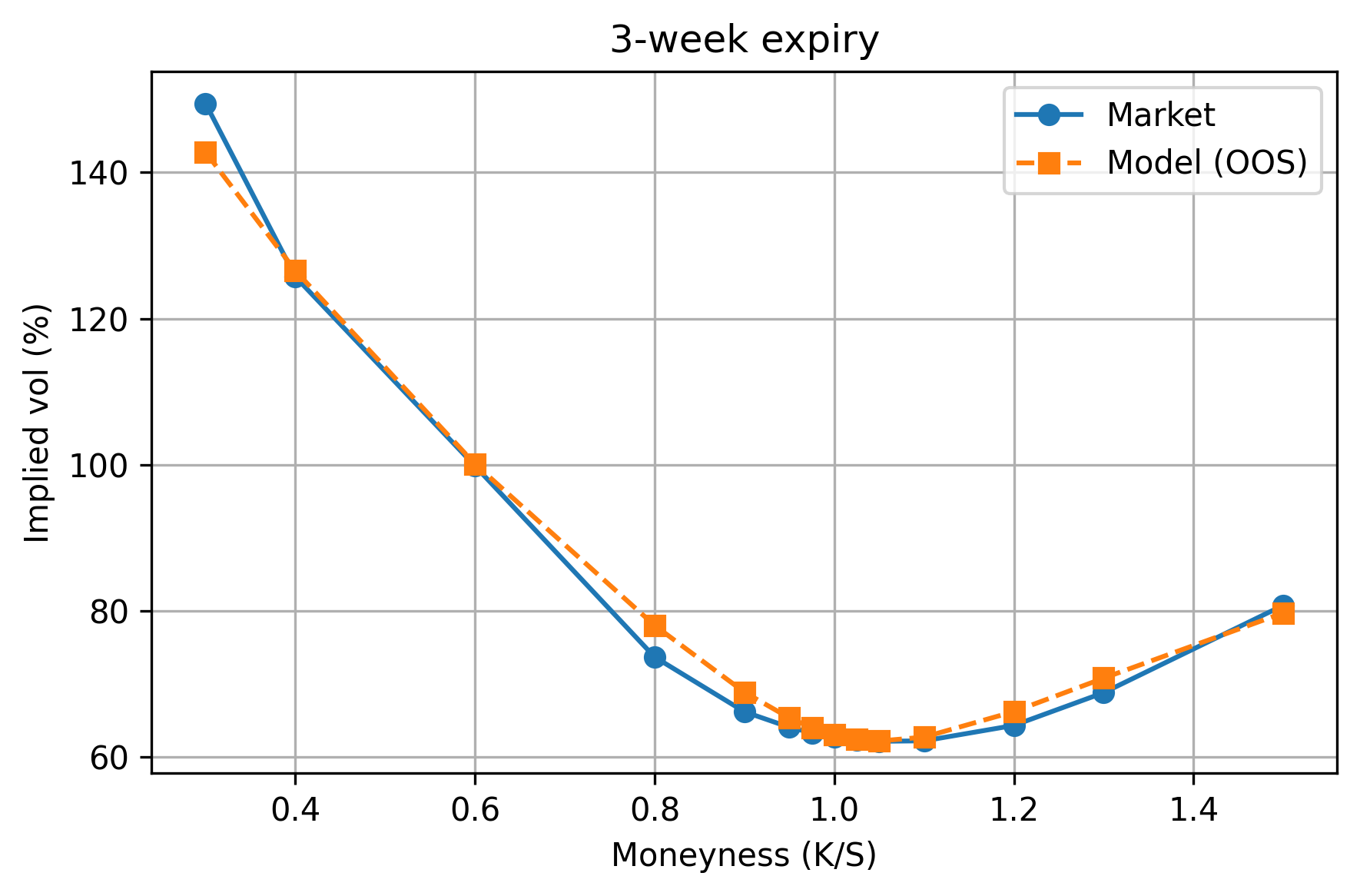}\hfill
	\includegraphics[width=0.48\textwidth]{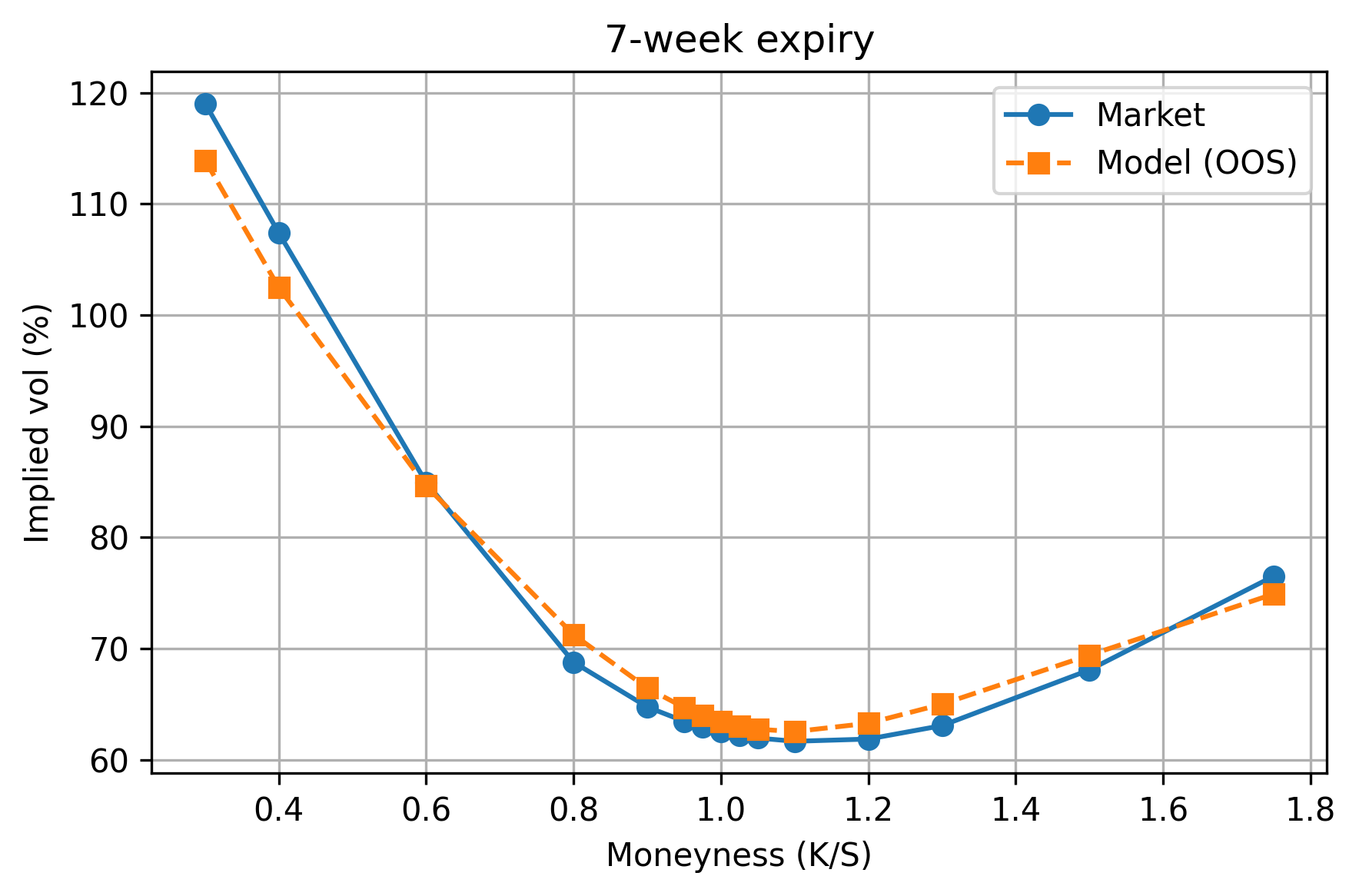}
	\caption{\small Out-of-sample implied volatility smiles at longer expiries on 2~May~2025, using the rough Heston parameters calibrated only on 1--2 week maturities. The model (orange lines) extrapolates the smile well for both the 3-week expiry (left) and 7-week expiry (right), staying in line with the market implied volatilities (blue lines). This demonstrates the model's robust extrapolation in the near-term maturity range. Model parameters are in equation \eqref{params:rough_short}.}
	\label{fig:outsample}
\end{figure}
\begin{figure}[t]
	\centering
	\includegraphics[width=0.48\textwidth]{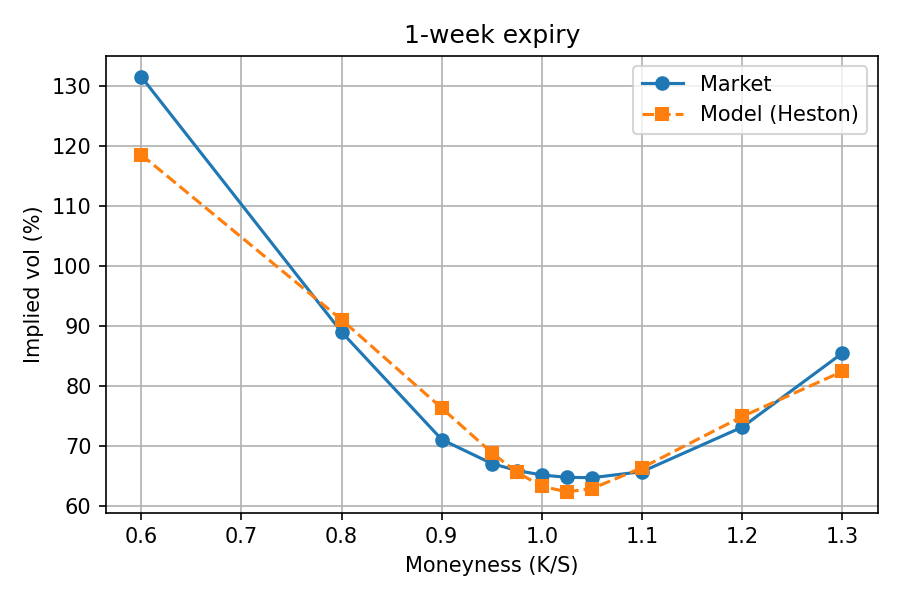}\hfill
	\includegraphics[width=0.48\textwidth]{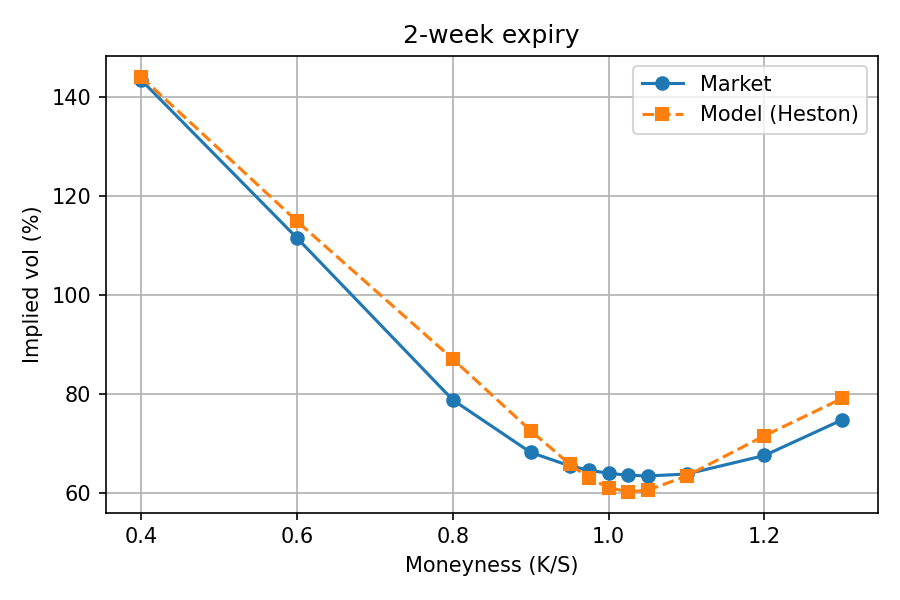}
	\caption{In-sample calibration of the \emph{regular} Heston model to TSLA option smiles on 2~May~2025. Implied volatility (IV) smiles for the 1-week expiry (left panel) and 2-week expiry (right panel). Market IVs (blue) versus model IVs (orange). Model parameters are in equation \eqref{params:heston_short}.}
	\label{fig:insample_heston}
\end{figure}

\begin{figure}[h!]
	\centering
	\includegraphics[width=0.48\textwidth]{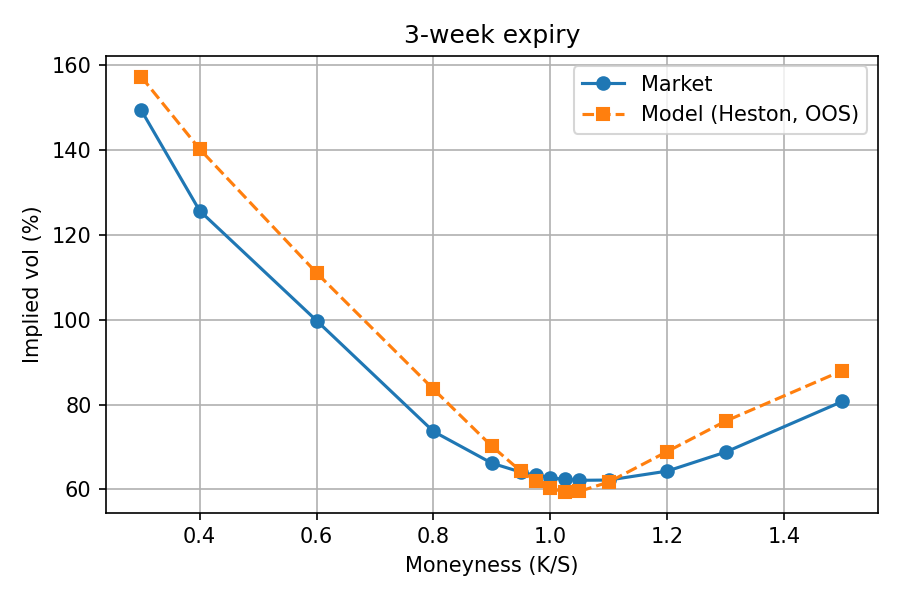}\hfill
	\includegraphics[width=0.48\textwidth]{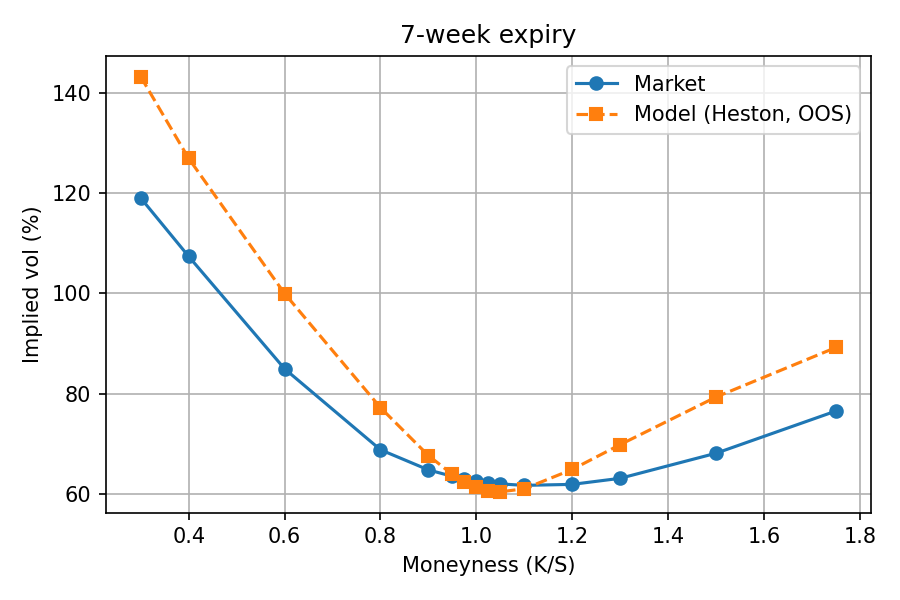}
	\caption{Out-of-sample implied volatility smiles at longer expiries on 2~May~2025, using the \emph{regular} Heston parameters calibrated only on 1--2 week maturities. The model (orange) extrapolates poorly relative to the rough Heston case, with the gap widening for the 3-week (left) and 7-week (right) expiries. Model parameters are in equation \eqref{params:heston_short}.}
	\label{fig:outsample_heston}
\end{figure}

\begin{figure}[t]
	\centering
	\includegraphics[width=0.48\textwidth]{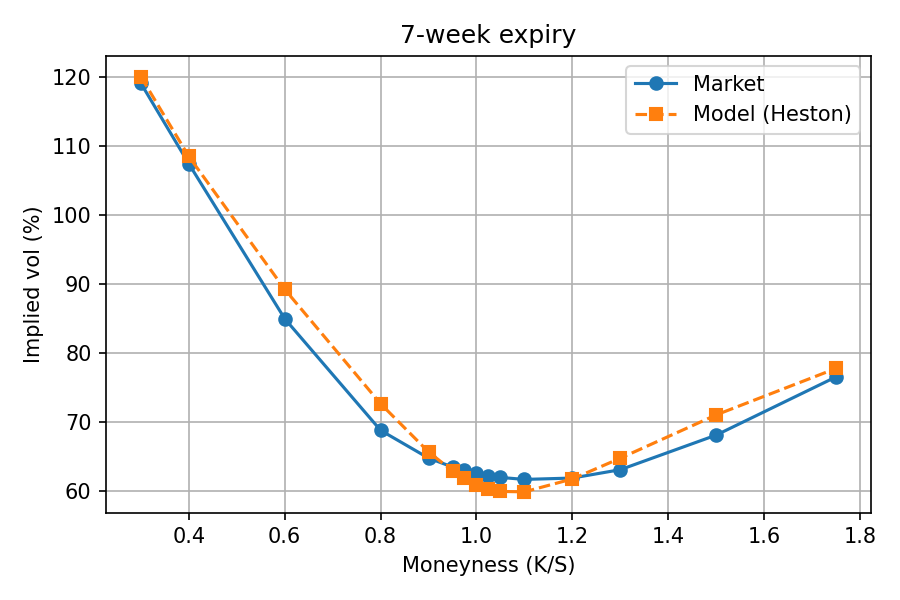}\hfill
	\includegraphics[width=0.48\textwidth]{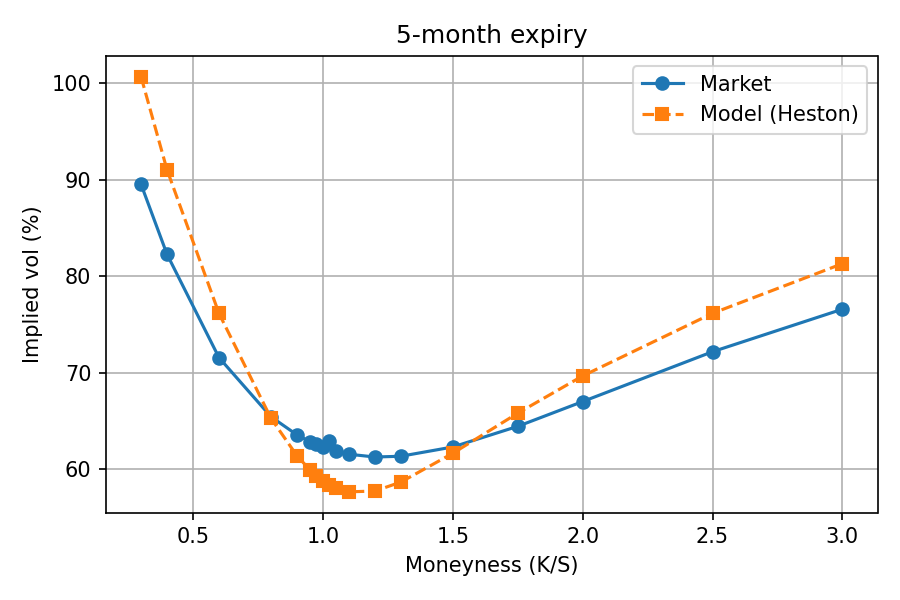}
	\caption{In-sample calibration of the \emph{regular} Heston model to TSLA option smiles on 2~May~2025 across six expiries between 3 weeks and 1.7 years. Implied volatility (IV) smiles for the 7-week expiry (left) and 5-month expiry (right). Model parameters are in equation \eqref{params:heston_long}.}
	\label{fig:insample_long_heston}
\end{figure}

In order to assess the calibration quality in a way familiar to practitioners, we follow the ``average percentage error'' (APE) definition of \cite{perfect-cal} and apply the same statistic directly to implied volatilities.\footnote{The original paper defines APE for option prices.  }  The resulting \emph{average volatility error} (AVE) at a given expiry is

\begin{equation}
	\mathrm{AVE}
	\;=\;
	\frac{\dfrac{1}{N}\displaystyle\sum_{i=1}^{N}
		\bigl| \sigma^{\text{mkt}}_i
		-\sigma^{\text{model}}_i \bigr|}
	{\overline{\sigma^{\text{mkt}}}}
	\times 100\%,
	\label{eq:ave_def}
\end{equation}
where $N$ is the number of strikes in the plot and
$\overline{\sigma^{\text{mkt}}}$ is the strike-averaged market IV.

Table~\ref{tab:ave} below lists the AVE values corresponding to the smiles
shown in Figures~\ref{fig:insample} - \ref{fig:outsample}.

\begin{table}[htb]
	\centering
	\caption{Average Volatility Error (AVE) by maturity, computed with Eq.~\eqref{eq:ave_def}, for the calibration to 1W and 2W. The rough Heston model (SINH-CB pricer) is compared against the regular Heston model. Rough Heston shows good out-of-sample performance at longer expiries, whereas regular Heston degrades markedly. Model parameters are in equations \eqref{params:rough_short} and \eqref{params:heston_short}.}
	\label{tab:ave}
	\begin{tabular}{@{}lccc@{}}
		\toprule
		Expiry & Strike range  & AVE (\%), rough Heston & AVE (\%), Heston \\ \midrule
		1~week  & $0.40\le K/S\le1.30$      & 2.85 & 4.03 \\
		2~weeks & $0.40\le K/S\le1.30$      & 1.95 & 3.93 \\
		3~weeks & $0.30\le K/S\le1.50$      & 2.04 & 6.96 \\
		7~weeks & $0.30\le K/S\le1.75$      & 2.43 & 10.08 \\
		5~months & $0.30\le K/S\le2.50$ & 1.96 & 16.99 \\ \bottomrule
	\end{tabular}
\end{table}

For the rough Heston calibration, the AVE stays well below 3\% even at five months' expiry, indicating robust extrapolation across maturities. By contrast, the regular Heston calibration (see Figures \ref{fig:insample_heston} and \ref{fig:outsample_heston}) is not only a poorer fit, but extrapolates much worse across time to maturity. The Heston parameters for the calibration to 1W and 2W expiries are
\begin{equation}\label{params:heston_short}
	(\gamma,\theta,\sg,\rho,v_0)_{Heston}^{short} = (5,\,1.07873,\,9.04923,\,-0.203652,\,0.474217).
\end{equation}

\sbr

To further validate the model's performance, we extend the calibration directly to these longer expiries, i.e. approximately 3 weeks, 7 weeks, 5 months, 7.5 months, 1.4 years, 1.7 years.  We focus specifically  on the 7-week and 5-month maturities. The parameters obtained using SINH-CB are:
\begin{equation}\label{params:rough_long}
(\alpha,\gamma,\theta,\sigma,\rho,v_0)_{rough}^{long} = (0.501043,\,1.37984,\,0.376052,\,1.26395,\,-0.17103,\,0.51981)\,.
\end{equation}
This calibration results in an average percentage error (APE) of $1.074\%$ and an average volatility error (AVE) of $1.104\%$. It is noteworthy that these errors are lower than those obtained for the calibration to only the first two expiries. This is to be expected, as implied volatility smiles at short expiries are more heavily dominated by jump dynamics, which are inherently more difficult to capture with pure diffusion models. Figure~\ref{fig:insample_long} illustrates the in-sample fit for these longer expiries. Furthermore, the in-sample errors for this longer expiry calibration are of the same order of magnitude as the out-of-sample errors discussed previously, reinforcing the consistency and stability of the method across the volatility surface, in this example. 
The regular Heston model behaves differently. When it is calibrated directly to the longer expiries (see Figure~\ref{fig:insample_long_heston}), it attains an APE of 3.262\%
 and an AVE 3.578\% in-sample. The corresponding parameters are
 \begin{equation}\label{params:heston_long}
 	(\gamma,\theta,\sigma,\rho,v_0)_{Heston}^{long} = (5,\,0.492407,\,4.85046,\,-0.217979,\,0.501616)\,.
 \end{equation}

 While this is a marked improvement on its poor out-of-sample performance at the same expiries (cf. Table~\ref{tab:ave}), it remains substantially worse than the rough Heston fit, whose in-sample errors stay near 1\%.

\begin{figure}[t]
	\centering
	\includegraphics[width=0.48\textwidth]{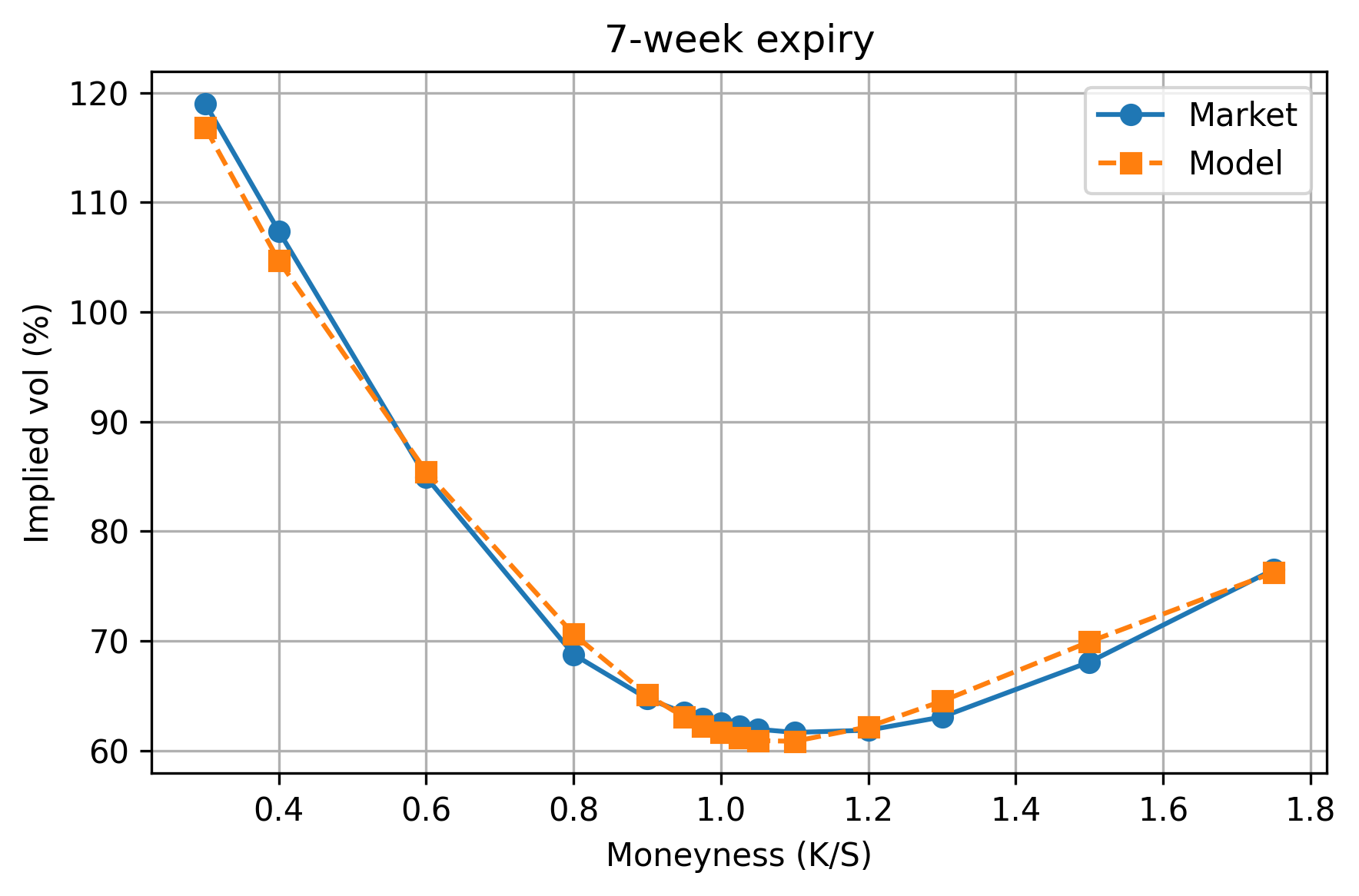}\hfill
	\includegraphics[width=0.48\textwidth]{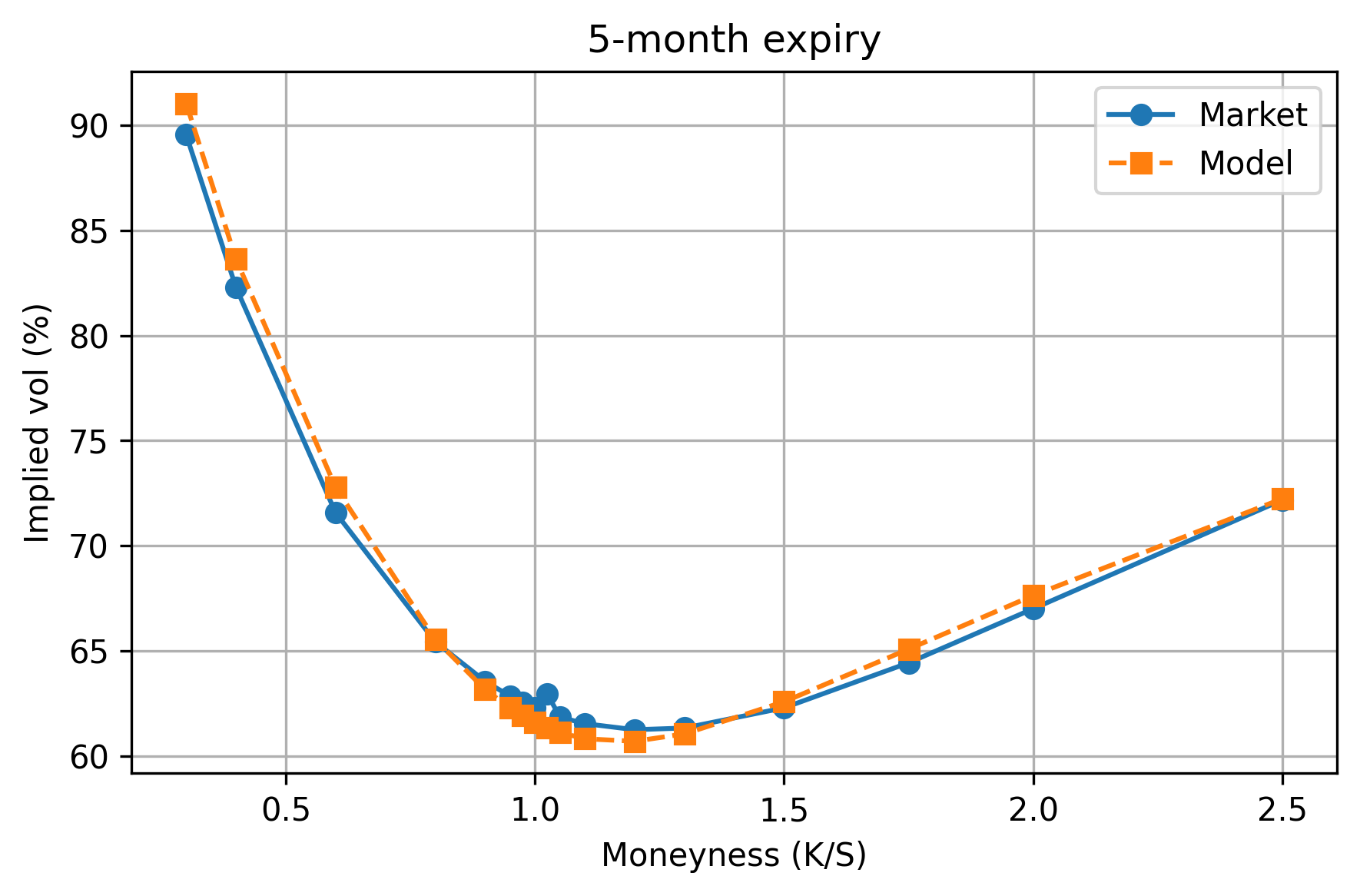}
	\caption{In-sample calibration of the rough Heston model to TSLA option smiles on 2~May~2025 to six expiries between 3 weeks and 1.7 years. Implied volatility (IV) smiles for the 7-week expiry (left panel) and 5-month expiry (right panel) are shown. Model parameters are in equation \eqref{params:rough_long}.}
	\label{fig:insample_long}
\end{figure}

\subsection{Calibration pitfalls}\label{ss:calib_laguerre}
We show the consequences on the calibration results of using fixed pricing settings, as often done by practitioners and 
academics alike, e.g. when using CM, COS, or even Gaussian quadratures. We compare the calibration results obtained in the previous
section with those obtained using Gauss-Laguerre quadrature with $N = 200$ nodes, and $M = 6000$ time steps in the Adams method, which is higher than the value of $M$ used for SINH-CB for any of the calibrations described in the previous section\footnote{We verified that increasing $M$ beyond this value left the results essentially unchanged, and $N=200$ is at the top of the range supported by double precision arithmetic for the Gauss-Laguerre nodes.} 

The same modification of the Adams method is used as for SINH-CB, and the model is calibrated to the same implied volatilities for 1W and 2W expiries, as in section \ref{ss:calib_sinh}. The resulting parameters are 
\begin{equation}\label{params:rough_short_laguerre}
(\alpha,\gamma,\theta,\sigma,\rho,v_0)^{short}_{rough;\text{GL}} = (0.587271,\,3.22767,\,0.219608,\,1.49494,\,-0.310089,\,0.552303)\,.
\end{equation}
The maximum absolute deviation in any parameter is  48.3\%, for $\theta$, and the average deviation is 31.2\%. 
Figure \ref{fig:ghost_sundial_calib} below shows the pitfalls of inaccurate pricing in model calibration. The left panel shows the ``ghost calibration'' effect: the solid red line is the 1W implied volatility calculated from the model parameters above, which appears to fit the market, however this is an illusion. When the model's true volatility curve is calculated using SINH-CB (dashed red line), it is revealed to be a poor match. The right panel shows the ``sundial effect'': our proposed SINH-CB method finds a superior fit to the market (solid blue line). If one tries to use the GL pricer on this superior model, it fails to reproduce the correct volatilities, especially in the right tail. This shows how a flawed tool can cause a good model to be incorrectly dismissed. 

\begin{figure}[t]
	\centering
	\includegraphics[width=0.5\textwidth]{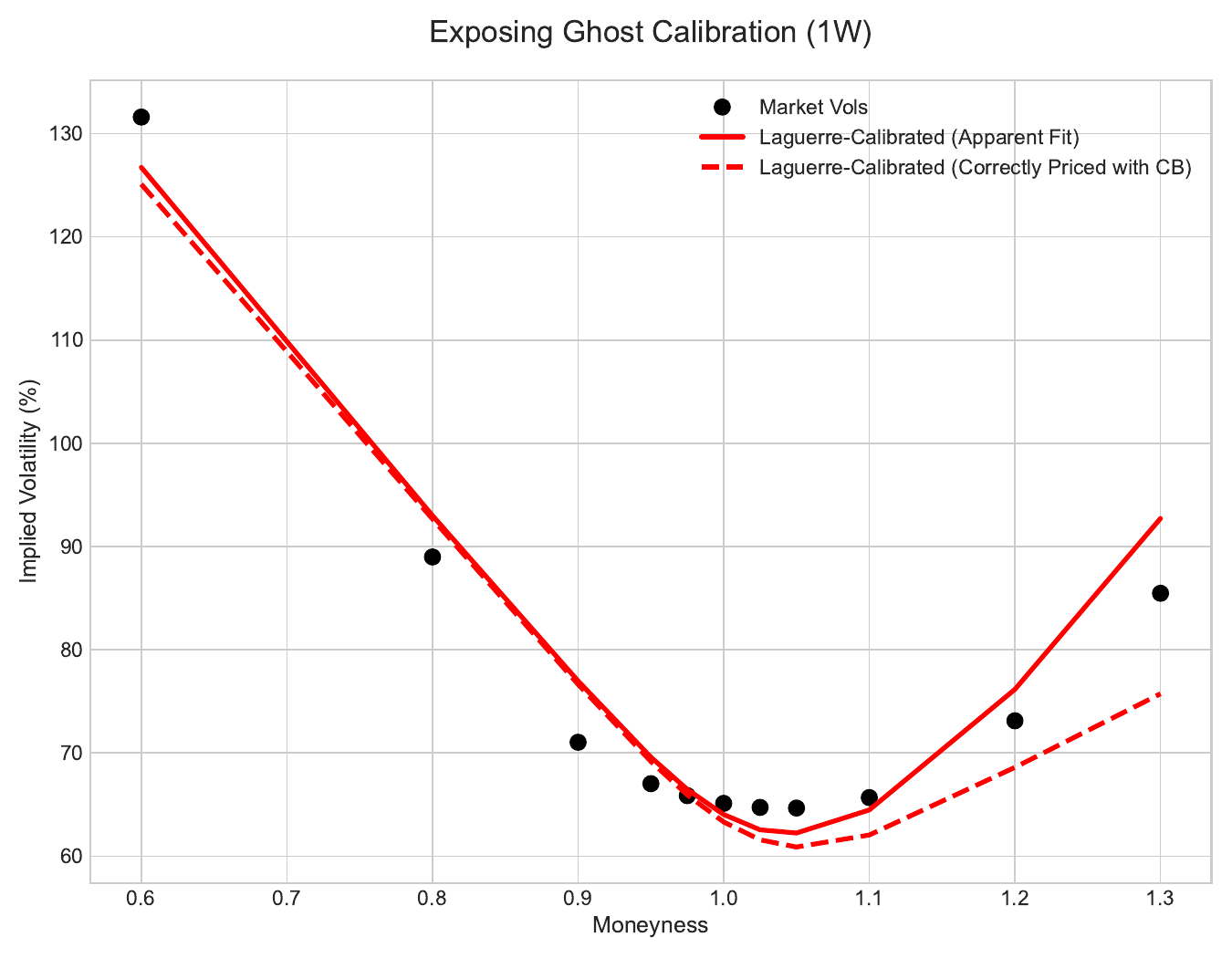}\hfill
	\includegraphics[width=0.5\textwidth]{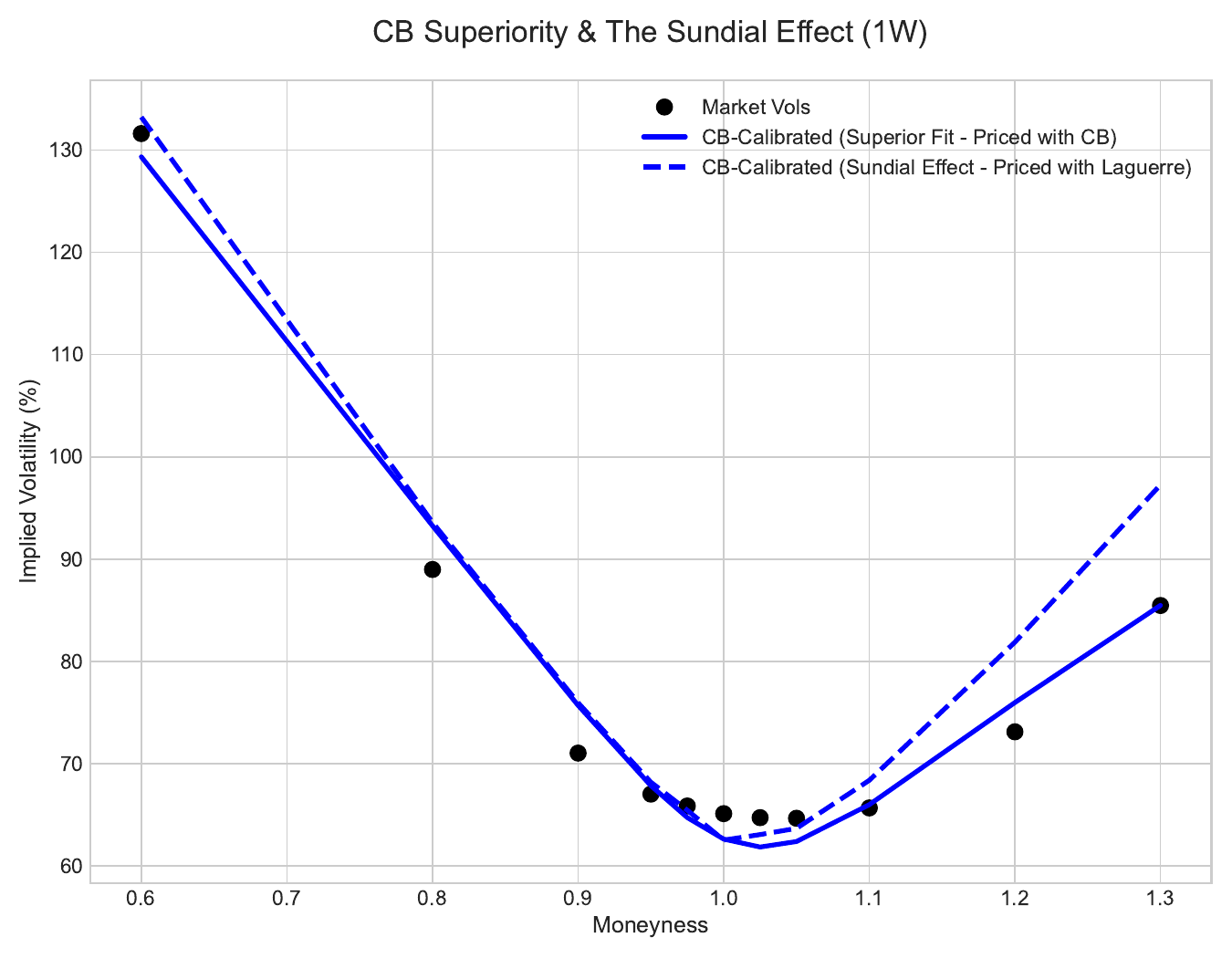}
	\caption{Calibrated 1W smiles with the rough Heston model to 1W and 2W TSLA implied vols  on 2~May~2025, using Gauss-Laguerre with 200 nodes and 1000 Adams time steps. Left panel: GL pricer produces a model that appears to fit the market (solid red line); the dashed red line shows the actual volatility curve calculated from the same parameters. Right panel: GL pricer fails to price accurately the superior fit to the market data. The CB-calibrated parameters are in \eqref{params:rough_short}, the GL-calibrated ones are in \eqref{params:rough_short_laguerre}.}
	\label{fig:ghost_sundial_calib}
\end{figure}

\sbr

To demonstrate that these numerical pitfalls are not unique to short-dated options, we replicate the experiment for the 7-week expiry. Figure~\ref{fig:ghost_sundial_calib_7W} illustrates the ghost calibration and sundial effects at this longer horizon.  
The Gauss-Laguerre-calibrated parameters for the longer expiries are 
\begin{equation}\label{params:rough_long_laguerre}
	(\alpha,\gamma,\theta,\sigma,\rho,v_0)^{long}_{rough;\text{GL}} = (0.530213,\,1.59172,\,0.216243,\,1.49494,\,,-0.519145,\,0.5342681)\,.
\end{equation}
Similar to the 1-week expiry, using the fixed Gauss-Laguerre setting produces an apparent fit (ghost calibration) that breaks down when the true volatility curve is evaluated with our precise SINH-CB method (left panel). Conversely, when the model is properly calibrated using SINH-CB, the GL pricer dramatically misprices the right tail, once again exhibiting the sundial effect (right panel). This confirms that accurate pricing mechanisms are strictly necessary across the entire term structure.

\begin{figure}[t]
	\centering
	\includegraphics[width=0.5\textwidth]{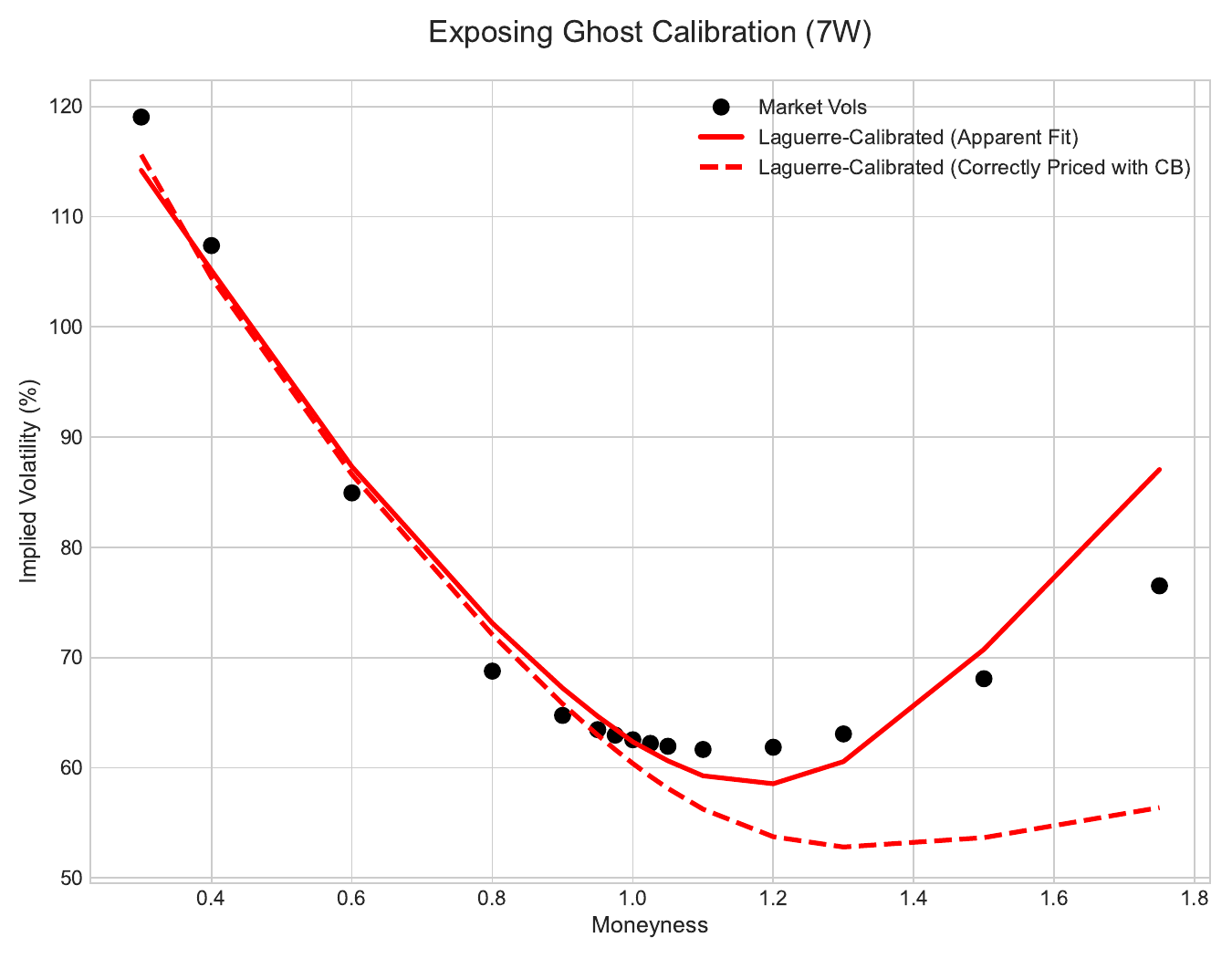}\hfill
	\includegraphics[width=0.5\textwidth]{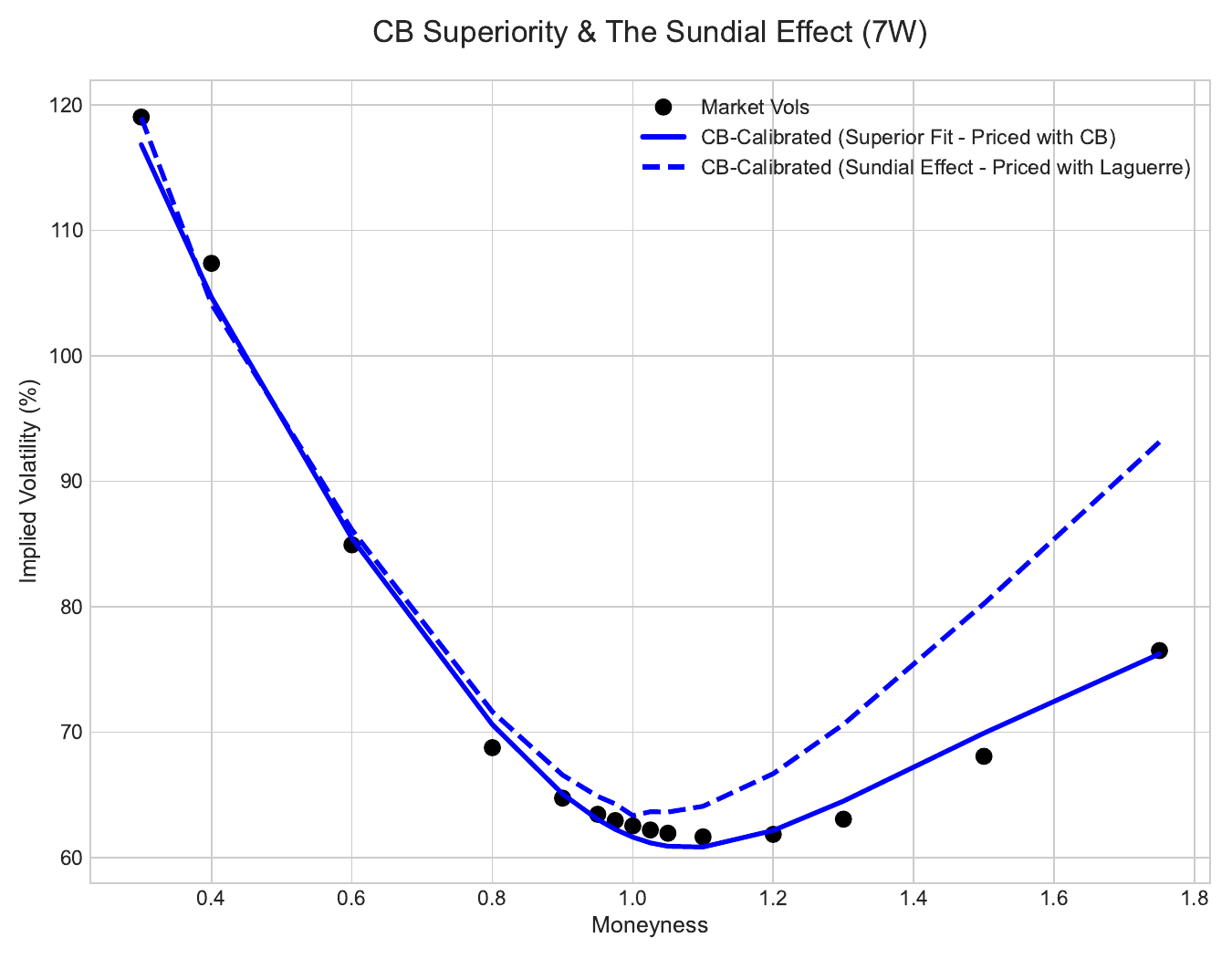}
	\caption{Calibration pitfalls at the 7-week expiry (cf. Figure \ref{fig:ghost_sundial_calib}). Left panel: The ghost calibration effect persists at longer maturities, where numerical errors in the GL pricer mask a poor fit. Right panel: The sundial effect demonstrates the GL pricer's failure to accurately evaluate the superior CB-calibrated model in the right tail. The CB-calibrated parameters are in \eqref{params:rough_long}, the GL-calibrated ones are in \eqref{params:rough_long_laguerre}.}
	\label{fig:ghost_sundial_calib_7W}
\end{figure}

\section{Conclusion}\label{s:concl}
In the paper, we analyzed sources of errors of several popular methods of Fourier inversion,
  and compare their performance in application to the Heston model and rough Heston model.
  The pricing method and analysis of other methods are quite general and applicable to wide classes of models where the characteristic functions can be calculated explicitly or numerically,
   to the Bates model and its rough analog in particular.  If the characteristic function can be calculated only numerically, then 
  an additional source of errors appears, and the errors increase as the maximal absolute value
  of the spectral parameter $\xi$ that is used in the quadrature for the Fourier inversion increases.
  This can be a serious problem in affine SV models. The difficulties are more serious in the case
  of the rough Heston model because the characteristic function is expressed in terms of
  the solution of fractional Volterra equation. We analyzed errors of solutions of various modifications of the Adams method, and constructed a new modification. We document an interesting effect. 
  Due to inherent instability of solutions of the fractional Volterra equation, if one variation of the Adams method is used, the numerically calculated characteristic function starts decreasing extremely fast as $\xi\to\infty$ along the contour of integration and one observes
  excellent albeit fallacious convergence. If another variation is used, then the absolute value
  of the characteristic exponent starts to increase and one is forced to either use a small number of terms $N$ in
  the quadrature for the evaluation of the characteristic exponent or significantly increase the number of terms
 $M$ in the Adams method or its variation. In the first case, the accuracy decreases, and in the second one - the CPU time increases. In practice, a time-extensive procedure is never used, hence, the final result depends on the choice of $N$ and $M$. We can say that for pricing in complicated SV models, the uncertainty principle holds.
 
 To control the errors in difficult situations where the theoretically sound error control is unavailable
or too complicated for practical purposes, we suggest the \emph{Conformal Bootstrap principle} (CB principle).
We use the principle in applications to the Heston and rough Heston models; the principle can be used
in applications to other models as well. Using the methods developed in the paper and the CB principle, we address the outstanding problem
\emph{Markovian or non Markovian} - which class of models fits the empirical data best. 
We produce examples that demonstrate that drawbacks of the rough Heston model documented in the literature can be artifacts of insufficiently accurate calculations. 

We calibrated the rough Heston model to the empirical data for TSLA on a hectic day in the markets.
The AVE pricing error is much lower across the board than for regular Heston, lying below 3\% in-sample and remaining around 2\%  out-of-sample even when extrapolated to five months, an order of magnitude smaller than the regular Heston model over these same expiries. The model demonstrates excellent and stable performance across the volatility surface, in our example.

In the paper, we constructed an accurate and fast SINH-CB method for the pricing of European options in
the rough Heston model, amenable to automatic and efficient error control. The main ingredients are:
1) an appropriate  conformal deformation of the contour of integration in the Fourier inversion formula,
followed by the corresponding conformal change of variables and application of the simplified trapezoid rule
(sinh-acceleration); 2) a modification of the fractional Adams method, which is a crucial improvement in the presence of a large spectral parameter; 3) an ad-hoc Conformal Bootstrap principle: if the prices obtained with two deformations
differ by less than $10^{-m}$, where $m\ge 5$, then the probability (understood in the colloquial sense) that either of the prices differs from the correct price (calculated using a perfect pricer)
by more than $10^{-m+2}$ is negligible.\footnote{Ingredients 1) and 3) can be used in calibration procedures of
any model where the Fourier transform technique is used} Using the pricer, we constructed a novel fast and accurate calibration scheme,
and applied the scheme to calibration of the rough Heston model to the real data. 
The calibration scheme satisfies all the requirements needed for practical applications.
  The pricing needs to take calculation times of the order of a  millisecond in a production environment, which usually corresponds to a few tens of a millisecond, when implemented in Python on a home PC, without compilation or optimisation. Moreover, as noted in \cite{HestonCalibMarcoMeRisk}, the calculation needs to be extremely accurate, in order to cope with both very long and very short maturities, and with options which are very far in and out of the money, for a wide range of model parameters that can result from calibration to market data. This is especially true in the context of regulatory counterparty credit risk, where exposure profiles are effectivised\footnote{For capital requirements, the ``Effective Expected Positive Exposure'' is set equal to $\text{EEPE} = \sup_{t \in [0, T]} \mathbb{E} \left[ \max(V_t, 0) \right]$, where $V_t$ is the time-$t$ portfolio value.} and therefore any large errors which occur during the Monte Carlo simulation can propagate across time.
 In recent years, a growing body of literature has explored the use of machine learning (ML) techniques, particularly deep neural networks, to accelerate pricing and calibration under rough volatility. Horvath, Muguruza, and Tomas~\cite{HorvathMuguruzaTomas2021}, as well as Bayer, Horvath, and Stemper~\cite{BayerHorvathStemper21}, have introduced deep learning  approaches to approximate the pricing map in rough Bergomi and rough Heston models, which are sufficiently fast for real-time calibration. However, such methods suffer from one common drawback: in addition to speed, banks must ensure accuracy, transparency, and stability for such methods to be used in pricing, hedging, and risk management, especially in models subject to regulatory approval. For example, the Federal Reserve's SR 11-7 regulatory standard for model validation, which is used in most sell-side institutions, asks to perform critical analyses to determine the model's assumptions and limitations, as well as to \emph{``establish the boundaries of model performance by identifying the acceptable range of inputs as well as conditions
	under which the model may become unstable or inaccurate''} \cite{SR11_7}. Clearly, this is extremely challenging for a black-box pricer. 
		
	The fast reliable 
 SINH-CB method (sinh-acceleration - conformal bootstrap) constructed in the paper is significantly faster than other methods and satisfies all the requirements above whereas popular methods are either too slow or inaccurate or unreliable.
  In our numerical experiments, the SINH-CB method demonstrated a sufficiently good accuracy for calibration purposes if the OTM option prices larger than $E-07$ of the spot price $S_0$ were used.
 Some of the popular methods work fairly well in a narrower region of the strike-maturity $(K,T)$ - plane \emph{
  provided the parameters of the scheme are chosen correctly}  but are slower and not so reliable;
 other methods are rather inaccurate in wide regions in the $(K,T)$ plane, which leads to very inaccurate
prices and implied volatility curves. Following Leo Tolstoy (``All happy families are alike; each unhappy family is unhappy in its own way"), we 
 can formulate the \emph{Anna Karenina principle for option pricing}: in a good region of $(K,T)$-space, all reasonable models and pricing methods are alike; close/far from maturity and far in the tails, models and pricing methods perform differently

Our numerical experiments demonstrate that SINH-CB method satisfies all the requirements better than
other methods, followed by a simple modification of the standard Fourier inversion method (Flat iFT with the BM correction); the Gauss-Laguerre method requires 1.3-5 times more terms than SINH-CB to achieve the accuracy sufficient for applications. If the number of nodes needs to be increased, then all values of the integrand needs to be recalculated whereas in the case SINH-CB, only a small number of additional terms need to be calculated.
In the case of the rough Heston model, the evaluation of the integrand at chosen nodes is very time consuming, hence,
SINH-CB has an additional advantage. We explained the sources of instability of the Gauss-Laguerre quadrature and
COS and SINC method. The Gauss-Laguerre quadrature is potentially unstable (and not mathematically justified)
if the rate of decay of the integrand is small, which may happen for short maturity options in the Heston model
(presumably, in the rough Heston model as well). COS and SINC method use unreliable recommendations for the choice of the truncation parameter and cannot be accurate for options of short maturity, especially the deep OTM options.

We also showed that even if the pricer is sound, one cannot hope to use the same parameters of the numerical scheme
for all $(K,T)$ in the data set and all parameters of the model; the number of terms in the Fourier inversion formula
and step in the modified Adams method needed to satisfy the desired error tolerance can significantly increase and decrease, respectively.
We demonstrated that an unstable and/or inaccurate pricer produces
spurious wings of volatility curves, and the 
 shape of the surface may strongly depend on the choice of the parameters of 
the numerical scheme.
We can formulate  {\sc The Uncertainty Principle of calibration}: using different parameters of the numerical scheme,
one can produce a host of different prices and volatility curves and surfaces, and choose shapes one likes better.

Finally, using the calibration method developed in the paper, we showed that the rough Heston model
with a very small Hurst index $H=0.012$ gives a very good fit (both in  and out of the sample) calibrated to options on Tesla stock of maturities 1W-2W; the performance remains equally good out of the sample (maturities 3W and 7W). 
We also demonstrated how pricing algorithms with fixed, seemingly conservative numerical settings can generate erroneous calibration parameters. These numerical errors can, in turn, lead to the incorrect dismissal of a valid model.
We further calibrated the model directly to a broad range of maturities, from three weeks to 1.7 years, and found that it continues to fit the market well, across the entire term structure, with in-sample errors of around 1\%. By contrast, the regular Heston model, calibrated to the same data, fits markedly worse and extrapolates poorly across maturities, in our example, degrading from roughly 4\% in-sample to 10\%--17\% out-of-sample, against the rough model's stable ~2\% fit. This indicates that the rough Heston model's good performance is not confined to the short end, but holds across the volatility surface, at least for this data set.

\appendix

\section{Pseudo-code implementation of the BL modification of the Adams method }\label{a:mod_2}
 We use the following definitions
\begin{align*}
	\Delta &:= \frac{T}{M}, \qquad t_k := k\,\Delta,\quad k=0,\dots,M,\quad
	\mathrm{abs}\,\xi := 1+|\xi|,\\
	\tilde h(\xi,t) &:= \frac{h(\xi,t)}{1+|\xi|}, \qquad
	\tilde h_{\mathrm{as}}(\xi,t) := \frac{-\tfrac12(\xi^2+i\xi)}{1+|\xi|}\,\frac{t^{\alpha}}{\Gamma(\alpha+1)},\\
	F(\xi,h) &:= -\tfrac12(\xi^2+i\xi) + \gamma\,(i\xi\rho\nu - 1)\,h + \tfrac{(\gamma\nu)^2}{2}\,h^2,\\
	\tilde F_{\mathrm{as1}}(\xi,\tilde h_{\mathrm{as}},\tilde h_1)
	&:= \gamma\,(i\xi\rho\nu - 1)\,(\tilde h_{\mathrm{as}}+\tilde h_1)
	\;+\; (1+|\xi|)\,\frac{(\gamma\nu)^2}{2}\,(\tilde h_{\mathrm{as}}+\tilde h_1)^2.
\end{align*}
The description, which can be found in the panel  on the next page, is provided for illustration purposes only and does not include any optimizations or parallelizations. 


{\small 
\begin{algorithm}[H]
	\caption{Rough Heston CF via Fractional Adams -- BL modification}
	\begin{algorithmic}[1]
		\Require $\alpha\in(0,1)$, $\gamma>0$, $\theta>0$, $\rho\in(-1,1)$, $\nu>0$, $v_0>0$, maturity $T$, steps $M$, Picard iterations $n$, frequency grid $\{\xi_m\}_{m=1}^{N}$
		\State $\Delta \gets T/M$; \quad $t_k \gets k\,\Delta$ for $k=0,\dots,M$
		\State \textbf{function} $a_{{\rm unif}}(\alpha,\Delta,k)$ returns Adams weights $\{a_{j,k+1}\}_{j=0}^{k}$ for uniform grid
		\State $r \gets \Delta^\alpha/\Gamma(\alpha+2)$ \Comment{$r = a_{k+1,k+1}$ on a uniform grid}
		
		\State \textbf{for all} $\xi$ in grid \textbf{do} \label{line:init-loop}
		\State $\mathrm{abs}\,\xi \gets 1+|\xi|$
		\For{$k=0,\dots,M$}
		\State $\tilde h_{\mathrm{as}}(\xi,t_k) \gets \displaystyle \frac{-\tfrac12(\xi^2+i\xi)}{\mathrm{abs}\,\xi}\,\frac{t_k^{\alpha}}{\Gamma(\alpha+1)}$
		\EndFor
		\State $\tilde h_1(\xi,t_0) \gets 0$
		\State $FF(\xi,0) \gets \tilde F_{\mathrm{as1}}\!\big(\xi,\tilde h_{\mathrm{as}}(\xi,t_0),\tilde h_1(\xi,t_0)\big)$
		\State \textbf{end for}
		
		\For{$k=0,\dots,M-1$}
		\State $a_{0:k,k+1} \gets a_{{\rm unif}}(\alpha,\Delta,k)$ \Comment{vector of weights $[a_{0,k+1},\dots,a_{k,k+1}]^\top$}
		
		\State \textbf{for all} $\xi$ in grid \textbf{do} \Comment{predictor}
		\State $\tilde h_0(\xi) \gets \sum_{j=0}^{k} a_{j,k+1}\; FF(\xi,j)$
		\State $z \gets \tilde h_0(\xi) + r\;\tilde F_{\mathrm{as1}}\!\big(\xi,\tilde h_{\mathrm{as}}(\xi,t_{k+1}),\tilde h_1(\xi,t_k)\big)$ \Comment{initial guess}
		\For{$m=1,\dots,n$} \Comment{corrector}
		\State $z \gets \tilde h_0(\xi) + r\;\tilde F_{\mathrm{as1}}\!\big(\xi,\tilde h_{\mathrm{as}}(\xi,t_{k+1}), z\big)$
		\EndFor
		\State $\tilde h_1(\xi,t_{k+1}) \gets z$
		\State $FF(\xi,k{+}1) \gets \tilde F_{\mathrm{as1}}\!\big(\xi,\tilde h_{\mathrm{as}}(\xi,t_{k+1}),\tilde h_1(\xi,t_{k+1})\big)$ \Comment{cache for next step}
		\State \textbf{end for}
		\EndFor
		
		\State \textbf{for all} $\xi$ in grid \textbf{do} \Comment{recover unscaled $h$ on nodes}
		\For{$k=0,\dots,M$}
		\State $\hat h(\xi,t_k) \gets (1+|\xi|)\,\big(\tilde h_{\mathrm{as}}(\xi,t_k)+\tilde h_1(\xi,t_k)\big)$
		\State $G_k(\xi) \gets \gamma\theta\,\hat h(\xi,t_k) \;+\; v_0\,F\!\big(\xi,\hat h(\xi,t_k)\big)$
		\EndFor
		\State $I(\xi) \gets \Delta\Big(\tfrac12 G_0(\xi) + \sum_{k=1}^{M-1}G_k(\xi) + \tfrac12 G_M(\xi)\Big)$ \Comment{trapezoid rule}
		\State $\Phi(\xi,T) \gets \exp\big(I(\xi)\big)$
		\State \textbf{end for}
		\Ensure $\{\Phi(\xi_m,T)\}_{m=1}^{N}$
	\end{algorithmic}\label{algo_mod2}
\end{algorithm}
}

\section{Benchmark pricing algorithm}\label{a:pricing_algo_bm}
For a given pricing configuration $X$, i.e. a set of model parameters, option strike and expiry, and underlying spot level, the numerical parameters of the pricing algorithm are the $\om$ in the definition
of the sinh-acceleration (the slope of the asymptote of the deformed contour in the right half-plane), the number of discretization time steps $M$,  the number of iterations $n$ in the modified Adams method,  the mesh $\zeta$ and the number of terms $N$ in the simplified trapezoid rule (see Section \ref{ss:Flat iFT and simpl. trap}). Of all these, $\omega$ plays a special role since its choice determines the contour of integration along which the characteristic function is calculated, and hence determines the value of all other parameters. In addition, while it is clear that lower values of $M$, $N$ and $n$, as well as higher values of $\ze$, correspond to lower CPU times, and decreasing accuracy, there is no such monotonic dependency on $\om$ for either.  Therefore, for the calculation of the ``best'' $\omega$, an optimization needs to be used to maximise pricing accuracy. We will use an objective function of $\omega$, $F(\om; X)$, which decreases for increasing pricing accuracy. 
For each configuration \( X \), we evaluate the objective function \( F(\omega; X) \) over a discrete grid\footnote{Here $\om_1$ is not to be confused with the parameter denoted with the same notation in the sinh-conformal map (cf. eq. \eqref{eq:sinh}).} 
$\Omega = \{\omega_1, \omega_2, \ldots, \omega_m\}$. 
Let
\[
\omega^{(0)} := \arg\min_{\omega_j \in \Omega} F(\omega_j; X)
\]
\noindent denote the grid point that minimizes \( F(\cdot; X) \). This value \( \omega^{(0)} \) is used as the initial guess for a one-dimensional Nelder--Mead optimization, which yields a refined estimate \( \omega_{\mathrm{best}} \). We record both \( \omega_{\mathrm{best}} \) and the associated pricing results for subsequent analysis and for training of the machine learning models.
\subsection{Choice of $\omega$ grid}\label{sss:data_generation_om_grid}
We give preference to values of $\om$ that are small in absolute value, e.g., for the $\om > 0$ strip case, $\om \le 0.2$, since such choices of $\omega$ tend to result in smaller $M$. Therefore we take $m = 10$ and $\Omega = \{\omega_1, \omega_2, \ldots, \omega_{10}\}$, with\footnote{See Section \ref{sss:data_generation_om_calc} regarding the choice of $\om_{10}$.} $\om_1 = 0.002$, $\om_6 = 0.2$, $\om_{10} = \pi/4 - 0.05$, with the intermediate points $\om_2$ to $\om_5$, and $\om_7$ to $\om_9$, equally spaced between $\om_1$ and $\om_6$, and between $\om_7$ and $\om_{10}$, respectively. During the optimisation, we will let $\omega$ vary, in the $\om > 0$ strip case, between a small nonzero value $\om_0 = 0.0001$ and $\om_{11} = \pi/4 - 0.0001$.
\subsection{Calculation of $F(\omega; X)$}\label{sss:data_generation_F_calc}
For each   $\om \in \Om$, we calculate $F(\om; X)$ as follows. First, if it has not already been calculated, we calculate the benchmark price $V_{LL}(T,K)$ corresponding to $\om = 0$, i.e., after the sinh transformation, by integrating along the line $\Im \xi = -1/2$ in Fourier space, which is the analogue of the Lewis-Lipton formula. This is done using the procedure in section \ref{sss:data_generation_V_calc} below.

Then, for each $\om\in \Om$, we calculate $V(\om; T, K)$, by using the same procedure with this value of $\om$. The value of $F(\om)$ is calculated as
\begin{equation}\label{e:objfun}
	F(\om) = (V_{LL} - V(\om))^2,
\end{equation}

\subsection{Calculation of the benchmark prices $V(T, K; \om)$}\label{sss:data_generation_V_calc}
For each value of $\om$, including $\om = 0$, we use the following algorithm.
\begin{enumerate}[1.]
	\item 	We calculate the parameters $\omega_1$ and $b$ of the sinh-deformation and the mesh $\ze$ according to the recommendations in section \ref{ss: SINH},	where, for the put case, we take $\mup = 1$, $\mum = 0$, $\gamma_\pm = \pm \pi/4$, $d = 0.9\cdot\min(\gap  - \om, \om - \gam)$.  We set $\ze = 2\pi d / \log(100/\eps)$, where $\eps = 10^{-16}$ is our error tolerance. 
	\item We use the ad-hoc procedure described in Section \ref{ss:asymp_psi} to calculate the truncation parameter $\La = N\ze$. Set $M = 1000$. \label{LaChoice}
	\item Calculate an initial price $V_0$ using the procedure in section \ref{ss: SINH}, with $n = 2$. 
	\item In a loop $j =1, 2, \ldots, 10$, we check if the value of $V_0$ diverges, since the ad hoc recommendation at point \ref{LaChoice} often results in too large values of $\La$, which can cause division-by-zero errors in numerical calculations. In that case, we replace $\La \mapsto 0.8 \cdot \La$ and re-price.
	\item Fix a tolerance $\eps_V = V_0/10000$.
	\item Search for $n$.  In a similar loop, calculate a new price $V_1$ and, while $|V_1 - V_0| \ge \eps_V$, and $V_1$ does not diverge, increase $n\mapsto n+1$ and set $V_1$ equal to the old price $V_0$. If $\om = 0$, then we take the last $n$, since this value of $\om$ is used to calculate $V_{LL}$, otherwise we reduce $n$ by 1. In almost all cases, for $\om \neq 0$, $n = 2$ is used. 
	\item Search for $M$. We use a similar  loop to increase $M$, if needed, by a factor $\ka_M =  1.5$ each time, up to a maximum of 4000 for the calculation of $V_{LL}$ and 2500 otherwise\footnote{Larger values of $M$ than 2500 are sometimes needed, of course, however those would make a practical computation very slow.}.
	\item Choice of $\ze$ and $\La$. We use similar loops, first for the mesh $\ze$ (with factors $\ka_\ze= 0.5$ for the calculation of $V_{LL}$ and 0.8 otherwise) and the truncation parameter $\La$ (with factors $\ka_\La= 1.2$ for the calculation of $V_{LL}$ and 1.1 otherwise). Care must be taken to restore the previous value of $\La$, i.e. $\La \mapsto \La/\ka_\La$, if division by zero error occurs. 
	\item If the loop over $\La$ resulted in its value being increased, then we carry out a last loop over $\ze$. 
	\item Since the algorithm above can result in too large a value for $M$, if $\om \neq 0$ we use the following approach to check if $M$ can be reduced. 
	\begin{enumerate}[a.]
		
		\item In a loop, reduce $M$ by a factor of  $\kappa'_M = 0.8$, as long as $|V - V_\mathrm{previous}|/V_\mathrm{previous} < 10^{-5}$.
		\item Take the last working value of $M$ before the check failed.
	\end{enumerate}
\end{enumerate}
\subsection{Calculation of $\om(X)$}\label{sss:data_generation_om_calc}
We store the values of $F(\om_j)$, $\om_j \in \{\om_1, \om_2, \ldots, \om_{10}\}$ in a vector and pick 
\[
\omega^{(0)} := \arg\min_{\omega_j \in \Omega} F(\omega_j; X)
\]
to be the starting point of a 1D Nelder-Mead optimisation\footnote{Using {\tt scipy.optimize.minimize}.} with function tolerance $\eps_f = 0.001$, argument tolerance $\eps_\om = \min(8\cdot 10^{-4}, 2\omega^{(0)}/100 )$, maximum number of iterations set to 20, and maximum number of objective function evaluations set to 30. Call $\om_\mathrm{best}$ the result of the optimisation. 
\sbr

After the optimization completes, we check:

\begin{itemize}
	\item  If $|V(\om_\mathrm{best}) - V_{LL}|/V_{LL} < 10^{-4}$, then we store, for the configuration $X$, the value of $\om_\mathrm{best}$, the price $V(\om_\mathrm{best})$, and the other numerical settings.
	\item Otherwise, we take the following alternative value of $\om$
	\[
	k_{\mathrm{alt}} := \operatorname*{arg\,max}_{k \in \{0,1\}} 
	\left| \omega_{k} - \omega_{\mathrm{best}} \right|, \quad
	\omega^{(\mathrm{alt})} := \omega_{k_{\mathrm{alt}}},
	\]
	i.e., out of the two ``best omegas'' in $\Om$, we pick the one that differs the most from $\omega_{\mathrm{best}}$. If $V(\omega^{(\mathrm{alt})})$ passes a similar check w.r.t. $V_{LL}$, then we store $\omega^{(\mathrm{alt})}$, as well as the corresponding price and numerical settings.
	\item If neither of the previous checks works, then we compare $V(\om_\mathrm{best})$ and $V(\omega^{(\mathrm{alt})})$ against each other. This can be useful on rare occasions when $V_{LL}$ converges very slowly and is difficult to calculate with high precision.
	\item Finally, if none of the previous checks work, then we return an error\footnote{In our experiments, this has never been observed to happen.}. 
\end{itemize}

\pagebreak
\section{Additional tables and figures}\label{s:addition}

\begin{figure}[h!]
	\centering
	\begin{subfigure}{0.49\textwidth}
		\centering
		\includegraphics[width=\linewidth]{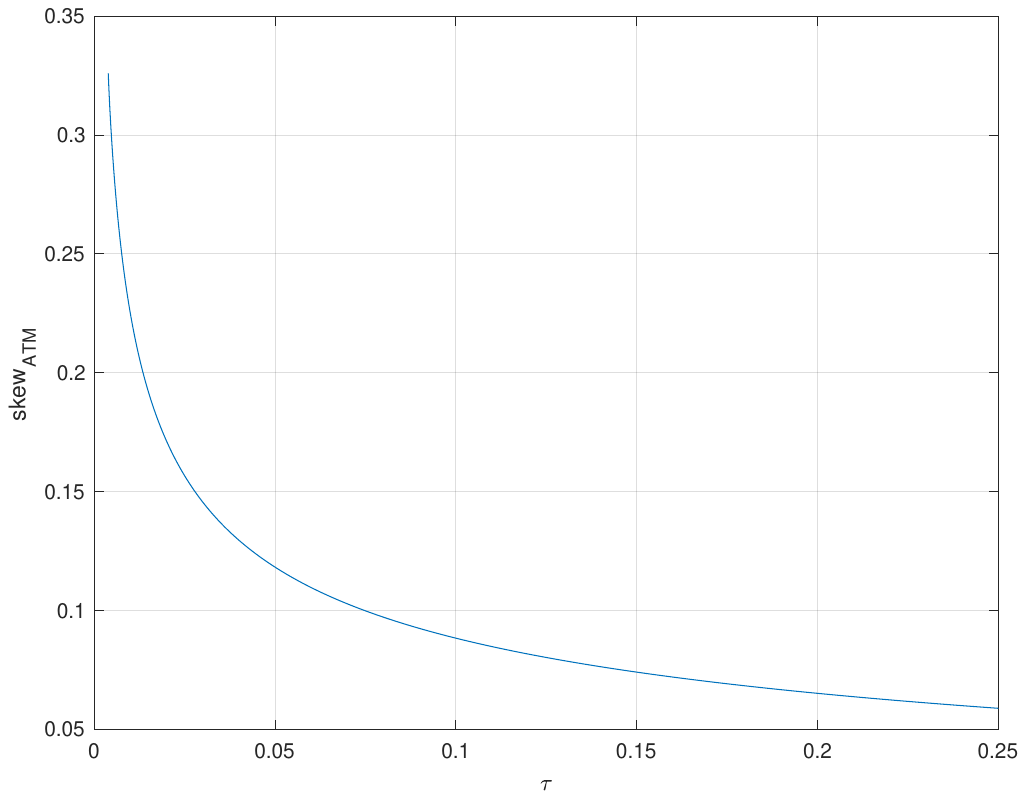}
		\caption{}\label{SkewSet1SinhXi}
	\end{subfigure}\hfill
	\begin{subfigure}{0.49\textwidth}
		\centering
		\includegraphics[width=\linewidth]{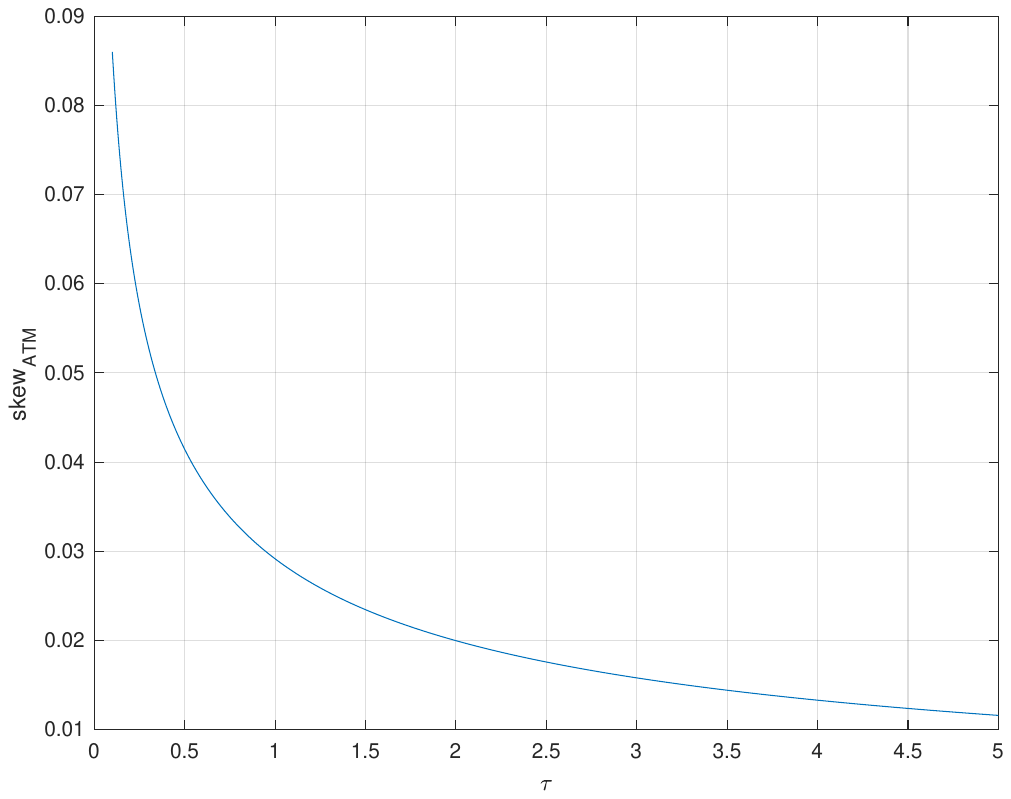}
		\caption{}\label{SkewSet1SinhXiLong}
	\end{subfigure}
	\caption{ATM skew; the parameters are in \eq{parEuRos}.}
	\label{Set1Skews}
\end{figure}

\begin{figure}
    \includegraphics[width=1\textwidth,height=0.8\textheight] {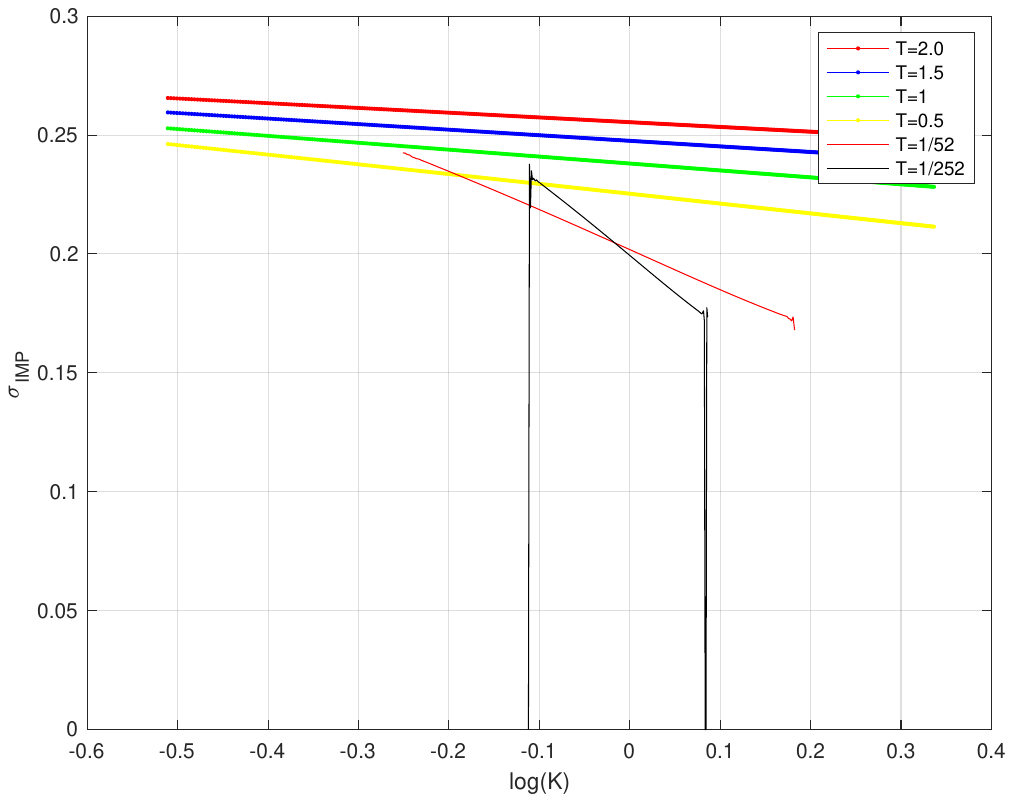}
    \caption{Implied volatility curves; the parameters are in \eq{parEuRos}. $\sg_{IMP}=0$ means that no-arbitrage condition is not satisfied.} \label{Set1Curves}
\end{figure}

\begin{figure}[htb]
	\captionsetup[subfigure]{skip=-40pt, belowskip=0pt} 
	\centering
	
	\begin{subfigure}{0.49\textwidth}
		\centering
		\includegraphics[width=\textwidth]{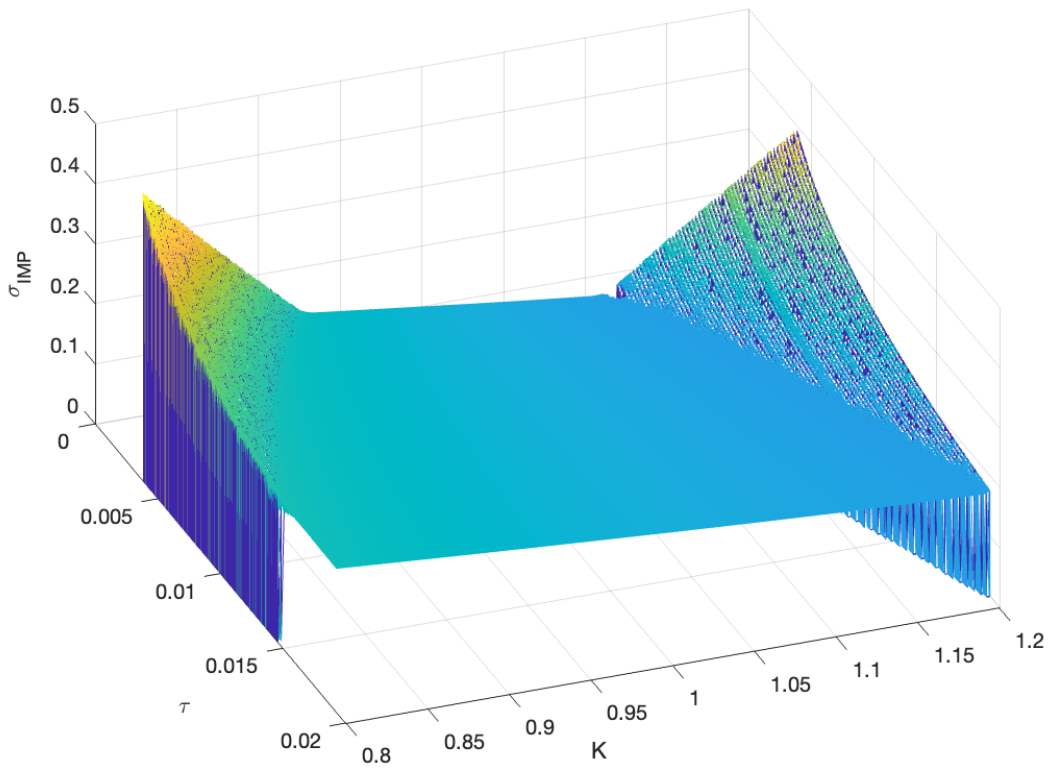}
		\caption{SINH}
		\label{Set1SurfaceB}
	\end{subfigure}
	\begin{subfigure}{0.49\textwidth}
		\centering
		\includegraphics[width=\textwidth]{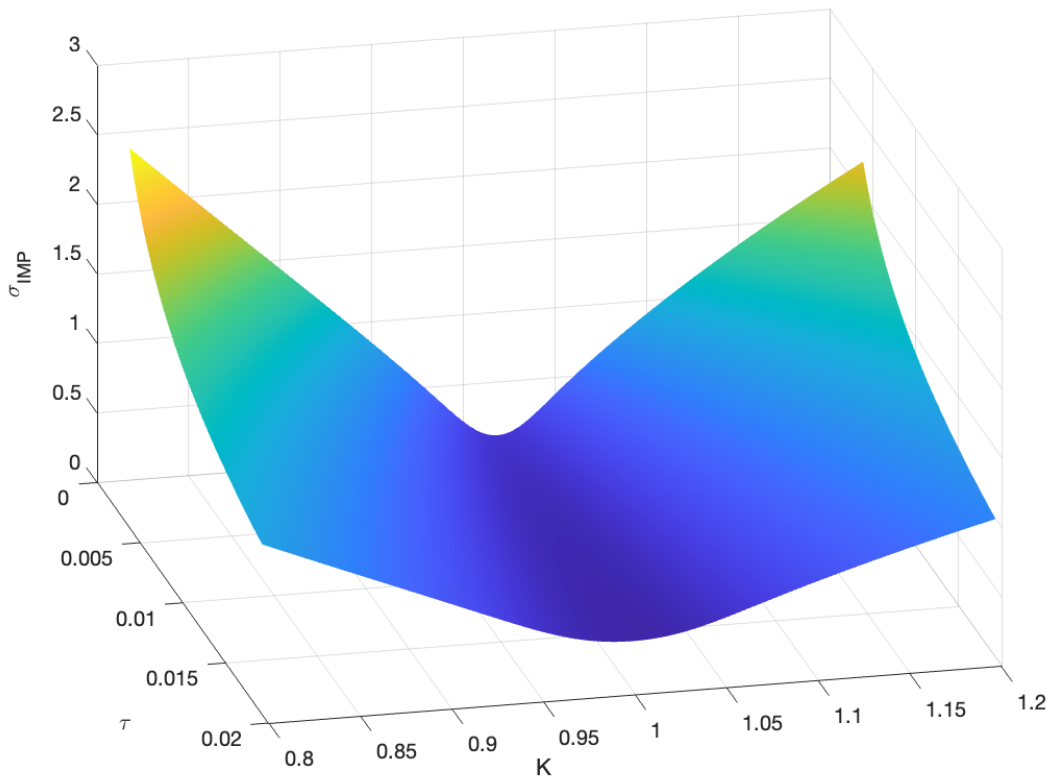}
		\caption{CM, $\omega_1=-1.1$}
		\label{Set1CMsurfaceT152om1m11}
	\end{subfigure}
	
	\begin{subfigure}{0.49\textwidth}
		\centering
		\includegraphics[width=\textwidth]{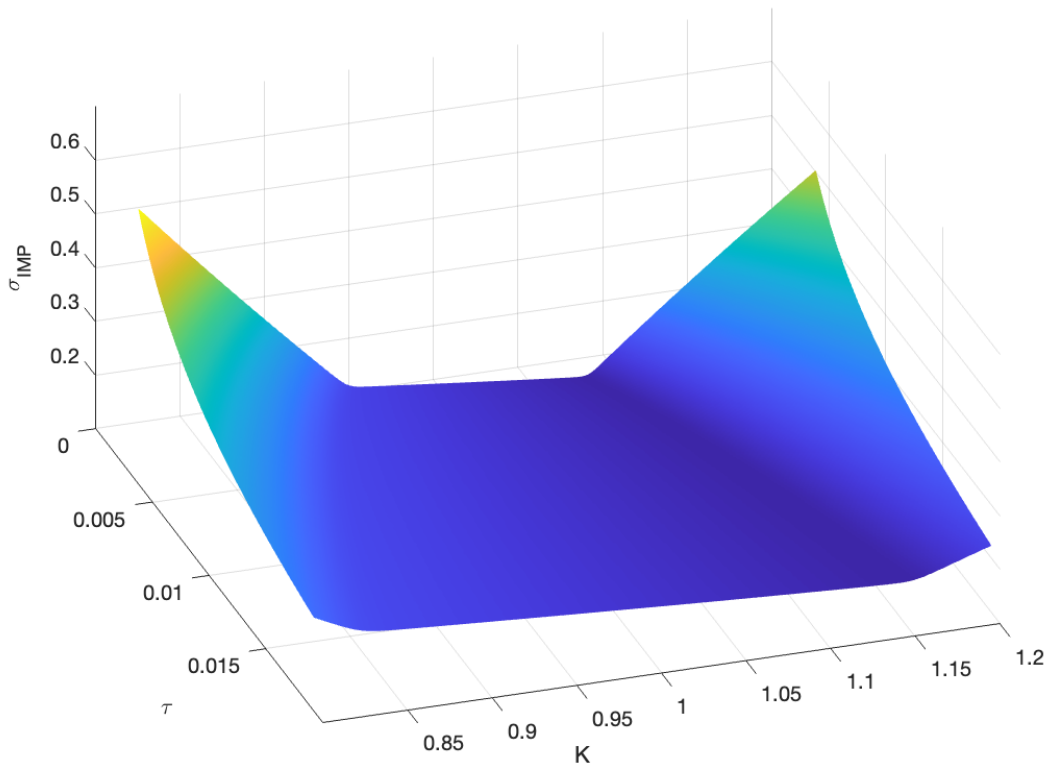}
		\caption{CM, $\omega_1=-1.5$}
		\label{Set1CMsurfaceT152om1m15}
	\end{subfigure}
	\begin{subfigure}{0.49\textwidth}
		\centering
		\includegraphics[width=\textwidth]{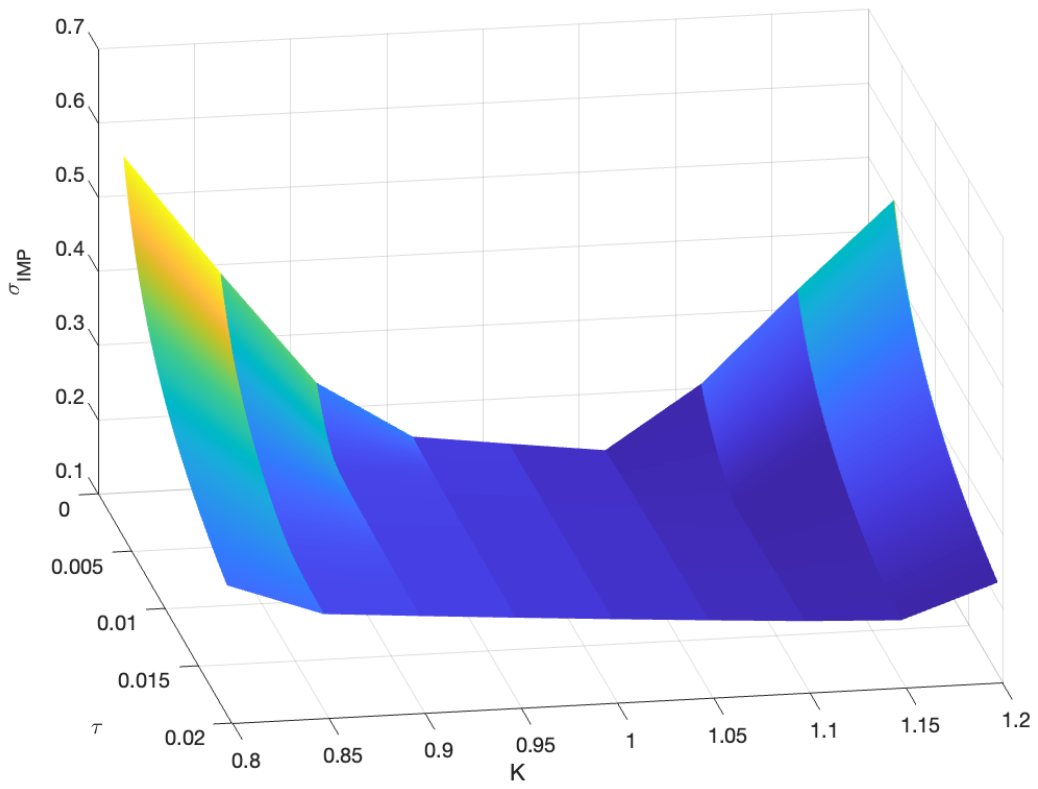}
		\caption{CM, $\omega_1=-1.5$}
		\label{Set1CMsurfaceT152om1m15R}
	\end{subfigure}

	\caption{\small
		Implied volatility surfaces in the rough Heston model \cite[Example 5.1]{EuchRosenbaum2019}, for time to maturity in the range (1 day, 1 week); spot $S_0=1$. If the price is outside the no-arbitrage bounds, $\sigma_{IMP}$ is set to 0. Panel (A): surface is calculated using the SINH-acceleration and the modified Adams method, the parameters are in \eq{parEuRos}. Irregular parts of the surface are where the OTM vanilla prices are smaller than E-10. Panels B-C: Flat iFT is used with $\zeta=0.125, N=8,192$ and $\omega_1=-1.1, -1.5, -1.5$, respectively, and the modified Adams method with $M=2000$. Irregular parts of the surface are where the OTM vanilla price is smaller than E-06. Panel (D) shows the effect of the interpolation: implied volatilities are calculated at points of a sparse grid, in the result, the interpolated surface is higher than the one on Panel (C), and the smiles are more regular.}
	\label{Set1ImpVolsurfacesXiT152}
\end{figure}

\begin{figure}
    \includegraphics[width=1\textwidth,height=1\textheight] {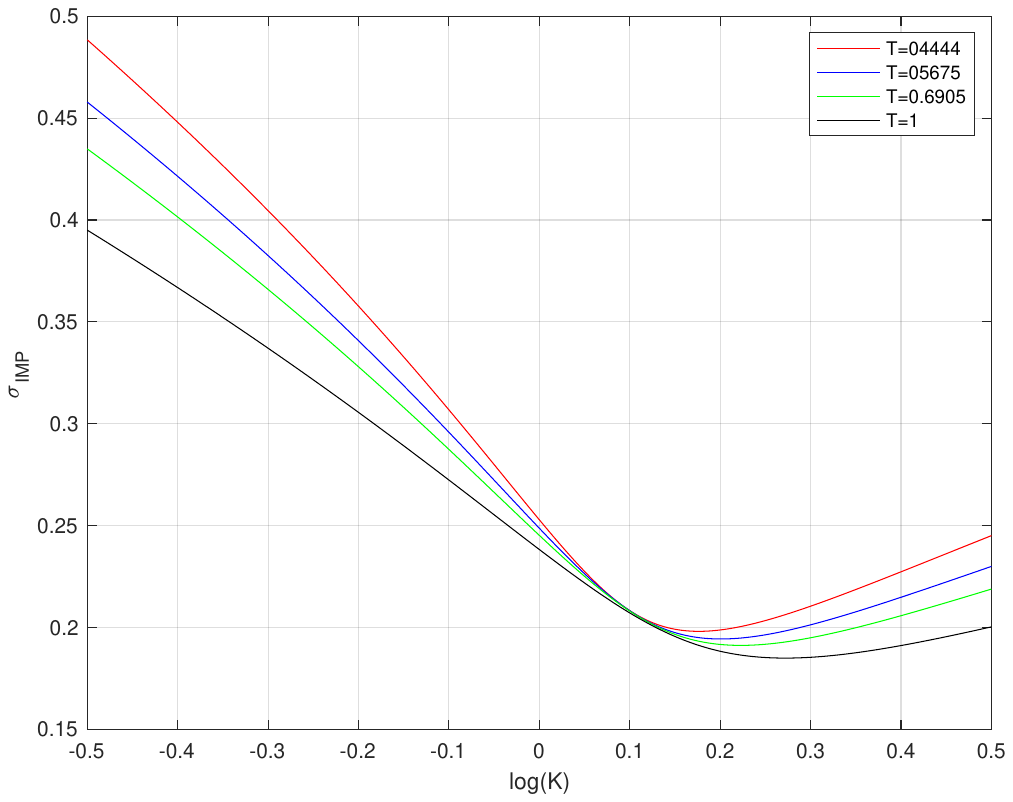}
    \vskip-4cm
    \caption{Correct implied volatility curves. The parameters  $\al=0.512,$
   $\ga=0.88,$
    $\rho=-0.7$,
    $\nu=0.96,$
    $\theta=0.016$,
    $v=0.148$, are the result of calibration to the real data in \cite[p.27]{Imperial2020}. The implied volatilities calculated using the Lewis and Adams methods and shown on Fig.~2.7 in \cite{Imperial2020} are somewhat different, in the tails especially. Note that on Fig.~2.7 in \cite{Imperial2020}, the range of log-strikes is asymmetric, and depends on maturity: $\ln K\in [-0.3, 0.35]$  for maturities $T=0.6905$ and $T=1$, and   $\ln K\in [-0.25, 0.35]$ for $T=0.4444$ and $T=0.5675$. A natural guess is that the results of calculations in the symmetric range $\ln K\in [-0.35, 0.35]$ are unsatisfactory. } \label{ImperialCurves}
\end{figure}

\begin{figure}
	\hspace*{-0.09\linewidth}
\begin{tabular}{cc}
\begin{subfigure}[h]{0.55\textwidth}

 \centering
    \includegraphics[width=1.1\textwidth]{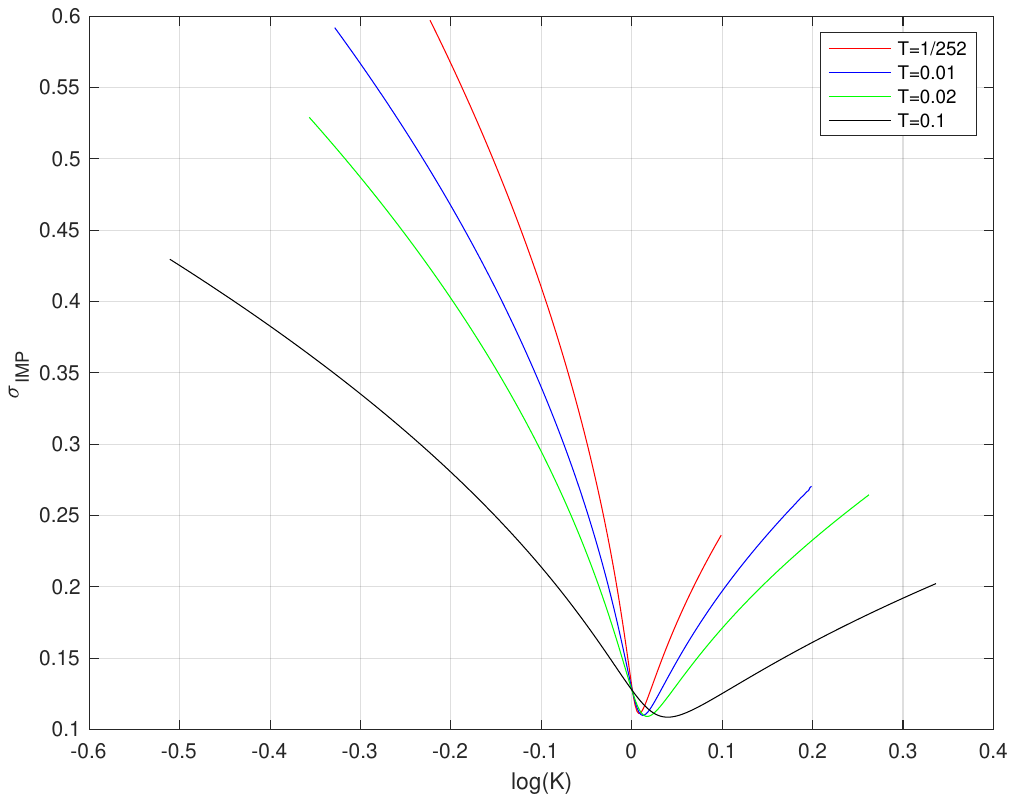}
    \caption{}\label{Set2Short}
\end{subfigure}
&
\begin{subfigure}[h]{0.55\textwidth}
\centering
    \includegraphics[width=1.1\textwidth]{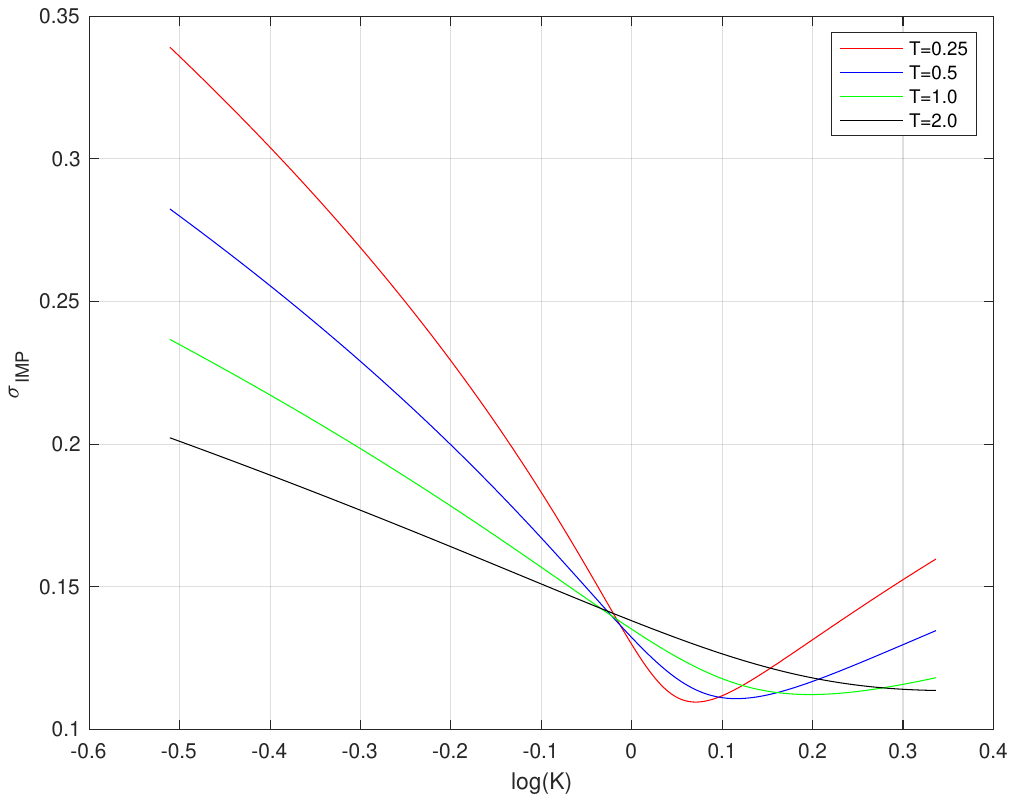}\caption{} \label{Set2CurvesLogKMod}
\end{subfigure}

\end{tabular}
\caption{Correct implied volatility curves in the rough Heston model  (Example in \cite[Sect. 6.2]{KamuranEmreErkan2020}); parameters $\al=0.6$, $\ga=2$, 	$\rho=-0.6$,
$\theta=	0.025$, $\nu=0.2$, $v_0=0.025$; $S_0=1$. 
 } 
\label{Set2Curves}
\end{figure}

\begin{table}
\caption{\small Dependence of implied volatilities (rounded) in the rough Heston model   on the numerical scheme.  
Example in \cite[Sect. 6.2]{KamuranEmreErkan2020}; parameters $\al=0.6$, $\ga=2$, $\rho=-0.6$,	
$\theta=	0.025$, $\nu=0.2$, $v_0=0.025$). Spot $S=1$, maturity $T=1/52$ years (1 week).}
{\tiny
\begin{tabular}{c|ccccccccc}
\hline\hline
$K$ & 0.8	& 0.85 &	0.9 &	0.95	& 1 &	1.05 &	1.1	& 1.15 &	1.2
\\
SINH & 0.4269 &	0.3686 &	0.3039&	0.2274 &	0.1280 &	0.1313 &	0.1687 &	0.2053 &	0.2053\\
iFT(0.25, 4096) & (*) & 0.3390 &	0.3009 &	0.2269 &	0.1280 &	0.1260 & (*) & (*) &(*) \\
FFT(0.25,4096) & (*) & (*) & 0.3000 &	0.2270 &	0.1279 &	0.1263 & (*) & (*) & (*) \\
iFT(0.125,9182) & 0.4273 &	0.3687 &	0.3039 &	0.2274 &	0.1280 &	0.1313 &	0.1687 &	0.2236 & (*)\\
FFT(0.125,9182) & (*) & 0.3539 &	0.3030 &	0.2274 &	0.1280 &	0.1315 &	0.1694 &	0.2175 &(*)\\
\end{tabular}
}
\begin{flushleft}{\tiny SINH - method of the present paper, $\om=0.2$ for puts, $\om=-0.2$ for calls.\\
iFT$(\ze,N)$: iFT with $\om_1=-0.5$ (Lewis-Lipton choice) and uniform grid, step $\ze$, $N$ terms.\\
FFT$(\ze,N)$: version of CM method based on FFT and interpolation, with $\om_1=-0.5$, step $\ze$, $N$ terms.\\
(*): price outside the no-arbitrage bounds.\\
\vskip-0.2cm
$\sg_{IMP}(1.2)$ in SINH-line is unreliable because the absolute value of the OTM option price is smaller than $10^{-12}$.
}
\end{flushleft}
\label{table: iFT-FFT}
\end{table}

\begin{table}
\caption{\small Moderate maturities, spot $S_0=1$. Relative errors (rounded) of calculations of OTM and ATM puts $(K\le 1$) and OTM calls $(K>1$)  in the rough Heston model with the parameters \eq{parEuRos} and CPU time (in msec., the average over 1000 runs)  for several numerical schemes.  
 }
{\tiny
\begin{tabular}{c|rrrrrrrrr|r}
\hline\hline
$T=2$ & & & & & & & & &  & Time\\\hline
$K$ & 0.80 & 0.85 & 0.90& 0.95 & 1.00 & 1.05 & 1.10 & 1.15 & 1.20 &\\\hline
SINH & -1.5E-05 &	-1.2E-05 &	-9.6E-06 &	-8.0E-06 &	-6.7E-06 &
-7.9E-06 &	-9.3E-06 &	-1.1E-05 &	-1.3E-05 & 169.6\\
$V_H$ & 3.9E-06 &	-2.6E-06	& 1.4E-06	& 3.8E-06	& 1.3E-06 & 1.1E-06	
&-1.8E-06	& 1.6E-06 &	-3.1E-06 &\\
Flat iFT-BM & -7.2E-06 &	-8.6E-06 &	-8.6E-06 &	-7.3E-06 &	-5.5E-06
& -5.2E-06 &	-4.46E-06	&-3.6E-06	 &-3.2E-06 & 97.7\\
Flat iFT & 2.7E-07 &	-1.6E-06 &	-2.4E-06	& -2.2E-06 &	-1.5E-06
& -1.0E-06 &	-4.9E-08 &	1.0E-06 &	1.8E-06 & 661.5\\
Lewis 30 & 4.4E-06	& -1.9E-07 &	-1.4E-06 &	-1.3E-06 &	-8.3E-07
& -5.7E-07 &	-2.4E-07 &	8.5E-08 &	3.7E-07 & 400.8
\\\hline
\end{tabular}
}

\begin{flushleft}{\tiny 
SINH: $\om_1=-0.5$, $b=0.769884522$, $\om=0$, 	$\ze=0.285754315$, $N=12$, BL Modification with $M=317$.
\\
$V_H$: hybrid method of \cite{RoughNotTough} \\
Flat iFT-BM: $\sg_0=1$, $\om_1=-0.5$, 	$\ze=0.717626524$, $N=	16$, BL Modification with $M=317$.
 \\
Flat FT: $\om_1=-0.5	$,	$\ze=0.109637386$, $N=110$ , BL Modification with $M=317$.
\\
Lewis 30: Lewis method and Gauss-Legendre quadrature with 30 terms, BL Modification with $M=317$.
\\
CPU time is for the evaluation of $\Phi(\xi_k, \tau_m)$, for $k=0,1,\ldots, N$, $m=1,\ldots, 317$.\\
Flat iFT-BM  is used with the parallelization w.r.t. $\xi$.\\
SINH, Flat iFT and Lewis method are used without the parallelization w.r.t. $\xi$. \\
For the Lewis method, the nodes and weights are precalculated.\\
\vskip-0.2cm
For $V_H$, the CPU time is in the range 593-667 msec. per strike.  
}
\end{flushleft}

{\tiny
\begin{tabular}{c|rrrrrrrrr|r}
\hline\hline
$T=1$ & & & & & & & & & & Time \\\hline
SINH & -2.1E-05 &	-1.5E-05 &	-1.1E-05 &	-8.0E-06 &	-6.1E-06	&
-7.8E-06 &	-1.0E-05 &	-1.3E-05 &	-1.7E-05 & 295.6\\
$V_H$ & -1.8E-05 &	9.9E-06 &	-4.5E-06 &	7.2E-06 &	-1.9E-06 &
-6.7E-07 &	-1.1E-05 &	-1.0E-05	& -1.9E-05 &\\

Flat iFT-BM & 9.3E-06 &	3.5E-06 &	-2.1E-06 &	-3.3E-06 &	-1.6E-06 &
7.3E-07 &	4.1E-06 &	6.0E-06 &	3.6E-06 & 102.9\\

Flat iFT & 1.1E-05 &	3.6E-06 &	-3.0E-06	& -4.0E-06 &	-1.8E-06 &
1.1E-06 &	5.1E-06 &	7.2E-06 &	4.14E-06 & 980.8\\

Lewis 30 & 1.2E-04 &	3.4E-05 &	-1.706 &	-7.7E-06 & -5.1E-06 & 
-3.4E-06 &	-7.4E-07 &	2.3E-06 &	5.8E-06 & 144.1
\\\hline
\end{tabular}
}
\begin{flushleft}{\tiny 
SINH: $\om_1=-0.5$, $b=0.769884522$, $\om=0$, 	$\ze=0.285754315$, $N=14$, BL Modification with $M=399$.
\\
$V_H$: hybrid method of \cite{RoughNotTough} \\
Flat iFT-BM: $\sg_0=0.5$, $\om_1=-0.5$, 	$\ze=0.717626524$, $N=	22$, BL Modification with $M=317$. \\
Flat FT: $\om_1=-0.5	$,	$\ze=0.0877$, $N=200$, BL Modification with $M=317$.\\
Lewis 30: Lewis method and Gauss-Legendre quadrature with 30 terms, BL Modification with $M=317$.\\
For $V_H$, the CPU time is in the range 548-582 msec. per strike. 
 
}
\end{flushleft}

{\tiny
\begin{tabular}{c|rrrrrrrrr|r}
\hline\hline
$T=0.5$ & & & & & & & & & & Time \\\hline
SINH & 4.4E-05 &	5.1E-05 &	-7.9E-06	& -1.9E-05 &	-2.7E-06 &
-3.5E-06 &	-37E-06	& -5.1E-06 &	-8.6E-06 & 329.6\\

$V_H$ & 3.3E-05 &	-2.0E-05 &	2.8E-06 &	-5.2E-06 &	-8.6E-06 &
-1.7E-06 &	-1.9E-05 &	7.7E-06 &	-4.3E-05 & \\

Flat iFT-BM & 4.4E-05 &	5.1E-05 &	-7.9E-06 &	-1.9E-05 &	-2.7E-06 &
1.8E-05 &	2.6E-05 &	-1.2E-05 &	-9.4E-05 & 107.3\\

Flat iFT & -1.4E-05 &	1.0E-05 &	2.1E-06 &	-4.4E-06 &	-2.1E-06 &
2.5E-06 &	6.2E-06 &	-1.73E-06 &	-2.2E-05 & 1,192.3\\

Lewis 35 & 7.7E-04 &	1.2E-04 &	-3.2E-05 &	-1.9E-05 &	-2.9E-06 &
1.5E-06 &	2.2E-06 &	5.8E-06 &	2.9E-05 & 465.4
\\\hline
\end{tabular}
}
\begin{flushleft}{\tiny 
SINH: $\om_1=-0.5$, $b=0.769884522$, $\om=0$, 	$\ze=0.1836992027$, $N=23$, BL Modification with $M=317$.
\\
$V_H$: hybrid method of \cite{RoughNotTough} \\
Flat iFT-BM: $\sg_0=0.5$, $\om_1=-0.5$, 	$\ze=0.789389176$, $N=	30$, BL Modification with $M=317$. \\

Flat FT: $\om_1=-0.5	$,	$\ze=0.0877$, $N=200$, BL Modification with $M=317$.\\

Lewis 35: Lewis method and Gauss-Legendre quadrature with 35 terms, BL Modification with $M=317$. \\

For $V_H$, the CPU time is in the range 666-689 msec. per strike. 

}
\end{flushleft}

\label{table:rel_errors_moderate}

 \end{table}
 
  \begin{table}
\caption{\small Short maturities,  spot $S_0=1$. Relative errors (rounded) of calculations of OTM and ATM puts $(K\le 1$) and OTM calls $(K>1$)  in the rough Heston model with the parameters \eq{parEuRos} for several numerical schemes and CPU time (in msec., the average over 1000 runs). 
  }

{\tiny
\begin{tabular}{c|rrrrrrrrr|r}
\hline\hline
$T=1/12$ & & & & & & & & & & Time \\\hline
$K$ & 0.80 & 0.85 & 0.90& 0.95 & 1.00 & 1.05 & 1.10 & 1.15 & 1.20 & \\\hline
SINH & 1.4E-05 &	-2.1E-04 &	-1.2E-05	& -2.9E-07 &	-1.8E-06 &
-5.8E-07	& 5.8E-06	& -1.3E-04 & 2.9E-03 & 415.7\\

$V_H$ & -0.057 &	-0.0018 &	3.5E-04	& -8.7E-05 &	-3.3E-05 &
-3.3E-05 &	-1.35E-05	& -2.4E-04&	7.8E-03 & \\

Flat iFT-BM & -0.092 &	1.5E-03 &	1.2E-04 &	-5.0E-05 &	1.6E-05 &
-3.4E-05 &	-1.1E-04 &	5.2E-04	& -0.15 & 133.3\\

Flat iFT & -5.6 &	0.13 &	0.017 &	-0.0083& 0.0024 &
-9.4E-04 &	-0.047 &	0.42 &	2.7 & 2,341.2\\

Lewis 80 & 0.045 &	-0.0073 &	1.7E-04 &	-5.8E-06 & 6.9E-08 &	-3.8E-07	& 
4.0E-06 &	-3.3E-05 &	-0.013 & 1,062.3
\\\hline
\end{tabular}
}
\begin{flushleft}{\tiny 
SINH: $\om_1=-0.5$, $b=0.769884522$, $\om=0$, 	$\ze=0.1836992027$, $N=28$, BL Modification with $M=317$.
\\
$V_H$: hybrid method of \cite{RoughNotTough} \\
Flat iFT-BM: $\sg_0=0.5$, $\om_1=-0.5$, 	$\ze=0.717626524$, $N=	80$, BL Modification with $M=317$. \\

Flat FT: $\om_1=-0.5	$,	$\ze=0.0877$, $N=450$, BL Modification with $M=317$.\\
Lewis 80: Lewis method and Gauss-Legendre quadrature with 80 terms, BL Modification with $M=317$. \\
For $V_H$, the CPU time is in the range 410-423 msec. per strike.

}
\end{flushleft}

{\tiny
\begin{tabular}{c|rrrrrrr|r}
\hline\hline
$T=1/52$ & & & & & & & &  Time \\\hline
$K$ &  0.85 & 0.90& 0.95 & 1.00 & 1.05 & 1.10 & 1.15 &\\\hline
SINH & -0.42 &	1.5E-03 &	-1.6E-05 &	-6.6E-06 &
	-2.3E-04	& -0.043 & -205 & 154.8 \\
$V_H$ & 	(**) &	0.013 & 0.085 &	0.016 &
0.096 &	0.32 &	0.82&	 \\
Flat iFT-BM & 
26.5 &	2.8E-03 &	-1.1E-04 &	-9.4E-07 &
1.3E-04 & 	0.075 &	1,030& 339.3\\
Flat iFT & 1,167 &	0.71 &	3.9E-04 &	1.5E-04 &
	1.7E-03 &	-1.7 & -49,413 & 1,664.2\\
Lewis 100 & 25,177 &	1.2 &	 3.5E-04 &	4.3E-07 &
6.3E-05 &	0.60 &	-119,127 & 187.8
\\\hline
\end{tabular}
}
\begin{flushleft}{\tiny 
At $K=0.8$ and $K=1.2$, the prices of OTM options are smaller than $10^{-12}$, and the benchmark prices
cannot be calculated using double precision arithmetic.\\

SINH, puts: $\om_1=0.325762041$, $b=1.014615984$, $\om=0.2$, 	$\ze=0.145086905$, $N=38$, BL Modification with $M=100$\\
SINH, calls: $\om_1=-1.325762041$, $b=1.014615984$, $\om=-0.2$, 	$\ze=0.145086905$, $N=38$, BL Modification with $M=100$.
\\
$V_H$: hybrid method of \cite{RoughNotTough}; (**): the call price in \cite{RoughNotTough} implies that the price of the put is 0. \\
Flat iFT-BM: $\sg_0=0.5$, $\om_1=-0.5$, 	$\ze=0.717626524$, $N=	200$, BL Modification with $M=100$. \\

Flat FT: $\om_1=-0.5	$,	$\ze=0.07309159$, $N=1500$, BL Modification with $M=100$.\\
Lewis 100: Lewis method and Gauss-Legendre quadrature with 100 terms, BL Modification with $M=100$. \\
The order of the errors of Flat iFT-BM, Flat FT and Lewis 100 does not decrease if $N$ increases further. \\
\vskip-0.2cm
For $V_H$, the CPU time is in the range 125-164 msec. per strike. 
}
\end{flushleft}

{\tiny
\begin{tabular}{c|rrr|r}
\hline\hline
$T=1/252$  & & & &   Time \\\hline
$K$ &   0.95 & 1.00 & 1.05  &\\\hline
SINH & 	-2.7E-03 &	4.7E-07 & 9.0E-03 & 212.1\\
$V_H$ & 	11.2 &	1.7E-04 &
18.3 &		 \\
Flat iFT-BM & -0.51 &	1.E-04 & 18.8 & 557.8\\
Flat iFT & -17.8 &	3.1E-03 & -370 & 1,664.2\\
Lewis 100 & 6.3&	-2.2E-05 & 270 & 190.1 
\\\hline
\end{tabular}
}
\begin{flushleft}{\tiny 
At $K=0.80, 0.85, 0.90$ and $K=1.10, 1.15, 1.20$, the prices of OTM options are smaller than $10^{-12}$, and the benchmark prices
cannot be calculated accurately using double precision arithmetic.\\

SINH, puts: $\om_1=0.325762041$, $b=1.014615984$, $\om=0.2$, 	$\ze=0.145086905$, $N=46$, BL Modification with $M=100$.\\
SINH, calls: $\om_1=-1.325762041$, $b=1.014615984$, $\om=-0.2$, 	$\ze=0.145086905$, $N=46$, BL Modification with $M=100$.
\\
$V_H$: hybrid method of \cite{RoughNotTough}. \\
Flat iFT-BM: $\sg_0=0.5$, $\om_1=-0.5$, 	$\ze=0.717626524$, $N=	350$, BL Modification with $M=100$. \\

Flat FT: $\om_1=-0.5	$,	$\ze=0.07309159$, $N=1500$, BL Modification with $M=100$.\\
Lewis 100: Lewis method and Gauss-Legendre quadrature with $N=100$ terms, BL Modification with $M=100$. \\
The order of the errors of Flat iFT-BM, Flat FT and Lewis  does not decrease if $N$ increases further.\\
\vskip-0.2cm
For $V_H$, the CPU time is in the range 154-196 msec. per strike.  
}
\end{flushleft}

\label{table:rel_errors_short}
 \end{table}
 
 \begin{table}
\caption{\small Implied volatilities for options of short maturities in Table~\ref{table:rel_errors_short}.}

{\tiny
\begin{tabular}{c|rrrrrrrrr}
\hline\hline
$T=1/12$ & & & & & & & & & \\\hline
$K$ &  0.80 & 0.85 & 0.90& 0.95 & 1.00 & 1.05 & 1.10 & 1.15 & 1.20\\\hline
BB & 0.2280 &	0.2226 &	0.2173 &	0.2123 &	0.2075 &	0.2030 &	0.1986 &	0.1945 &	0.1907\\
SINH & 0.2280 &	0.2225 &	0.2173 &	0.2123 &	0.2075 &	0.2030 &	0.1986 &	0.1945 &	0.1907  \\
$V_H$ & 	0.2271 &	0.2225 &	0.2173 &	0.2123&	0.2075 &	0.2030 &	0.1986 &	0.1944 &	0.1907	 \\
Flat iFT-BM & 
0.2265 &	0.2226 &	0.2173 &	0.2123 &	0.2075 &	0.2030 &	0.1986 &	0.1947 &	0.1884\\
Flat iFT & (**) &	0.2257 &	0.2181 &	0.2116 &	0.2080 &	0.2030 &	0.1968 &	0.2029 &	0.2116 \\
Lewis 100 & 0.2243 &	0.2226 &	0.2173 &	0.2123 &	0.2075 &	0.2030 &	0.1986 &	0.1945 &	0.1911

\\\hline
\end{tabular}

{\tiny
\begin{tabular}{c|rrrrrrrrr}
\hline\hline
$T=1/52$ & & & & & & &  & & \\\hline
$K$ & 0.8 & 0.85 & 0.90& 0.95 & 1.00 & 1.05 & 1.10 & 1.15  & 1.20\\\hline
BB &  0.2383 &	0.2288& 	0.2195&	0.2105 &0.2018
&0.1935 &	0.1857 & 0.1786 & 0.1737\\
SINH &  
0.2450 &	0.2288 &	0.2195 &	0.2105 &	0.2018&	0.1935 &	0.1857 &	0.1786&	0.1703\\
$V_H$ & (**) & 	(**) &	0.2197 &	0.2138 &	0.2051 &	0.1968 &	0.1889 &	0.1818 & 0.1843 \\
Flat iFT-BM & (*) & 0.2600 &	0.2196 &	0.2105 &	0.2018 &	0.1935 &	0.1866 &	0.2291 & 0.2929\\
Flat iFT& 0.4029 & 0.3147 &	0.2280 &	0.2107 &	0.2019 &	0.1936&(*) &	(*)  & 0.3071\\
Lewis 100& 0.5883 & 0.3928&	0.2321 &	0.2106 &	0.2018&0.1935 &	0.1913 &	(*)  & (*)
\\\hline
\end{tabular}
}

{\tiny
\begin{tabular}{c|rrr}
\hline\hline
$T=1/252$  & & &     \\\hline
$K$ &   0.95 & 1.00 & 1.05  \\\hline
BB & 0.2154 &	0.1994 &	0.1841\\
SINH & 	0.2154 &	0.1994 &	0.1841 \\
$V_H$ & 	0.2552 &	0.1994 &	0.2174 \\
Flat iFT-BM & 0.2068 &	0.1994 &	0.2178 \\
Flat iFT & (*)	&0.2000 &	(*) \\
Lewis 100 & 0.2456 &	0.1994 &	0.2661  
\\\hline
\end{tabular}
}

\begin{flushleft}{\tiny BB: benchmark.\\
(*): the price outside the no-arbitrage bounds.\\
(**): the put price is smaller than $10^{-12}$.\\
}
\end{flushleft}
}
\label{table:implvol_short}

\end{table}

\begin{table}
\caption{\small ``Bad region in the parameters space". Prices of OTM and ATM put and OTM call options of short maturity $T=1/365$, $r=0.1$, spot $S_0=1000$,  in the KoBoL model of small order, with parameters $(\mu, c, \nu, \lp, \lm)=(0.1,	1,	0.5,	0.2,	-1.2)$, and  relative errors of SINH-, Gauss-Laguerre (GL)
and Gauss-Kronrod  (GK) quadratures w.r.t. $V$. $N$ is the number of terms.
 }
{\tiny \begin{tabular}{c|ccccc| c}
\hline\hline
$K$ & 	0.6	& 0.8 &	1	&1.2 &	1.4	& N\\\hline
$V$ &1.15596308274723 &	2.2769702099306 &	5.97763601818645 &	3.06782338691266 &
2.39427586598299 &  105-274-114 \\\hline
 $SINH $ & 
 -7.22E-06 &	-4.33E-06	&-1.90E-06 &	-3.96E-06	& -5.38E-06
  &33-94-36 \\
 \hline
 $GL$ & 1.86E-04 &	-5.05E-03 &	0.036	& 4.12E-03 &	9.53E-04& 175
  \\\hline
 $GK $ 
 & 1.45E-04 &	1.25E-05	& -2.14E-05 &	1.59E-05 &	-2.70E-05 & \\\hline
\end{tabular}
}
\begin{flushleft}{\tiny Relative errors of the benchmark prices are smaller than $E-13$, and defined as differences
of prices calculated with $\om=\pm \pi/4$ and $\om=\pm \pi/8$, following the general prescription
with the error tolerance $\eps=E-15$ and dividing (resp., multiplying) $\ze$ and $\La$ by 1.4.\\
$SINH$ - calculated for $\om=\pm \pi/4$, following the general prescription for $\eps=E-07$}
\end{flushleft}
\label{table:relerrKBLBad}
\end{table}
\vskip-1.5cm
\begin{table}
\caption{\small ``Good region in the parameters space". Prices of OTM and ATM put and OTM call options of short maturity $T=1$, $r=0.1$, spot $S_0=1000$,  in the KoBoL model of large order, with parameters $(\mu, c, \nu, \lp, \lm)=(0.1,	1,	0.5,	0.2,	-1.2)$, and  relative errors of SINH-, Gauss-Laguerre (GL)
and Gauss-Kronrod  (GK) quadratures w.r.t. $V$. $N$ is the number of terms.
 }
{\tiny \begin{tabular}{c|ccccc| c}
\hline\hline
$K$ & 	0.6	& 0.8 &	1	&1.2 &	1.4	& N\\\hline
$V$ &326.631884432011 &	469.61080845886 &	618.758920614544 &	686.548810890079 &
662.549963116958 &  53-54 \\\hline
 $SINH $ & 
 0 &	0	&0 &	0	& -1.55E-12
  &38-39 \\
 \hline
 $GL$ & 2.86E-11	& 2.22E-11 &	1.86E-11 &	1.84E-11 &	2.07E-11& 175
  \\\hline
 $GK $ 
 & 1.15E-14 &	-3.51E-15	& -6.61E-15	& -6.62E-15 &	-5.15E-15 & \\\hline
\end{tabular}
}
\begin{flushleft}{\tiny Relative errors of the benchmark prices are 0 (calculated in Matlab with double precision arithmetic), and defined as differences
of prices calculated with $\om=\pm \pi/4$ and $\om=\pm \pi/8$, following the general prescription
with the error tolerance $\eps=E-15$ and dividing (resp., multiplying) $\ze$ and $\La$ by 1.4.\\
$SINH$ - calculated for $\om=\pm \pi/4$, following the general prescription for $\eps=E-15$}
\end{flushleft}
\label{table:relerrKBLGood}
\end{table}

\newpage

  \section{Pricing algorithm for the Markovian approximation (BL2)}\label{a:BL2}
This appendix outlines the BL2 algorithm described 
in \cite{MarkovianGG}. The procedure in Appendix F of \cite{MarkovianGG} is used to calculate the weights and nodes of the Markovian
approximation. The full description of the pricing algorithm  is not provided in  \cite{MarkovianGG} after the node construction, but can be found in the Python code published by the authors on GitHub \cite{breneis2025}, specifically in the function \texttt{compute\_Fourier\_inversion}, in file \texttt{rHestonFourier.py}. 
This algorithm adaptively refines the parameters of the Riccati solver and the Fourier inversion routine in order to satisfy a prescribed error tolerance. In particular, it adjusts the number of time steps in the Riccati ODE solver ($M$), the truncation point of the Fourier integral ($L$), and the number of quadrature points used in the Fourier inversion ($N$). For clarity of presentation, the algorithm is written in plain text rather than in pseudocode.
\mbr 
\begin{enumerate}[1.]
	\item \textbf{Inputs:}
	\begin{itemize}
		\item Maturity $T$
		\item An array of strikes
		\item Relative error tolerance $\varepsilon$
		\item Hurst parameter $H > 0$
		\item A pricing routine \texttt{compute}$(M, L, N)$, that returns a vector of prices or implied volatilities, denoted as $\sigma$. This always uses flat iFT with, in our notation, $\omega=2$ for put options and		$\omega=-2$ for calls.
	\end{itemize}
	
	\item \textbf{Initial parameter guess:} The numerical parameters are initialized based on the maturity and Hurst parameter.
	\[
	L = 100 \cdot T^{-0.5 + H}, \quad M = \mathrm{int}(10 \cdot L), \quad N = \mathrm{int}(8 \cdot L)
	\]
	
	\item \textbf{Baseline calculation:} A first solution is computed.
	\[
	\sigma^{(0)} \leftarrow \mathrm{compute}(M, L, N)
	\]
	
	\item \textbf{Initial error estimation:} The error is estimated by comparing the baseline solution to a  coarser one. This error is used to determine if the adaptive loop is necessary.
	\[
	\sigma^{\mathrm{coarse}} \leftarrow \mathrm{compute}\left(\mathrm{int}(M/1.6), L/1.2, \mathrm{int}(N/2)\right)
	\]
	\[
	\mathrm{error} = \max_k \frac{|\sigma^{\mathrm{coarse}}_k - \sigma^{(0)}_k|}{|\sigma^{(0)}_k|}
	\]
	
	\item \textbf{Adaptive refinement loop:} The loop continues as long as the error is above the tolerance\footnote{For the calculations in Table \ref{tab:markov-times}, we used $\varepsilon=10^{-3}$.} $\varepsilon$ or if the solution $\sigma^{(k)}$ contains invalid numbers (NaN). Let $k=0$.
	\sbr 
	\textbf{While} ($\mathrm{error} > \varepsilon$ \textbf{or} ($\sigma^{(k)}$) contains NaN):
	\sbr
	\begin{enumerate}
		\item[a.] \textbf{Store current solution:} $\sigma^{\mathrm{old}} \leftarrow \sigma^{(k)}$.
		
		\item[b.] \textbf{Check for NaN values in the solution:}
		\sbr
		\begin{itemize}
			\item \textbf{Case 1: No NaN values in $\sigma^{(k)}$}
			\begin{enumerate}
				\item[i.] \textbf{Shrink Test for Riccati time steps ($M$):} Check if $M$ can be reduced.
				\[
				\sigma^{\mathrm{test}} \leftarrow \mathrm{compute}(\mathrm{int}(M/1.8), L, N)
				\]
				\[
				\mathrm{error}_{\mathrm{R}} = \max_k \frac{|\sigma^{\mathrm{test}}_k - \sigma^{(k)}_k|}{|\sigma^{(k)}_k|},
				\]
				where $\mathrm{int}(\cdot)$ rounds to the nearest integer. If $\mathrm{error}_{\mathrm{R}} < \varepsilon/5$, then update $M \leftarrow \mathrm{int}(M/1.6)$.
				
				\item[ii.] \textbf{Shrink Test for Fourier quadrature points ($N$):} Check if $N$ can be reduced.
				\[
				\sigma^{\mathrm{test}} \leftarrow \mathrm{compute}(M, L, \mathrm{int}(N/2))
				\]
				\[
				\mathrm{error}_{\mathrm{F}} = \max_k \frac{|\sigma^{\mathrm{test}}_k - \sigma^{(k)}_k|}{|\sigma^{(k)}_k|}
				\]
				If $\mathrm{error}_{\mathrm{F}} < \varepsilon/5$, then update $N \leftarrow \mathrm{int}(N/1.8)$.
				
				\item[iii.] \textbf{Refine parameters:} Increase parameters based on the component-wise errors.
				\[
				L \leftarrow 1.4 \cdot L
				\]
				\[
				M \leftarrow
				\begin{cases}
					\mathrm{int}(1.4 \cdot M), & \text{if } \mathrm{error}_{\mathrm{R}} < \varepsilon/2 \\
					2 \cdot M, & \text{otherwise}
				\end{cases}
				\]
				\[
				N \leftarrow
				\begin{cases}
					\mathrm{int}(1.4 \cdot N), & \text{if } \mathrm{error}_{\mathrm{F}} < \varepsilon/2 \\
					2 \cdot N, & \text{otherwise}
				\end{cases}
				\]
				
			\end{enumerate}
			
			\sbr
			\item \textbf{Case 2: $\sigma^{(k)}$ contains NaN values}
			\sbr
			\begin{enumerate}
				\item[i.] Increase parameters as follows:
				\[
				L \leftarrow 1.6 \cdot L, \quad N \leftarrow \mathrm{int}(1.7 \cdot N), \quad M \leftarrow \mathrm{int}(2.5 \cdot M)
				\]
			\end{enumerate}
		\end{itemize}
		
		\item[c.] \textbf{Recompute Solution:} Calculate the new solution with the updated parameters
		\[
		\sigma^{(k+1)} \leftarrow \mathrm{compute}(M, L, N)
		\]
		
		\item[d.] \textbf{Update Error:} Calculate the new error by comparing the new and old solutions:
		\[
		\mathrm{error} \leftarrow \max_k \frac{|\sigma^{\mathrm{old}}_k - \sigma^{(k+1)}_k|}{|\sigma^{(k+1)}_k|}
		\]
		
		\item[e.] \textbf{Increment:} $k \leftarrow k+1$.
	\end{enumerate}
	\sbr
	\item \textbf{Return:} The final converged solution $\sigma^{(k)}$ and the final error estimate.
\end{enumerate}

\clearpage

\end{document}